\documentclass[letterpaper,draftclsnofoot,onecolumn]{IEEEtran}
\usepackage{amsmath,amsfonts,amsthm}
\usepackage{algorithmic}
\usepackage{algorithm}
\usepackage{array}
\usepackage[caption=false,font=normalsize,labelfont=sf,textfont=sf]{subfig}
\usepackage{textcomp}
\usepackage{stfloats}
\usepackage{url}
\usepackage{tikz}
\usepackage{verbatim}
\usepackage{graphicx}
\usepackage{cite}
\usepackage{multirow}
\usepackage{tabularx}
\usepackage{booktabs}
\usepackage{longtable}
\usepackage{arydshln}
\usepackage{MnSymbol}
\usepackage{hyperref}

\newtheorem{definition}{Definition}
\newtheorem{theorem}{Theorem}
\newtheorem{corollary}{Corollary}
\newtheorem{lemma}{Lemma}
\newtheorem{proposition}{Proposition}
\newtheorem{remark}{Remark}
\newtheorem{example}{Example}

\DeclareMathAlphabet\mathbfcal{OMS}{cmsy}{b}{n}

\allowdisplaybreaks[3]

\begin{document}

\title{Performance Analysis and Optimal Design of ORB-Type GRAND Algorithms}

\author{Li Wan, {\it Student Member, IEEE} and Wenyi Zhang, {\it Senior Member, IEEE}
\thanks{The authors are with Department of Electronic Engineering and Information Science, University of Science and Technology of China, Hefei, China, 230027.}
}



\maketitle

\begin{abstract}
    Guessing Random Additive Noise Decoding (GRAND) performs decoding by sequentially guessing channel error patterns (EPs). Ordered Reliability Bits GRAND (ORBGRAND) is a notable instance suitable for efficient implementation, as it schedules EPs solely according to the ranking of soft channel outputs.
    In this paper, we generalize this principle to a broader class of GRAND algorithms whose testing order depends only on reliability ranking, referred to as ORB-type GRAND. We develop a unified analytical framework based on a key quantity termed the average guessing posterior (AGP), which captures the effectiveness of each EP and reduces decoding into an ordering problem over the EP space.
    For random code ensembles, we derive exact expressions for the block error rate (BLER), stopping-time distribution, and average number of tests under a fixed test budget. The analysis separates target-miss and target-preemption errors and shows that ordering EPs by non-increasing AGP is optimal over the EP set under consideration.
    For fixed linear block codes, we derive the BLER expression that isolates the code-dependent target-preemption term and characterize this term through higher-order weight relationships of codeword tuples, with a computable first-order upper bound as a useful special case.
    Guided by these insights, we formulate ReShuffled-ORBGRAND (RS-ORBGRAND) as an offline AGP-based reshuffling scheme. Numerical results for the Bose--Chaudhuri--Hocquenghem (BCH)$(127,113)$ code show that RS-ORBGRAND consistently improves existing ORB-type GRAND algorithms and lies within $0.1$~dB of a maximum-likelihood decoding lower-bound benchmark at a BLER of $10^{-6}$.
\end{abstract}

\begin{IEEEkeywords}
    Average guessing posterior, block error rate, error-pattern ordering, GRAND, linear block codes, maximum likelihood decoding, ORBGRAND.
\end{IEEEkeywords}


\section{Introduction}\label{Sec:introduction}

Shannon's pioneering work demonstrated that increasing the code length allows for constructions with progressively stronger error-correction capability \cite{shannon1948}. However, since maximum likelihood decoding (MLD) of linear block codes is NP-hard, the associated computational burden typically grows rapidly with block length \cite{berlekamp1978}. 
As wireless communications, optical communications, and high-speed serializer/deserializer (SerDes) links continue to evolve toward higher reliability, lower latency, and higher throughput \cite{you2021towards,winzer2018fiber,mickevicius2021serdes}, efficient short-blocklength channel decoding has become a common key problem across modern communication and information-processing systems \cite{shirvanimoghaddam2018short, yue2023efficient}. A central challenge is to design decoding strategies that approach MLD performance while maintaining manageable decoding complexity.

Guessing Random Additive Noise Decoding (GRAND) has recently emerged as a universal decoding framework \cite{duffy2019GRAND}. Instead of searching over the codeword space, GRAND performs decoding in the noise domain by sequentially generating candidate error patterns (EPs) and subtracting them from the received sequence until a valid codeword is identified. With an appropriate ordering of EPs, GRAND can achieve MLD for both hard-output channels such as the binary symmetric channel (BSC) and binary erasure channel (BEC) \cite{duffy2018guessing, duffy2019GRAND} and soft-output channels such as the additive white Gaussian noise (AWGN) channel \cite{solomon2020SGRAND}. Since the decoding procedure depends only on channel noise statistics, GRAND can operate with arbitrary linear block codes without code-specific decoder design.





To adapt the GRAND framework to different channels and system requirements, a number of variants have been proposed \cite{solomon2020SGRAND, duffy2021SRGRAND, abbas2022grand, duffy2023using, sarieddeen2022grand, feng2024laplacian}, along with several methods aimed at improving decoding performance and reducing complexity \cite{condo2021high, yuan2023guessing}. 
For soft-output channels, Soft GRAND (SGRAND) generates EPs according to their likelihood using the exact magnitudes of log-likelihood ratios (LLRs), and is equivalent to MLD when the search over EPs is exhaustive \cite{solomon2020SGRAND, Lin2004ErrorCC}. However, generating EPs according to precise LLR magnitudes requires dynamic scheduling of candidate EPs, leading to substantial computational overhead and complicating hardware implementations \cite{wan2025parallelism}. In contrast, Ordered Reliability Bits GRAND (ORBGRAND) schedules EPs using only the ranking of LLR magnitudes \cite{duffy2022ORBGRAND}. Once the reliability ranking is obtained, EPs can be generated according to a fixed structure independent of the exact LLR magnitudes \cite{an2023soft, galligan2023block, abbas2022high, abbas2023guessing, ji2024efficient}.


ORBGRAND provides an attractive trade-off between decoding performance and implementation complexity by generating EPs based solely on the reliability ranking of LLR magnitudes rather than their exact values. This ranking-based structure enables efficient parallel implementations while maintaining competitive decoding performance. Motivated by this property, several works have proposed improved EP ordering strategies that preserve the reliance on reliability ranking, including heuristic ordering rules \cite{duffy2022ORBGRAND, condo2021high, wang2023improved}, empirical rearrangement approaches \cite{condo2022fixed}, and search-based optimization methods \cite{wan2024approaching}. We collectively refer to this class of decoding strategies as ORB-type GRAND, whose formal definition will be given later.

From an information-theoretic perspective, ORBGRAND and certain variants have been shown to be capacity-achieving or nearly capacity-achieving~\cite{liu2022orbgrand,li2024orbgrand,li2025orbgrand}, and finite-blocklength analysis~\cite{li2026finite} further provides achievable-rate expansions and dispersion-based normal approximations. These results characterize reliable transmission through rate-oriented asymptotic analyses, but they do not directly quantify operational list-size-one block error rate (BLER) under a fixed EP test budget. At such finite block lengths, ORB-type GRAND nonetheless exhibits a noticeable BLER gap from MLD, particularly at high signal-to-noise ratios (SNRs).

These observations motivate a closer examination of EP ordering in ORB-type GRAND under finite test budgets. The ordered EP sequence determines which EPs are tested before the decoder terminates, so the resulting BLER is governed jointly by the testing order and the codebook. Existing ordering rules are mostly heuristic, empirical, or based on idealized search criteria, while a systematic BLER analysis from the perspective of EP ordering is still lacking.

Existing analytical results mainly focus on list-GRAND decoding and its extensions to soft-output decoding \cite{abbas2022list,galligan2023upgrade,yuan2024softoutput}. These approaches often approximate the occurrences of competing codewords along the guessing order by geometric distributions, yielding tractable performance predictions that are empirically accurate for random code ensembles. However, for structured linear block codes in the practically relevant list-size-one case---i.e., standard GRAND that terminates upon identifying the first valid codeword---competing-codeword events are strongly code-dependent and are no longer well captured by such geometric approximations; consequently, these predictions become unreliable~\cite{duffy2024soft}. This limitation is particularly relevant to ORB-type GRAND, whose performance is governed by a fixed EP ordering.

For the EP-ordering aspect specifically, ReShuffled-ORBGRAND (RS-ORBGRAND) was proposed in \cite{wan2024approaching}. That work studied an idealized search problem in which decoding succeeds once the target EP is queried, thereby isolating the effect of EP ordering on locating the target while excluding target preemption. Thus, although it provides useful insight into ordering candidate EPs, it does not explain how code structure affects target preemption. A systematic BLER analysis of ORB-type GRAND under finite test budgets is therefore still needed.


This paper develops an average guessing posterior (AGP)-based theoretical framework for analyzing ORB-type GRAND and establishing EP-ordering principles. The main contributions are summarized as follows.

\begin{itemize}
    \item We formalize ORB-type GRAND as the class of GRAND algorithms whose testing order depends only on reliability ranking. We then introduce the AGP, which quantifies EP effectiveness, is independent of any particular LLR realization, and characterizes ORB-type decoding behavior.
    \item For random code ensembles, we derive exact expressions for the BLER, decoding success probability, stopping-time distribution, and average number of tests under a fixed test budget. The analysis separates target-miss and target-preemption errors, and proves that ordering EPs by non-increasing AGP simultaneously minimizes the average BLER and the average number of tests over the EP set under consideration.
    \item We extend the BLER analysis to fixed linear block codes. The resulting exact expression isolates the code-dependent target-preemption term, which is characterized through higher-order weight relationships of codeword tuples. It also supports efficient finite-order low-BLER evaluation and yields a simple first-order upper bound.
    \item Guided by these insights, we formulate RS-ORBGRAND as an AGP-driven offline scheme that preserves the ranking-based ORB-type structure. Compared with the idealized search setting in~\cite{wan2024approaching}, the present design is supported by a BLER analysis that accounts for competing codewords; simulations for Bose--Chaudhuri--Hocquenghem (BCH)$(127,113)$ show consistent gains and performance within $0.1$~dB of the MLD lower-bound benchmark down to BLERs of $10^{-6}$.
\end{itemize}

The remaining part of this paper is organized as follows. 
Section~\ref{Sec:preliminaries} introduces the channel model and formally defines ORB-type GRAND. 
Section~\ref{Sec:random_code} analyzes ORB-type GRAND under random code ensembles and derives the decoding success probability, the average number of tests, and the AGP-based optimal ordering rule. 
Section~\ref{Sec:linear_code} extends the analysis to fixed linear block codes, establishes the BLER expression, and studies AGP-ordered EP sequences in this fixed-code setting. 
Section~\ref{Sec:simulation} validates the analytical results and evaluates the proposed RS-ORBGRAND. 
Finally, Section~\ref{Sec:conclusion} concludes the paper.

\section{ORB-Type GRAND Algorithms}\label{Sec:preliminaries}

In this section, we establish the analytical framework for the rest of the paper. We begin with the channel model and the posterior quantities associated with EPs, then introduce the GRAND decoding principle and the formal definition of ORB-type GRAND, and finally present RS-ORBGRAND as a practical scheme motivated by this framework. These preliminaries will serve as the foundation for our analysis developed in subsequent sections.

\subsection{Channel Model}\label{SubSec:basic_channel_model}

We use uppercase letters (e.g., $W$) to represent random variables and their corresponding lowercase letters (e.g., $w$) to represent their realizations. We append an underscore to a letter to represent a vector (e.g., $\underline{w}$), whose length is equal to the code length unless otherwise noted. For a sequence of vectors, we use parenthesis, such as $\underline{w}(i)$, to represent the ordinal, and for a vector $\underline{w}$, we use $w_i$ to represent its $i$-th component. For a continuous random variable, we use lowercase $p$ to represent its probability density function (e.g., $p_{Y}(y)$); for a discrete random variable, we use $P$ to represent its probability mass function, which characterizes the probability distribution that the random variable satisfies (e.g., $P_{W|Y}(w \mid y)$); for a random event, we use $\Pr(\cdot)$ to represent its probability of occurrence. Depending on the context, we may also use certain specific symbols for some frequently used events or probabilities.

We consider a general block code of rate $R = K/N$, which consists of $M = 2^K$ different binary codewords $\{\underline{w}(m)\}_{m = 1, 2, \ldots, M}$, each of a length $N$, thus forming a codebook $\mathcal{C} = \{\underline{w}(1),\ldots,\underline{w}(M)\}$. We let $\underline{W}$ represent a codeword drawn uniformly at random from the codebook. At the receiving end, the channel output vector is $\underline{Y}$, which is assumed to possess a probability density function over $\mathbb{R}^N$ for concreteness. We adopt a memoryless channel model, and let the output probability density distribution be $q_0(y)$ under input $w = 0$ and $q_1(y)$ under input $w = 1$, respectively. We call a channel output-symmetric if
\begin{eqnarray}\label{eqn:output-symmetric-channel}
    q_0(y) = q_1(-y)
\end{eqnarray}
holds for any $y \in \mathbb{R}$; for example, the commonly considered AWGN channel is output-symmetric. 

After receiving $\underline{Y}$, the LLR vector $\underline{L}$ is determined as follows:
\begin{equation}\label{eq:def_ell}
    L_i = {\log\frac{p_{Y|W}(Y_i \mid W_i = 0)}{p_{Y|W}(Y_i \mid W_i = 1)}} = \log\frac{q_0(Y_i)}{q_1(Y_i)},
\end{equation}
for $i = 1, \ldots, N$, and subsequently we use $\ell_{i}$ to represent the realization of $L_i$. Based on $\underline{L}$, the hard decision vector, $\theta(\underline{Y})$, is given by $\theta(Y_i) = (1 - \operatorname{sgn}(L_i))/2$. Equivalently, $\theta(Y_i)=\mathbf{1}(L_i < 0)$, i.e., $\theta(Y_i)=0$ if $L_i\ge 0$ and $1$ otherwise.

For an output-symmetric channel, we have the following useful properties:
\begin{itemize}
    \item If $P_W(W=0) = P_W(W=1) = 1/2$, then $p_Y(y) = p_Y(-y)$, because 
    \begin{equation}\label{eqn:prop_symm_1}
        p_Y(y) = \frac{1}{2}(p_{Y|W}(y \mid 0) + p_{Y|W}(y \mid 1)) \overset{\eqref{eqn:output-symmetric-channel}}{=} \frac{1}{2}(p_{Y|W}(-y \mid 1) + p_{Y|W}(-y \mid 0)) = p_Y(-y). 
    \end{equation}
    \item If the LLR of $y$ calculated by \eqref{eq:def_ell} is $\ell$, then the LLR corresponding to $-y$ is $-\ell$, because
    \begin{equation}\label{eqn:prop_symm_2}
        \ell = \log\frac{p_{Y|W}(y \mid 0)}{p_{Y|W}(y \mid 1)} \Rightarrow \log\frac{p_{Y|W}(-y \mid 0)}{p_{Y|W}(-y \mid 1)} = \log\frac{p_{Y|W}(y \mid 1)}{p_{Y|W}(y \mid 0)} = -\log\frac{p_{Y|W}(y \mid 0)}{p_{Y|W}(y \mid 1)} = -\ell. 
    \end{equation}
\end{itemize}

\subsection{Average Guessing Posterior}\label{SubSec:average_guessing_posterior}

Throughout the paper, we assume that for each $i = 1, \ldots, N$, $W_i$ is uniform over $\mathbb{F}_2$. The codebook ensembles studied in the paper satisfy this assumption.

To study EP ordering in noise-guessing decoding, we introduce a key metric that quantifies the effectiveness of an EP. For a received vector $\underline{y}$ and a candidate EP $\underline{e}$ which is a vector in $\mathbb{F}_2^N$, the vector $\theta(\underline{y})\oplus \underline{e}$ represents the codeword candidate obtained by flipping the hard decision according to $\underline{e}$. A natural quantity associated with $\underline{e}$ is therefore the posterior probability that this codeword candidate is indeed the transmitted codeword, referred to as the guessing posterior. Its expectation over the channel output, termed the average guessing posterior (AGP), will serve as the central quantity for characterizing EP ordering and decoding performance in the sequel.

It is clear that the hard decision vector may not coincide with the sent codeword; that is, $\theta(Y_i) = W_i$ may not hold for some indices in $\{1, 2, \ldots, N\}$. Regarding this fact, we have the following lemma.

\begin{lemma}\label{lem:y-neq-w-cond-ell}
    Conditioned on $Y = y$, the probability that $\theta(Y)$ and $W$ are different is given by
    \begin{equation}\label{eq:p_error}
    \Pr(\theta(Y) \neq W \mid Y = y) = \frac{1}{1 + \exp(|\ell|)},
\end{equation}
in which $\ell$ follows the definition in \eqref{eq:def_ell}.
\end{lemma}
\begin{proof}
    This lemma follows from Bayes' rule and the definition of $L$ in \eqref{eq:def_ell}. First, consider the case where $\theta(y) = 0$. We have $\Pr(\theta(Y) \neq W \mid Y = y) = \Pr(W = 1 \mid Y = y)$ and 
    \begin{eqnarray}
        \Pr(W = 1 \mid Y = y) &=& \frac{\Pr(Y = y \mid W = 1) \Pr(W = 1)}{\Pr(Y = y \mid W = 0) \Pr(W = 0) + \Pr(Y = y \mid W = 1) \Pr(W = 1)}\\
        &\overset{(a)}{ = }& \frac{q_1(y)}{q_0(y) + q_1(y)}  \overset{(b)}{ = } \frac{1}{1 + \exp(|\ell|)},
    \end{eqnarray}
    where (a) follows from the assumption that $W$ is uniform, and (b) holds because $\theta(y) = 0$ implies $\ell \geq 0$. Similarly, for the case where $\theta(y) = 1$, we also have
    \begin{eqnarray}
        \Pr(W = 0 \mid Y = y) = \frac{1}{1 + \exp(|\ell|)}.
    \end{eqnarray}
\end{proof}

When studying GRAND, we will frequently consider testing whether the modulo-two sum of the hard decision vector $\theta(\underline{Y})$ and a binary vector $\underline{e}$ is a codeword, where $\underline{e}$ is usually interpreted as an EP. For any given $\underline{y}$, we call $P_{\underline{W}|\underline{Y}}(\theta(\underline{y})\oplus \underline{e} \mid \underline{y})$ the guessing posterior, which will be seen to play a pivotal role in decoding.

We first consider an idealized setting in which the components of $\underline{W}$ are i.i.d. In this case, the guessing posterior admits a simple product form that reveals its basic channel-dependent structure.

\begin{lemma}\label{lemma:posterior_probability}
    When $\underline{W} = [{W}_1, {W}_2, \ldots, {W}_N]$ are i.i.d., given $\underline{Y} = \underline{y}$, the guessing posterior $P_{\underline{W}|\underline{Y}}(\theta(\underline{y})\oplus \underline{e} \mid \underline{y})$ is given by:
    \begin{equation}\label{eq:posterior}
        P_{\underline{W}|\underline{Y}}(\theta(\underline{y})\oplus \underline{e} \mid \underline{y}) =  \prod_{i:e_i = 1}\frac{1}{1+\exp(|\ell_i|)} \prod_{i:e_i = 0}\frac{\exp(|\ell_i|)}{1+\exp(|\ell_i|)}.
    \end{equation}
\end{lemma}
\begin{proof}
    Using the i.i.d. property of $\underline{W}$ and the 
    memoryless property of the channel, Bayes' rule gives
    \begin{eqnarray}
        P_{\underline{W}|\underline{Y}}\left(\theta(\underline{y}) \oplus \underline{e} \mid \underline{y}\right) = \prod_{i = 1}^N P_{{W}|Y} \left(\theta(y_i) \oplus e_i \mid y_i\right).
    \end{eqnarray}
Inspecting this product, the factors with $e_i = 1$ correspond to flipped hard decisions and, by Lemma~\ref{lem:y-neq-w-cond-ell}, contribute $\prod_{i:e_i = 1} \frac{1}{1+\exp(|\ell_i|)}$. The factors with $e_i = 0$ correspond to unchanged hard decisions and contribute $\prod_{i:e_i = 0} \frac{\exp(|\ell_i|)}{1+\exp(|\ell_i|)}$. Multiplying these two parts together leads to \eqref{eq:posterior}.
\end{proof}

For an actual codebook $\mathcal{C}$, however, the i.i.d. condition in Lemma~\ref{lemma:posterior_probability} no longer holds. The following result shows that the guessing posterior still admits a similar factorized form, up to a codebook constraint and a normalization coefficient.

\begin{lemma}\label{lemma:posterior_probability_codebook}
    For a given codebook $\mathcal{C}$, given $\underline{Y} = \underline{y}$, the guessing posterior $P_{\underline{{W}}|\underline{Y}}(\theta(\underline{y})\oplus \underline{e} \mid \underline{y})$ is given by:
    \begin{equation}\label{eq:posterior_2}
        P_{\underline{W}|\underline{Y}}\left(\theta(\underline{y})\oplus \underline{e} \mid \underline{y}\right) = \left\{\begin{array}{ll}
        C(\underline{y})\prod_{i:e_i = 1}\frac{1}{1+\exp(|\ell_i|)}\prod_{i:e_i = 0}\frac{\exp(|\ell_i|)}{1+\exp(|\ell_i|)}      & \text{if } \theta(\underline{y})\oplus \underline{e} \in \mathcal{C} \\
        0     & \text{else}
        \end{array} \right. ,
\end{equation}
where $C(\underline{y})$ is a normalization coefficient given by $C(\underline{y}) = 2^{N - K} \cdot  \prod_{i = 1}^N p_Y(y_i)/p_{\underline{Y}}(\underline{y})$. 
\end{lemma}
\begin{proof}
    We deduce via the following steps:
    \begin{align}
    P_{\underline{W}|\underline{Y}}\left(\theta(\underline{y})\oplus \underline{e} \mid \underline{y}\right) = & \frac{P_{\underline{W}}(\theta(\underline{y})\oplus \underline{e}) p_{\underline{Y}|\underline{W}}\left(\underline{y} \mid \theta(\underline{y})\oplus \underline{e}\right)}{ p_{\underline{Y}}(\underline{y})} \label{eqn:posterior_C_1}\\
    = &  \left\{\begin{array}{ll}
        C'(\underline{y})\cdot p_{\underline{Y}|\underline{W}}\left(\underline{y} \mid \theta(\underline{y})\oplus \underline{e}\right)   & \text{if } \theta(\underline{y})\oplus \underline{e} \in \mathcal{C} \\
        0     & \text{else}
        \end{array} \right. \label{eqn:posterior_C_2}\\
    = & \left\{\begin{array}{ll}
        C'(\underline{y})\cdot \prod_{i=1}^{N} p_{Y|W}\left(y_i \mid \theta(y_i)\oplus e_i\right)   & \text{if } \theta(\underline{y})\oplus \underline{e} \in \mathcal{C} \\
        0     & \text{else}
        \end{array} \right. \label{eqn:posterior_C_3}\\
    = & \left\{\begin{array}{ll}
        C(\underline{y})\cdot \prod_{i=1}^{N} P_{W|Y}\left(\theta(y_i)\oplus e_i \mid y_i\right)   & \text{if } \theta(\underline{y})\oplus \underline{e} \in \mathcal{C} \\
        0     & \text{else}
        \end{array} \right. \label{eqn:posterior_C_4}\\
    = & \left\{\begin{array}{ll}
        C(\underline{y})\cdot \prod_{i:e_i = 1}\frac{1}{1+\exp(|\ell_i|)}\prod_{i:e_i = 0}\frac{\exp(|\ell_i|)}{1+\exp(|\ell_i|)}.   & \text{if } \theta(\underline{y})\oplus \underline{e} \in \mathcal{C} \\
        0     & \text{else}
        \end{array} \right. ,\label{eqn:posterior_C_5}
\end{align}
where:
\begin{itemize}
    \item \eqref{eqn:posterior_C_1} is Bayes' rule;
    \item \eqref{eqn:posterior_C_2} holds, because if $\theta(\underline{y})\oplus \underline{e} \in \mathcal{C}$ then $P_{\underline{W}}(\theta(\underline{y})\oplus \underline{e}) = 2^{-K}$ (codeword drawn uniformly at random from $\mathcal{C}$), and otherwise $P_{\underline{W}}(\theta(\underline{y})\oplus \underline{e}) = 0$, with
    \begin{eqnarray}
        C'(\underline{y}) = \frac{P_{\underline{W}}(\theta(\underline{y})\oplus \underline{e})}{p_{\underline{Y}}(\underline{y})} = \frac{1}{2^K p_{\underline{Y}}(\underline{y})};
    \end{eqnarray}
    \item \eqref{eqn:posterior_C_3} follows from the memoryless property of the channel;
    \item \eqref{eqn:posterior_C_4} is Bayes' rule, with
    \begin{eqnarray}\label{eqn:C-y}
        C(\underline{y}) = C'(\underline{y}) \frac{\prod_{i = 1}^N p_Y(y_i)}{\prod_{i = 1}^N P_W(\theta(y_i) \oplus e_i)} = 2^{N - K} \frac{\prod_{i = 1}^N p_Y(y_i)}{p_{\underline{Y}}(\underline{y})},
    \end{eqnarray}
    where we note that each $W_i$ has the same marginal probability, i.e., uniform over $\mathbb{F}_2$, and consequently that each $Y_i$   
    has the same marginal probability density distribution, i.e., $p_Y(y) = (q_0(y) + q_1(y))/2$.
    \item \eqref{eqn:posterior_C_5} follows from the same argument as that in the proof of Lemma~\ref{lemma:posterior_probability}.
\end{itemize}
\end{proof}


When considering the random channel output vector $\underline{Y}$ induced by a codebook $\mathcal{C}$, the guessing posterior $P_{\underline{W}|\underline{Y}}\left(\theta(\underline{Y}) \oplus \underline{e} \mid \underline{Y}\right)$ is itself a random variable. Its expectation is called the AGP. The following proposition shows that, although the instantaneous guessing posterior depends on the codebook $\mathcal{C}$, this dependence disappears after averaging over the channel output when the channel is output-symmetric in the sense of \eqref{eqn:output-symmetric-channel}.

\begin{proposition}\label{prop:agp}
    When the channel is output-symmetric in the sense of \eqref{eqn:output-symmetric-channel}, for any codebook $\mathcal{C}$, the AGP is given by:
    \begin{eqnarray}\label{eqn:agp}
        \mathbf{E}_{\underline{Y}} \left[P_{\underline{W}|\underline{Y}}(\theta(\underline{Y})\oplus \underline{e} \mid \underline{Y})\right] = \int_{\mathbb{R}^N}\prod_{i=1}^{N}p_{Y}(y_i) \left(\prod_{i:e_i = 1} \frac{1}{1+\exp(|\ell_i|)} \prod_{i:e_i = 0} \frac{\exp(|\ell_i|)} {1+\exp(|\ell_i|)}\right)\mathrm{d}\underline{y}.
    \end{eqnarray}
\end{proposition}
\begin{proof}
    Considering an arbitrary codebook $\mathcal{C}$, we have
    \begin{eqnarray}\label{eq:agp_expect2int}
        \mathbf{E}_{\underline{Y}} \left[P_{\underline{W}|\underline{Y}}(\theta(\underline{Y})\oplus \underline{e} \mid \underline{Y})\right] = \int_{\mathbb{R}^N} p_{\underline{Y}}(\underline{y}) P_{\underline{W}|\underline{Y}}(\theta(\underline{y})\oplus \underline{e} \mid \underline{y})
        \mathrm{d}\underline{y}.
    \end{eqnarray}
    Applying Lemma~\ref{lemma:posterior_probability_codebook} and plugging into the expression \eqref{eqn:C-y} of $C(\underline{y})$, we obtain
    \begin{eqnarray}
        && \mathbf{E}_{\underline{Y}} \left[P_{\underline{W}|\underline{Y}}(\theta(\underline{Y})\oplus \underline{e} \mid \underline{Y})\right]\\
        &=& \int_{\mathbb{R}^N} 2^{N - K} \prod_{i=1}^{N}p_{Y}(y_i) \left(\prod_{i:e_i = 1} \frac{1}{1+\exp(|\ell_i|)} \prod_{i:e_i = 0} \frac{\exp(|\ell_i|)} {1+\exp(|\ell_i|)}\right) \mathbf{1}\left(\theta(\underline{y}) \oplus \underline{e} \in \mathcal{C}\right) \mathrm{d}\underline{y}.\label{eqn:agp-all}
    \end{eqnarray}
    Under the output-symmetric condition \eqref{eqn:output-symmetric-channel} and its associated properties \eqref{eqn:prop_symm_1} and \eqref{eqn:prop_symm_2}, we have $p_Y(y) = p_Y(-y)$ for any $y \in \mathbb{R}$, and the values $y$ and $-y$ result in the same $|\ell|$. Consider any $\underline{y} \in \mathbb{R}_{+}^{N}$ (i.e., the first quadrant in $\mathbb{R}^N$). There are $2^N$ ways of flipping the signs of the $N$ elements of $\underline{y}$ to yield $2^N$ points in $\mathbb{R}^N$ (including $\underline{y}$ itself), and all these points lead to the same value of $\prod_{i=1}^{N}p_{Y}(y_i) \left(\prod_{i:e_i = 1} \frac{1}{1+\exp(|\ell_i|)} \prod_{i:e_i = 0} \frac{\exp(|\ell_i|)} {1+\exp(|\ell_i|)}\right)$ in the integral in \eqref{eqn:agp-all}. Furthermore, exactly $2^K$ among all these $2^N$ points lead to $\mathbf{1}\left(\theta(\underline{y}) \oplus \underline{e} \in \mathcal{C}\right) = 1$ since there are exactly $2^K$ different codewords in the codebook $\mathcal{C}$. Consequently, we can evaluate the integral \eqref{eqn:agp-all} as
    \begin{eqnarray}
        && \mathbf{E}_{\underline{Y}} \left[P_{\underline{W}|\underline{Y}}(\theta(\underline{Y})\oplus \underline{e} \mid \underline{Y})\right]\\
        &=& 2^K 2^{N - K} \int_{\mathbb{R}_{+}^{N}} \prod_{i=1}^{N}p_{Y}(y_i) \left(\prod_{i:e_i = 1} \frac{1}{1+\exp(|\ell_i|)} \prod_{i:e_i = 0} \frac{\exp(|\ell_i|)} {1+\exp(|\ell_i|)}\right) \mathrm{d}\underline{y}\\
        &=& \int_{\mathbb{R}^{N}} \prod_{i=1}^{N}p_{Y}(y_i) \left(\prod_{i:e_i = 1} \frac{1}{1+\exp(|\ell_i|)} \prod_{i:e_i = 0} \frac{\exp(|\ell_i|)} {1+\exp(|\ell_i|)}\right) \mathrm{d}\underline{y},
    \end{eqnarray}
    where in the last step we have restored the integral over $\mathbb{R}^N$ by noticing that the integrands in all the $2^N$ quadrants are identical according to the output-symmetric condition.
\end{proof}

\begin{remark}
    Inspecting the expression \eqref{eqn:agp} in Proposition~\ref{prop:agp}, we observe that $\underline{W}$ behaves as if it consists of i.i.d. elements, rather than being drawn from a specific codebook $\mathcal{C}$. This is an attractive property for practical purposes, because we can then readily generate abundant samples of i.i.d. $\underline{W}$ vectors and utilize efficient numerical integration methods such as Monte Carlo to evaluate the AGP, for any given $\underline{e}$, without considering the specific codebook to be used.
\end{remark}

As will be shown in the sequel, the metric of AGP provides the appropriate lens for characterizing the ordering behavior of ORB-type GRAND and for establishing principled EP-ordering rules.

\subsection{GRAND}\label{SubSec:GRAND}

A GRAND algorithm mainly consists of two components: an EP generator and a codeword tester. An EP generator takes the LLR vector $\underline{\ell}$ as input, and sequentially outputs a series of vectors in $\mathbb{F}_2^N$ called EPs. Denote the $t$-th EP as $\underline{e}(t)$, $1 \leq t \leq T \leq 2^N$, where $T$ is the maximum number of tests the GRAND algorithm is permitted to conduct.\footnote{For implementation, setting such a maximum number $T$ typically much smaller than $2^N$ is necessary, and in the literature this is called GRAND with abandonment (GRANDAB) \cite{duffy2019GRAND}.} A codeword tester sequentially tests whether $\theta(\underline{y}) \oplus \underline{e}(t) \in \mathcal{C}$ for $t=1,\ldots,T$. It is clear that there are two possibilities:
\begin{itemize}
    \item There exists at least one index within $\{1, 2, \ldots, T\}$ such that the corresponding EP passes the test of codeword tester, and we denote the smallest such index as $\hat{t}$ and declare the decoded codeword to be $\hat{\underline{w}} = \theta(\underline{y}) \oplus \underline{e}(\hat{t})$;
    \item There is no index within $\{1, 2, \ldots, T\}$ such that the corresponding EP passes the test of codeword tester, and we declare a decoding failure with $\hat{\underline{w}} = \emptyset$.
\end{itemize}

For the sent codeword $\underline{w}$, the EP that can lead to successful decoding is
\begin{equation}\label{eqn:target-error-pattern}
    \underline{e}^{*} = \theta(\underline{y}) \oplus \underline{w},
\end{equation}
which we call the target EP. Relative to the position of this target EP in the tested sequence, decoding errors can be decomposed into the following two types:
\begin{itemize}
    \item \textbf{Type~I error (target-miss error):} $\underline{e}^{*} \notin \{\underline{e}(t)\}_{t=1}^{T}$. In this case, the correct codeword cannot be reached within the test budget; the decoder may either stop at a competing codeword or declare a decoding failure.
    \item \textbf{Type~II error (target-preemption error):} $\underline{e}^{*} \in \{\underline{e}(t)\}_{t=1}^{T}$, but an EP $\underline{e}(\hat t)$ satisfying $\theta(\underline{y})\oplus\underline{e}(\hat{t})\in\mathcal{C}$ is encountered before $\underline{e}^{*}$, so the decoder stops before reaching the target.
\end{itemize}

The following proposition revisits a known result (see, e.g., \cite{Lin2004ErrorCC, solomon2020SGRAND, liu2022orbgrand}) regarding the optimal GRAND, which is in fact equivalent to MLD if $T = 2^N$.

\begin{proposition}\label{prop:equal_ml}
    Define the soft weight of $\underline{e} \in \mathbb{F}_2^N$ as $\zeta(\underline{e}) = \sum_{i:e_i = 1} |\ell_i|$. Setting $T = 2^N$ and letting the EP generator output $\underline{e}(t)$, $1 \leq t \leq 2^N$, in such a way that $\{\zeta(\underline{e}(t))\}_{t = 1,\ldots,2^N}$ is monotonically non-decreasing, the resulting GRAND algorithm achieves MLD.
\end{proposition}
\begin{proof} We show that MLD and the thus described GRAND algorithm lead to the same decoded codeword: 
\begin{align}
    \hat{\underline{w}} & = \arg\max_{\underline{w}\in \mathcal{C}}\left\{p_{\underline{Y}|\underline{W}}(\underline{y} \mid \underline{w})\right\} = \arg\max_{\underline{w}\in \mathcal{C}}\left\{\frac{p_{\underline{Y}}(\underline{y})P_{\underline{W}|\underline{Y}}(\underline{w} \mid \underline{y})}{P_{\underline{W}}(\underline{w})}\right\} \\
    & = \theta(\underline{y}) \oplus \arg\max_{\underline{e}\in \mathbb{F}_2^N}\left\{P_{\underline{W}|\underline{Y}}(\theta(\underline{y})\oplus \underline{e} \mid \underline{y})\mathbf{1}(\theta(\underline{y})\oplus \underline{e}\in \mathcal{C})\right\}\label{eq:ml_23}\\
    & = \theta(\underline{y}) \oplus \arg\max_{\underline{e}\in \mathbb{F}_2^N}\left\{\prod_{i:e_i = 1}\frac{1}{1+\exp(|\ell_i|)}\prod_{i:e_i = 0}\frac{\exp(|\ell_i|)}{1+\exp(|\ell_i|)}\mathbf{1}(\theta(\underline{y})\oplus \underline{e}\in \mathcal{C})\}\right\}\label{eq:ml_3} \\
    & = \theta(\underline{y}) \oplus \arg\max_{\underline{e}:\theta(\underline{y})\oplus \underline{e}\in \mathcal{C}}\left\{\exp\left(\sum_{i:e_i = 0}|\ell_i|\right)\right\} \\
    & = \theta(\underline{y}) \oplus \arg\min_{\underline{e}:\theta(\underline{y})\oplus \underline{e}\in \mathcal{C}}\left\{\sum_{i:e_i=1} |\ell_i|\right\} = \theta(\underline{y}) \oplus \arg\min_{\underline{e}:\theta(\underline{y})\oplus \underline{e}\in \mathcal{C}} \left\{\zeta(\underline{e})\right\}. \label{eq:ml_4}
\end{align}
Here,
\begin{itemize}
    \item \eqref{eq:ml_23} rewrites the MLD criterion, noting that $P_{\underline{W}}$ is uniform over $\mathcal{C}$ and that $p_{\underline{Y}}(\underline{y})$ does not depend on the choice of $\underline{w}$;
    \item \eqref{eq:ml_3} is from \eqref{eq:posterior_2} in Lemma~\ref{lemma:posterior_probability_codebook}, again noting that the leading normalization coefficient $C(\underline{y})$ therein can be removed.
\end{itemize}
The rule described in \eqref{eq:ml_4} is that we identify the EP as the one, among all those leading to some codeword in $\mathcal{C}$, that attains the smallest soft weight. When we execute the GRAND algorithm with an EP generator that outputs monotonically non-decreasing soft weights, the procedure exactly solves \eqref{eq:ml_4}. This completes our proof.
\end{proof}

\begin{remark}\label{remark:equal_ml}
    An immediate implication of Proposition~\ref{prop:equal_ml} is that the optimal EP generator produces EPs in decreasing order of the guessing posterior $\left\{P_{\underline{W}|\underline{Y}}(\theta(\underline{y})\oplus \underline{e}(t) \mid \underline{y})\right\}_{t = 1, \ldots, 2^N}$.
    In implementation, we usually set some $T < 2^N$, and it is expected that the resulting GRAND algorithm approximates MLD as $T$ becomes sufficiently large.
\end{remark}

The SGRAND algorithm~\cite{solomon2020SGRAND} implements the optimal EP generator in Proposition~\ref{prop:equal_ml}. Since its soft weight depends upon the actual values of LLRs, it can only be implemented in an ``on-the-fly'' fashion, with its parallelization and hardware implementation still remaining at an exploratory stage~\cite{wan2025parallelism}. More generally, many GRAND algorithms can be viewed as
ordering EPs according to a weighted metric. Following \cite{liu2022orbgrand}, we introduce a generalized measure
\begin{equation}
\zeta'(\underline{e})
=\sum_{i:e_i=1}\gamma_i(\underline{\ell},\underline{e}),
\end{equation}
where $\gamma_i$ are nonnegative functions that may depend on both
$\underline{\ell}$ and $\underline{e}$. We have the following unified description of many GRAND algorithms:
\begin{align}\label{eq:gamma}
    \hat{\underline{w}} = \theta(\underline{y})\oplus \arg\min_{\underline{e}:\theta(\underline{y})\oplus \underline{e}\in \mathcal{C}}\sum_{i:e_i=1} \gamma_{i}(\underline{\ell},\underline{e}).
\end{align}
The flexibility of \eqref{eq:gamma} comes from allowing $\gamma_i$, $i=1,\ldots,N$, to depend on $\underline{\ell}$, on the candidate EP $\underline{e}$, or on both. Table~\ref{Tab:1} gives a compact overview of representative GRAND variants under this formulation, and also indicates which of them admit a pre-defined, ranking-based EP sequence. This distinction is important for the ORB-type class studied below.

\begin{table*}[ht]
    \renewcommand{\arraystretch}{1.12}
    \centering
    \begin{tabular}{|c|c|c|c|c|}
    \hline
    \textbf{Algorithm} & \textbf{$\gamma$ function} & \textbf{Pre-generated?} & \textbf{Channel dependent?} & \textbf{ORB-type?} \\ \hline
    Hard GRAND~\cite{duffy2019GRAND} & $1$ & Yes & No & Yes \\ \hline
    Quantized GRAND~\cite{yuan2023guessing} & $|\ell_i|$ (with quantization) & No & -- & No \\ \hline
    SGRAND~\cite{solomon2020SGRAND} & $|\ell_i|$ & No & -- & No \\ \hline
    SRGRAND~\cite{duffy2021SRGRAND} & $1$ or $\infty$ & No & -- & No \\ \hline
    ORBGRAND~\cite{duffy2022ORBGRAND} & $r_i$ & Yes & No & Yes \\ \hline
    UP-ORBGRAND~\cite{liu2022orbgrand} & $\text{Box}(r_i)$ & Yes & Yes & Yes \\ \hline
    CDF-ORBGRAND~\cite{duffy2022ORBGRAND, liu2022orbgrand} & $\text{CDF}^{-1}(r_i/(N{+}1))$ & Yes & Yes & Yes \\ \hline
    Line-ORBGRAND~\cite{duffy2022ORBGRAND} & $r_i$ (with linear adjustment) & Yes & Yes & Yes \\ \hline
    iLWO-GRAND$^{*}$~\cite{condo2021high} & $r_i$ (with fixed linear penalty) & Yes & No & Yes \\ \hline
    sLWO-GRAND$^{*}$~\cite{ji2024efficient} & $r_i$ (with variable linear penalty) & Yes & No & Yes \\ \hline
    $\beta$Weight-GRAND$^{*}$~\cite{wang2023improved} & $r_i$ (with exponential penalty) & Yes & No & Yes \\ \hline
    RS-ORBGRAND$^{*}$~\cite{wan2024approaching} & -- \quad (see Sec.~\ref{SubSec:RS-ORB}) & Yes & Yes & Yes \\ \hline
    \end{tabular}
    \caption{
    Unified characterization of GRAND algorithms under the generalized metric formulation \eqref{eq:gamma}.
    For algorithms with superscript $^{*}$, $\gamma$ depends on both $\underline{\ell}$ and $\underline{e}$; otherwise, $\gamma$ depends only on $\underline{\ell}$. 
    ``Pre-generated?'' denotes whether the algorithm relies on a pre-defined EP sequence $\mathcal{E}$; 
    ``Channel dependent?'' specifies whether $\mathcal{E}$ is adapted according to channel parameters such as SNR; 
    ``ORB-type?'' identifies whether the algorithm belongs to the ORB-type class (see Sec.~\ref{SubSec:ORBTYPE}).
    CDF denotes the cumulative distribution function.
    }
    \label{Tab:1}
\end{table*}

In particular, when the metric depends only on the reliability ranking of the received symbols, the resulting decoder belongs to the class of ORB-type GRAND algorithms, which will be formally defined in the next subsection.

\subsection{ORB-Type GRAND}\label{SubSec:ORBTYPE}

While the optimal GRAND ordering relies on the exact LLR magnitudes $|\ell_i|$, implementing such ordering may be computationally demanding. A natural idea is therefore to replace $|\ell_i|$ with alternative measures $\gamma_i$ ($i=1,\ldots,N$) that are easier to compute, especially measures determined only by the reliability ranking.
In particular, ORBGRAND~\cite{duffy2022ORBGRAND} has received attention due to its simplicity and effectiveness. To describe ORBGRAND, let $\underline{r}$ denote the ranking vector of $|\underline{\ell}|$, where $r_i$ represents the position of $|\ell_i|$ in the ascending order among the elements of $|\underline{\ell}|$. For example, if $|\underline{\ell}| = (1.2, 0.7, 0.4, 1.1, 0.3)$, then $\underline{r} = (5, 3, 2, 4, 1)$. ORBGRAND simply sets $\gamma_i = r_i$ in \eqref{eq:gamma}. Intuitively, the idea is to approximate the exact value of $|\ell_i|$ with its ranking $r_i$ among $|\underline{\ell}|$.

An alternative description of the ORBGRAND algorithm is as follows:
\begin{itemize}
    \item For each $\underline{e} \in \mathbb{F}_2^N$, calculate the ranking weight of $\underline{e}$ as $\sum_{i:e_i=1} i$, and arrange the elements of $\mathbb{F}_2^N$ as an ordered list $\mathcal{E}$ such that their ranking weights are monotonically non-decreasing (ties broken arbitrarily). If the maximum number of tests $T$ is imposed, then truncate $\mathcal{E}$ to include its first $T$ elements only.
    \item For the received vector $\underline{y}$ and thus its corresponding $|\underline{\ell}|$, obtain the ranking $\underline{r}$ which induces a permutation $\pi_{\underline{y}}$ over $\{1, 2, \ldots, N\}$. We refer to this mapping rule as the rank-based permutation.
    \item The EP generator outputs $\underline{e}(t)$ sequentially, and the codeword tester checks whether $\theta(\underline{y}) \oplus \pi_{\underline{y}}(\underline{e}(t)) \in \mathcal{C}$ for $t=1,\ldots,T$.
\end{itemize}

We emphasize that the ordered EP list $\mathcal{E}$ is an abstract representation of the EP ordering rather than a specific implementation. One possible implementation is to pre-generate and store $\mathcal{E}$ offline, in which case the EP generator sequentially reads elements from $\mathcal{E}$ and applies the permutation $\pi_{\underline{y}}$ to obtain $\{\pi_{\underline{y}}(\underline{e}(t))\}_{t=1,\ldots,T}$. Alternatively, as discussed in \cite{duffy2022ORBGRAND}, the same ordered EP sequence can be generated on-the-fly using efficient search procedures that exploit the integer structure of the ranking weights (e.g., $\sum_{i:e_i=1} i$), without explicitly storing $\mathcal{E}$. These implementations are algorithmically equivalent in that they produce the same ordered EP sequence.

In this paper, we use $\mathcal{E}$ as an abstract notation for the ordered EP set, independent of any particular implementation. For convenience, we denote
\begin{equation}
\pi_{\underline{y}}(\mathcal{E})
=
\{\pi_{\underline{y}}(\underline{e}(1)),\ldots,\pi_{\underline{y}}(\underline{e}(T))\}.
\end{equation}


To illustrate the difference between the ordering rules of SGRAND and ORBGRAND, we consider the following example.

\begin{example}\label{example:ORB}
    As shown in Fig.~\ref{fig:ORBvsS}, assume that the transmitted codeword is $\underline{w}=0000011$ and the received LLR vector is $\underline{\ell} = (2.5, 1.1, -0.8, -0.2, 3.3, -4.1, 0.4)$. Thus, the target EP is $\underline{e}^* = \theta(\underline{y}) \oplus \underline{w} = 0011001$. The SGRAND algorithm tests EPs in ascending order according to the soft weight $\zeta(\underline{e}) = \sum_{i:e_i=1} |\ell_i|$, as shown in the table ``SGRAND ordering''. For $T \ge 10$, the target EP $\underline{e}^*$ is tested at $t=10$. If decoding terminates earlier due to identifying another valid codeword, the correct codeword may not be recovered; otherwise the decoder successfully identifies the transmitted codeword at $t=10$.

    To execute the ORBGRAND algorithm, we first generate the ordered list $\mathcal{E}$, and then obtain the permutation $\pi$ based on the ranking $\underline{r} = (5,4,3,1,6,7,2)$. For convenience of exposition, 
    we may represent $\pi_{\underline y}$ by an $N\times N$ permutation matrix $\mathcal P$ whose $(r_i,i)$-th entry equals $1$ for $i=1,\ldots,N$, with all other entries equal to $0$.
    This yields $\pi_{\underline{y}}(\underline{e}) = \underline{e}\cdot\mathcal{P}$. 
    As can be observed from the table ``ORBGRAND ordering'', ordering EPs according to the ranking weight $\sum_{i:e_i=1} r_i$ does not necessarily preserve the ordering induced by the soft weight $\sum_{i:e_i=1} |\ell_i|$. Consequently, the ORBGRAND algorithm may test some EPs in different ordering than the SGRAND algorithm, thereby explaining its potential performance loss.
\begin{figure*}[htbp]
    \begin{gather*}
        \underline{w} = (0,0,0,0,0,1,1)\xrightarrow{\text{BPSK,AWGN}} \underline{\ell} = (2.5, 1.1, -0.8, -0.2, 3.3, -4.1, 0.4)\\
        \underline{y} \xrightarrow{\text{Hard}} \theta(\underline{y}) = (0,0,1,1,0,1,0), \  \underline{\ell}\xrightarrow{\text{Sort}} \underline{r} = (5,4,3,1,6,7,2)
    \end{gather*}
    \textbf{SGRAND ordering} \hspace{4.5cm} \textbf{ORBGRAND ordering} \hspace{3.2cm} \\[2mm] 
    \footnotesize
$\begin{array}{|c|c|c|}
    \hline \text{Time} & \text{EP } \underline{e} & \sum_{i:e_i=1}|\ell_i| \\
    \hline 1 & 0000000 & 0 \\
    \hline 2 & 0001000 & 0.2 \\ 
    \hline 3 & 0000001 & 0.4 \\
    \hline 4 & 0001001 & 0.6 \\
    \hline 5 & 0010000 & 0.8 \\
    \hline 6 & 0011000 & 1.0 \\
    \hline 7 & 0100000 & 1.1 \\
    \hline 8 & 0010001 & 1.2 \\
    \hline 9 & 0101000 & 1.3 \\
    \hline 10 & 0011001 & 1.4 \\
    \hline
\end{array}$\quad \quad \quad 
        $\begin{array}{|c|}
            \hline \mathcal{E} \\
            \hline 0000000  \\
            \hline 1000000  \\ 
            \hline 0100000  \\
            \hline 0010000  \\
            \hline 1100000  \\
            \hline 0001000  \\
            \hline 1010000  \\
            \hline 0000100  \\
            \hline 1001000  \\
            \hline 0110000  \\
            \hline
        \end{array}\  \begin{aligned} 
            \cdot & \left[\begin{array}{lllllll}
            0 & 0 & 0 & 1 & 0 & 0 & 0 \\
            0 & 0 & 0 & 0 & 0 & 0 & 1 \\
            0 & 0 & 1 & 0 & 0 & 0 & 0 \\
            0 & 1 & 0 & 0 & 0 & 0 & 0 \\
            1 & 0 & 0 & 0 & 0 & 0 & 0 \\
            0 & 0 & 0 & 0 & 1 & 0 & 0 \\
            0 & 0 & 0 & 0 & 0 & 1 & 0
            \end{array}\right]\\
        &\  \left(\begin{array}{lllllll}
            5 & 4 & 3 & 1 & 6 & 7 & 2
            \end{array}\right)\end{aligned}= \begin{array}{|c|c|c|c|}
                \hline \text{Time} & \text{EP } \underline{e} & \sum_{i:e_i=1} |\ell_i| & \sum_{i:e_i=1} r_i \\
                \hline 1 & 0000000 & 0 & 0\\
                \hline 2 & 0001000 & 0.2 & 1\\ 
                \hline 3 & 0000001 & 0.4 & 2\\
                \hline 4 & 0010000 & 0.8 & 3\\
                \hline 5 & 0001001 & 0.6 & 3\\
                \hline 6 & {0100000} & {1.1} & 4\\
                \hline 7 & {0011000} & {1.0} & 4\\
                \hline 8 & 1000000 & 2.5 & 5\\
                \hline 9 & 0101000 & 1.3 & 5\\
                \hline 10 & 0010001 & 1.2 & 5\\
                \hline
            \end{array}$\centering
            \caption{SGRAND and ORBGRAND orderings for the same received vector.}
            \label{fig:ORBvsS}
\end{figure*}
\end{example}

Motivated by the observation in Example~\ref{example:ORB}, several subsequent works have used more sophisticated, yet still ranking-based, choices of $\gamma_i$ to better approximate $|\ell_i|$~\cite{duffy2022ORBGRAND, liu2022orbgrand, condo2021high, wang2023improved, ji2024efficient, wan2024approaching}. As suggested in these works, algorithms whose ordering rule depends on the reliability ranking vector $\underline{r}$ rather than the raw LLR vector $\underline{\ell}$ are referred to as ORB-type GRAND algorithms. In this paper, we formalize this notion by slightly generalizing the description of ORBGRAND, as follows:

\begin{definition}[ORB-type GRAND]\label{def:orb_type}
    An ORB-type GRAND algorithm consists of three components: an ordered list of EPs $\mathcal{E} \subseteq \mathbb{F}_2^N$, a permutation law $\pi_{\underline{y}}$ over $\{1,\ldots, N\}$, and a codeword tester. Here $\mathcal{E}$ is a predetermined ordered list of EPs independent of the received vector $\underline Y$. We denote its $t$-th element as $\underline{e}(t)$, $t = 1, \ldots, T$. 
    The permutation $\pi_{\underline{y}}$ depends upon the received $\underline{y}$, and it can be equivalently represented as an $N \times N$ permutation matrix $\mathcal{P}$. The execution procedure is summarized in Algorithm~\ref{alg:ORBTYPE}.
    \begin{algorithm}[ht]
    \caption{ORB-type GRAND Algorithm}
    \label{alg:ORBTYPE}
    \begin{algorithmic}[1]
        \REQUIRE Received vector $\underline{y}$; ordered EP set $\mathcal{E}$; permutation law $\pi_{\underline{y}}$; maximum number of tests $T$.
        \ENSURE Estimated codeword $\hat{\underline{w}}$.
        \STATE Compute $\underline{\ell}$ from $\underline{y}$, and sort $|\underline{\ell}|$ to obtain $\underline{r}$.
        \STATE Determine $\pi$ and  $\mathcal{P}$ according to $\underline{r}$.
        \STATE Initialize $t \leftarrow 0$ and $\textit{flag} \leftarrow 0$.
        \WHILE{$\textit{flag} = 0$ \AND $t < T$}
            \STATE $t \leftarrow t + 1$
            \STATE $\pi(\underline{e}(t)) \leftarrow \underline{e}(t) \cdot \mathcal{P}$ \hfill // Permute the EP
            \IF{$\theta(\underline{y}) \oplus \pi(\underline{e}(t)) \in \mathcal{C}$} 
                \STATE $\textit{flag} \leftarrow 1$
            \ENDIF
        \ENDWHILE
        \IF{$\textit{flag} = 1$}
            \STATE $\hat{\underline{w}} \leftarrow \theta(\underline{y}) \oplus \pi(\underline{e}(t))$ \hfill // Declare decoded codeword
        \ELSE
            \STATE $\hat{\underline{w}} \leftarrow \emptyset$ \hfill // Decoding failure
        \ENDIF
    \end{algorithmic}
    \end{algorithm}
\end{definition}

In the sequel, whenever probabilities or expectations are taken over the channel output, we usually abbreviate the random permutation $\pi_{\underline{Y}}$ simply as $\pi$ whenever no ambiguity arises. For a fixed received vector $\underline{y}$, $\pi_{\underline{y}}$ denotes the corresponding realization; in algorithmic descriptions, $\pi$ denotes this realized permutation. Two representative examples of ORB-type GRAND are given below.

\begin{itemize}
    \item If the receiver only has hard decision results $\theta(\underline{y})$, we can set $\pi$ to be an identical mapping, and generate $\mathcal{E}$ according to the Hamming weight order. This is the original GRAND for hard decision channels \cite{duffy2019GRAND}. 
    \item For a soft-output channel, the rank-based permutation used above maps the $r_i$-th component of an unpermuted EP to channel coordinate $i$.
    Together with the ranking-weight ordered list $\mathcal{E}$ described before Example~\ref{example:ORB}, this specialization gives standard ORBGRAND~\cite{duffy2022ORBGRAND}.
    The same permutation law can also be combined with other ordered lists $\mathcal{E}$, which is the main degree of freedom behind general ORB-type designs~\cite{liu2022orbgrand, condo2021high}.
\end{itemize}

In Sections~\ref{Sec:random_code} and~\ref{Sec:linear_code}, we will study the performance of ORB-type GRAND algorithms for random code ensembles and fixed linear block codes, respectively.

\subsection{Reshuffled ORBGRAND}\label{SubSec:RS-ORB}

A specific ORB-type GRAND algorithm, called ReShuffled-ORBGRAND (RS-ORBGRAND), has been proposed in \cite{wan2024approaching}. 
The design is motivated by the observation in Proposition~\ref{prop:equal_ml} that the optimal GRAND algorithm generates a non-increasing sequence of guessing posteriors (see also Remark~\ref{remark:equal_ml}). However, sorting the guessing posterior sequence requires ``on-the-fly'' generation of EPs and thus does not lead to an ORB-type GRAND algorithm.

To retain the ORB-type structure, RS-ORBGRAND arranges the EPs in $\mathcal{E}$ such that their AGPs are non-increasing. In an ORB-type GRAND algorithm, the $t$-th tested EP is denoted by $\pi(\underline{e}(t))$, and hence the corresponding AGP is given by
\begin{equation}\label{eq:orb_type_agp_def}
    p_t := \mathbf{E}_{\underline{Y}} \left[P_{\underline{W}|\underline{Y}}(\theta(\underline{Y})\oplus \pi(\underline{e}(t)) \mid \underline{Y})\right].
\end{equation}
Similar to Proposition~\ref{prop:agp}, Proposition~\ref{prop:agp_pi} further shows that $p_t$ can be computed offline, thereby ensuring the practical feasibility of this design.

\begin{proposition}\label{prop:agp_pi}
    Consider an output-symmetric channel satisfying \eqref{eqn:output-symmetric-channel} and a permutation rule $\pi_{\underline y}$ that depends on $\underline y$ only through the reliability values $|\ell_1|,\ldots,|\ell_N|$. Then for any codebook $\mathcal C$, we have
    \begin{equation}\label{eq:agp_pi}
        \mathbf{E}_{\underline{Y}} \left[P_{\underline{W}|\underline{Y}}(\theta(\underline{Y})\oplus \pi(\underline{e}) \mid \underline{Y})\right] = \int_{\mathbb{R}^N}\prod_{i=1}^{N}p_{Y}(y_i) \left(\prod_{i:\pi_{\underline{y}}(\underline{e})_i = 1} \frac{1}{1+\exp(|\ell_i|)} \prod_{i:\pi_{\underline{y}}(\underline{e})_i = 0} \frac{\exp(|\ell_i|)} {1+\exp(|\ell_i|)}\right)\mathrm{d}\underline{y},
    \end{equation}
    
    In particular, for the rank-based permutation, we have:
    \begin{align}\label{eq:exact_pt}
        p_t = N!\int_{\underline{y}:|\ell_{1}|<|\ell_{2}|<\cdots<|\ell_{N}|}\prod_{i=1}^{N}p_{Y}(y_i)\left(\prod_{i:e_i(t) = 1} \frac{1}{1+\exp(|\ell_i|)} \prod_{i:e_i(t) = 0} \frac{\exp(|\ell_i|)}{1+\exp(|\ell_i|)}\right)\mathrm{d}\underline{y}.
    \end{align}
\end{proposition}
\begin{proof}
    For an arbitrary codebook $\mathcal{C}$, following the steps in \eqref{eq:agp_expect2int} and \eqref{eqn:agp-all}, we have 
    \begin{align}
        & \mathbf{E}_{\underline{Y}} \left[P_{\underline{W}|\underline{Y}}(\theta(\underline{Y})\oplus \pi(\underline{e}) \mid \underline{Y})\right] = \int_{\mathbb{R}^N} p_{\underline{Y}}(\underline{y}) P_{\underline{W}|\underline{Y}}(\theta(\underline{y})\oplus \pi_{\underline{y}}(\underline{e}) \mid \underline{y})
        \mathrm{d}\underline{y}.  \\
        = &\int_{\mathbb{R}^N} 2^{N - K} \prod_{i=1}^{N}p_{Y}(y_i) \left(\prod_{i:\pi_{\underline{y}}(\underline{e})_i = 1} \frac{1}{1+\exp(|\ell_i|)} \prod_{i:\pi_{\underline{y}}(\underline{e})_i = 0} \frac{\exp(|\ell_i|)} {1+\exp(|\ell_i|)}\right) \mathbf{1}\left(\theta(\underline{y}) \oplus \pi_{\underline{y}}(\underline{e}) \in \mathcal{C}\right) \mathrm{d}\underline{y}. 
    \end{align}
    Then, similar to the derivation in Proposition~\ref{prop:agp}, we have:
    \begin{eqnarray}\label{eq:Long4.1}
        \mathbf{E}_{\underline{Y}} \left[P_{\underline{W}|\underline{Y}}(\theta(\underline{Y})\oplus \pi(\underline{e}) \mid \underline{Y})\right] = \int_{\mathbb{R}^N}\prod_{i=1}^{N}p_{Y}(y_i)\left(\prod_{i:\pi_{\underline{y}}(\underline{e})_i = 1 } \frac{1}{1+\exp(|\ell_i|)} \prod_{i:\pi_{\underline{y}}(\underline{e})_i = 0} \frac{\exp(|\ell_i|)} {1+\exp(|\ell_i|)}\right)\mathrm{d}\underline{y}
    \end{eqnarray}
    For the rank-based permutation $\pi$, we further have:
    \begin{align}
        & \mathbf{E}_{\underline{Y}} \left[P_{\underline{W}|\underline{Y}}(\theta(\underline{Y})\oplus \pi(\underline{e}(t)) \mid \underline{Y})\right] = \int_{\mathbb{R}^N}\prod_{i=1}^{N}p_{Y}(y_{i})\left(\prod_{i:e_{r_i} = 1} \frac{1}{1+\exp(|\ell_i|)} \prod_{i:e_{r_i} = 0} \frac{\exp(|\ell_i|)}{1+\exp(|\ell_i|)}\right)\mathrm{d}\underline{y}.\label{eq:Long4.2} \\
        = & N!\int_{\underline{y}:|\ell_{1}|<|\ell_{2}|<\cdots<|\ell_{N}|}\prod_{i=1}^{N}(p_{Y}(y_i))\left(\prod_{i:e_i(t) = 1} \frac{1}{1+\exp(|\ell_i|)} \prod_{i:e_i(t) = 0} \frac{\exp(|\ell_i|)}{1+\exp(|\ell_i|)}\right)\mathrm{d}\underline{y}.\label{eq:Long4.3}
    \end{align}

    To obtain the second equality in \eqref{eq:Long4.3} from \eqref{eq:Long4.2}, we use the symmetry of the integrand under coordinate permutations. Since the channel is memoryless, the density factorizes as $\prod_{i=1}^{N}p_Y(y_i)$, and the remaining product terms depend on $\underline{y}$ only through the multiset $\{|\ell_1|,\ldots,|\ell_N|\}$. Hence, permuting the coordinates of $\underline{y}$ leaves the integrand unchanged. Let $\mathcal{A}$ denote the order-statistics region $\mathcal{A}=\{\underline{y}:|\ell_1|<|\ell_2|<\cdots<|\ell_N|\}$. Assuming a continuous output distribution (so ties occur with probability zero), $\mathbb{R}^N$ can be partitioned into $N!$ disjoint regions obtained by permuting the inequalities in $\mathcal{A}$, and the integral over each region is identical. Therefore, the integral over $\mathbb{R}^N$ equals $N!$ times the integral over $\mathcal{A}$. Moreover, on $\mathcal{A}$ the ranking is $\underline{r}=(1,2,\ldots,N)$, i.e., $r_i=i$, so that $e_{r_i}(t)=e_i(t)$, which yields \eqref{eq:Long4.3}.
\end{proof}

\noindent\textit{Practical constructions.}
The AGP-ordering principle can be implemented offline in different ways; these implementations should be distinguished from the theoretical ordering results proved later. We highlight two possible schemes.

\begin{itemize}
    \item \textit{Scheme~1: AGP computation over a candidate list.}
    Start from a finite ordered EP list $\mathcal{E}$ generated by an existing ORB-type GRAND algorithm, compute the AGPs of its elements using Proposition~\ref{prop:agp_pi} or Monte Carlo integration, reshuffle the list in non-increasing AGP order, and then keep the first $T$ EPs. This scheme is directly tied to the AGP expression, but it requires evaluating $p_t$ for the candidate EPs and depends on the quality of the initial candidate list.
    \item \textit{Scheme~2: empirical target-EP ordering.}
    We repeatedly draw a transmitted codeword $\underline{W}$ and generate the corresponding channel output $\underline{Y}$ according to the channel. For each realization, we record the target EP before applying the rank-based permutation,
    \begin{equation}
        \underline{E}_{\mathrm{tar}}
        =
        \pi_{\underline{Y}}^{-1}\!\left(\theta(\underline{Y})\oplus \underline{W}\right).
    \end{equation}
    Conditioned on $\underline{Y}$, the event $\underline{E}_{\mathrm{tar}}=\underline{e}$ is equivalent to $\underline{W}=\theta(\underline{Y})\oplus\pi(\underline{e})$. Hence, by averaging over $\underline{Y}$, the empirical frequency of each unpermuted EP estimates its AGP:
    \begin{equation}
        \Pr(\underline{E}_{\mathrm{tar}}=\underline{e})
        =
        \mathbf{E}_{\underline{Y}}\!\left[
        P_{\underline{W}|\underline{Y}}\!\left(
        \theta(\underline{Y})\oplus \pi(\underline{e})\mid \underline{Y}
        \right)\right].
    \end{equation}
    This scheme avoids evaluating a closed-form expression for $p_t$ and may be useful when such an expression is difficult to obtain.
\end{itemize}

In practice, one may generate an auxiliary list of size $T_1>T$, compute the AGP of each EP in this list, sort the EPs in non-increasing AGP order, and then keep the first $T$ EPs. Algorithm~\ref{alg:E_RS} records this offline finite-list construction. The sorting permutation $\widetilde{\pi}$ is used only to construct the reshuffled list $\widetilde{\mathcal{E}}$; online decoding then follows Algorithm~\ref{alg:ORBTYPE} with this fixed list.

\begin{algorithm}[ht]
    \caption{Finite-list construction of $\widetilde{\mathcal{E}}$ for RS-ORBGRAND}
    \label{alg:E_RS}
    \begin{algorithmic}[1]
        \REQUIRE Initial ordered EP list $\mathcal{E}=\{\underline{e}(t)\}_{t=1}^{T_1}$; maximum number of tests $T<T_1$; permutation criterion $\pi_{\underline{Y}}$.
        \ENSURE Reshuffled EP list $\widetilde{\mathcal{E}}$.
        \FOR{$t = 1$ to $T_1$}
            \STATE Compute the AGP:
            \begin{equation}
                p_t \leftarrow \mathbf{E}_{\underline{Y}}\!\left[
                P_{\underline{W}|\underline{Y}}\!\left(
                \theta(\underline{Y}) \oplus \pi(\underline{e}(t)) \mid \underline{Y}
                \right)\right].
            \end{equation}
        \ENDFOR
        \STATE Sort $\{p_t\}_{t=1}^{T_1}$ in descending order to obtain $\widetilde{\pi}$.
        \STATE Reshuffle $\mathcal{E}$ according to $\widetilde{\pi}$ and keep the first $T$ EPs to form $\widetilde{\mathcal{E}}$.
        \STATE Use \textbf{Algorithm~\ref{alg:ORBTYPE}} with $(\widetilde{\mathcal{E}},\, \pi_{\underline{Y}})$ for decoding.
    \end{algorithmic}
\end{algorithm}

The subsequent sections analyze the AGP-based ordering rule used in the construction above. We first show that ordering EPs by non-increasing AGP is optimal for random code ensembles, in the sense that it simultaneously minimizes the ensemble-average decoding error probability and the ensemble-average number of tests over the EP set under consideration. For fixed linear block codes, the corresponding result is more limited but still useful: under an output-symmetric channel and the rank-based permutation, within any fixed $T$-element candidate EP set, some optimal ordering has a non-increasing AGP sequence $\{p_t\}$. The choice of the candidate set itself is part of the practical construction; in RS-ORBGRAND, this is handled by the finite-list procedure in Algorithm~\ref{alg:E_RS}.

\section{Random Code Ensemble} \label{Sec:random_code}

In this section, we analyze ORB-type GRAND under the random code ensemble. This setting provides a clean baseline in which the effect of EP ordering can be isolated: the target-miss component is determined by the AGPs, while the target-preemption component has an ensemble-averaged form independent of any particular code structure. We first describe the random codebook model under consideration.
We then study three key aspects of decoding performance:
\begin{enumerate}
    \item block error rate (BLER) of ORB-type GRAND algorithms;
    \item distribution of the number of decoding tests;
    \item AGP-based ordering principles for ORB-type GRAND algorithms.
\end{enumerate}
Based on these results, we further establish the ordering principle based on non-increasing AGP, which motivates RS-ORBGRAND and serves as a reference point for the fixed-code analysis in the next section.

\subsection{Random Code Ensemble}\label{SubSec:channel_model}

Starting with the channel model in Section~\ref{SubSec:basic_channel_model}, we further consider the codebook as a subset of size $M = 2^K$ uniformly drawn from all possible size-$2^K$ subsets of $\mathbb{F}_2^N$ at random. Denote the random codebook by $\mathbfcal{C}$, and its realization by $\mathcal{C}$.

Operationally, one may generate $\mathbfcal{C}$ via sampling without replacement. First, draw $\underline{W}(1)$ from $\mathbb{F}_2^N$ uniformly at random; then, draw $\underline{W}(2)$ from $\mathbb{F}_2^N \backslash \{\underline{W}(1)\}$ uniformly at random; continuing this procedure, draw $\underline{W}(m)$ from $\mathbb{F}_2^N \backslash \{\underline{W}(1), \ldots, \underline{W}(m-1)\}$ uniformly at random, until $m = 2^K$.

Note that the random code ensemble of $\mathbfcal{C}$ is different from the more commonly considered i.i.d. random code ensemble, which corresponds to sampling $\mathbb{F}_2^N$ with replacement. For $\mathbfcal{C}$, all the codewords are distinct, and this property facilitates our analysis. Also note that for $\underline{W}$ uniformly drawn from $\mathbfcal{C}$ at random, each of its $N$ positions is uniform over $\mathbb{F}_2$.

\subsection{Analysis of Block Error Rate}\label{SubSec:random_error_rate}

In this subsection, we analyze the BLER averaged over the random code ensemble. The result is stated in the following theorem.

\begin{theorem}[BLER of random code ensemble]\label{th:random_error_rate}
    For a given ORB-type GRAND algorithm, the BLER averaged over the random code ensemble is
    \begin{align}\label{eq:random_bler_main}
        P_{\text{err}} = 1 - P_{\text{succ}} = 1 - \sum_{t=1}^{T}p_t\left(\prod_{i=1}^{t-1}\frac{2^N-2^K+1-i}{2^N-i}\right),
    \end{align}
    where $p_t = \mathbf{E}_{\underline{Y}} \left[P_{\underline{W}|\underline{Y}}(\theta(\underline{Y})\oplus \pi(\underline{e}(t)) \mid \underline{Y})\right]$, $T$ denotes the maximum number of tests allowed by the decoder.
\end{theorem}

The rest of this subsection is devoted to proving Theorem~\ref{th:random_error_rate}. Fig.~\ref{fig:tikz1} visually shows the logic flow chart of the proof. In a nutshell, we decompose $P_{\text{err}}$ into the target-miss and target-preemption components, as shown in Proposition~\ref{prop:error_division}. These two parts are then quantified in Propositions~\ref{prop:search_problem} and~\ref{prop:random_pre_target_hit}, respectively.

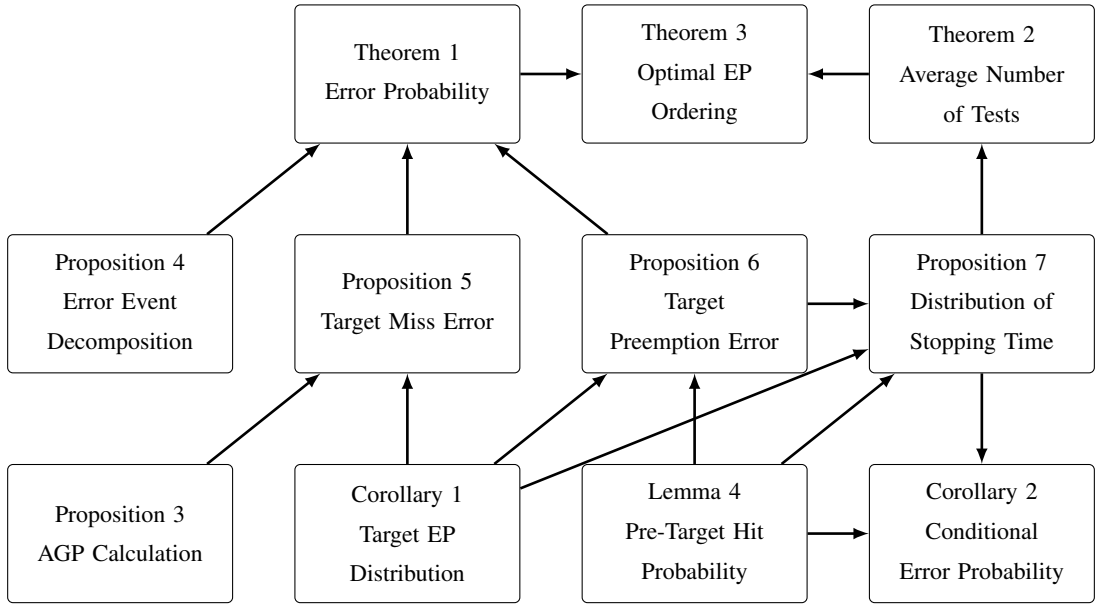
\begin{figure*}[htbp]
    \centering
    \begin{tikzpicture}[
        scale=0.95, >=latex,
        box/.style={
            draw,
            rounded corners=2pt,
            minimum height=5.2em,
            text width=7.8em,
            align=center,
            font=\small
        },
        arrow/.style={
            ->,
            line width=1.0pt
        }
    ]

        \node[box] (n2) at (4,0) {Theorem~\ref{th:random_error_rate} \\ Error Probability};
        \node[box] (n3) at (8,0) {Theorem~\ref{th:random_opti} \\ Optimal EP Ordering};
        \node[box] (n4) at (12,0) {Theorem~\ref{th:ave_guess_num} \\ Average Number of Tests};

        \node[box] (n5) at (0,-3.2) {Proposition~\ref{prop:error_division} \\ Error Event Decomposition};
        \node[box] (n6) at (4,-3.2) {Proposition~\ref{prop:search_problem} \\ Target Miss Error};
        \node[box] (n7) at (8,-3.2) {Proposition~\ref{prop:random_pre_target_hit} \\ Target Preemption Error};
        \node[box] (n8) at (12,-3.2) {Proposition~\ref{prop:stop_at_t} \\ Distribution of Stopping Time};

        \node[box] (n9) at (0,-6.4) {Proposition~\ref{prop:agp_pi} \\ AGP Calculation};
        \node[box] (n10) at (4,-6.4) {Corollary~\ref{coro:distribution_pt} \\ Target EP Distribution};
        \node[box] (n11) at (8,-6.4) {Lemma~\ref{lemma:1} \\ Pre-Target Hit Probability};
        \node[box] (n13) at (12,-6.4) {Corollary~\ref{coro:stop_succ_prob} \\ Conditional Error Probability};

        \draw[arrow] (n5) -- (n2);
        \draw[arrow] (n6) -- (n2);
        \draw[arrow] (n7) -- (n2);

        \draw[arrow] (n9) -- (n6);
        \draw[arrow] (n10) -- (n6);
        \draw[arrow] (n10) -- (n7);
        \draw[arrow] (n11) -- (n7);

        \draw[arrow] (n10) -- (n8);
        \draw[arrow] (n7) -- (n8);
        \draw[arrow] (n11) -- (n8);

        \draw[arrow] (n8) -- (n4);
        \draw[arrow] (n2) -- (n3);
        \draw[arrow] (n4) -- (n3);

        \draw[arrow] (n8) -- (n13);
        \draw[arrow] (n11) -- (n13);

    \end{tikzpicture}
    \caption{Main dependencies among the random code results.}
    \label{fig:tikz1}
\end{figure*}

\begin{proposition}[Error event decomposition]\label{prop:error_division}
The BLER averaged over the random code ensemble can be decomposed as
\begin{align}\label{eq:error-decomposition}
P_{\text{err}} = P_{\text{I}} + P_{\text{II}},
\end{align}
where $P_{\text{I}}$ represents the Type~I (target-miss) error and $P_{\text{II}}$ represents the Type~II (target-preemption) error, specified as follows:
\begin{align}
P_{\text{I}} &= 
\Pr\!\left(
\theta(\underline{Y})\oplus\pi(\underline{e}(t))
\neq
\underline{W}(1),\,
\forall\ t\le T
\right), \\
P_{\text{II}}
&=
\sum_{t=1}^{T}
\Pr\!\left(
\theta(\underline{Y})\oplus\pi(\underline{e}(t))
=
\underline{W}(1)
\right)
\Pr\!\left(
\theta(\underline{Y})\oplus\pi(\underline{e}(t'))
\in\mathbfcal{C},\,
\exists\ t'<t
\mid
\theta(\underline{Y})\oplus\pi(\underline{e}(t))=\underline{W}(1)
\right).
\end{align}
\end{proposition}

\begin{proof}\label{pf:prop:error_division} 
Consider the execution of an ORB-type GRAND algorithm.
Taking the ensemble average of the decoding error probability over the random codebook $\mathbfcal{C}$ yields
\begin{align}
    P_{\text{err}} & = \sum_{\mathcal{C}}\Pr(\mathbfcal{C} = \mathcal{C})\sum_{m=1}^{2^{K}}\Pr(\underline{W} = \underline{w}(m) \mid \mathbfcal{C} = \mathcal{C})\cdot \left[\Pr(\theta(\underline{Y})\oplus \pi(\underline{e}(t))\neq \underline{w}(m),\forall\ t \leq T \mid \underline{W} = \underline{w}(m)) + \right.  \notag \\
    & \left. \sum_{t=1}^{T}\Pr(\theta(\underline{Y})\oplus \pi(\underline{e}(t)) = \underline{w}(m) \mid \underline{W} = \underline{w}(m))\cdot \right. \notag\\
    & \left. \Pr(\theta(\underline{Y})\oplus \pi(\underline{e}(t'))\in \mathcal{C},\exists\ t'< t \mid \theta(\underline{Y})\oplus \pi(\underline{e}(t))=\underline{w}(m), \underline{W} = \underline{w}(m))\right], \label{eq:Long0.1}
\end{align}
which decomposes the decoding error event for given $\mathcal{C}$ and $\underline{w}(m)$ into the two types introduced in Section~\ref{SubSec:GRAND}: 
The first term corresponds to the event that the target EP $\underline{e}^*=\theta(\underline{Y})\oplus\underline{W}(m)$ does not appear in the first $T$ elements of the permuted EP sequence $\{\pi(\underline{e}(t))\}$.
The second term corresponds to the event that the target EP appears at position $t$ with $t\le T$, but another EP $\pi(\underline{e}(t'))$ with $t'<t$ already produces a valid codeword in $\mathcal{C}$, causing a target-preemption error.

Due to the symmetry of $\mathbfcal{C}$, similar to the standard argument in channel coding theorems, we can swap the summations over $\mathcal{C}$ and $m$ in \eqref{eq:Long0.1} and only consider $\underline{W}(1)$ drawn from $\mathbfcal{C}$ and its induced $\underline{Y}$, leading to
\begin{align}
    P_{\text{err}} & = \Pr(\theta(\underline{Y})\oplus \pi(\underline{e}(t))\neq \underline{W}(1),\forall\ t\leq  T) \notag \\
    & \quad \quad + \sum_{t=1}^{T}\Pr(\theta(\underline{Y})\oplus \pi(\underline{e}(t)) = \underline{W}(1)) \notag \\
    & \quad \quad \cdot\Pr(\theta(\underline{Y})\oplus \pi(\underline{e}(t'))\in \mathbfcal{C},\exists\ t'< t \mid \theta(\underline{Y})\oplus \pi(\underline{e}(t))=\underline{W}(1)), \label{eq:Long0.5}
\end{align}
thereby completing the proof.
\end{proof}

In the sequel, we proceed to calculate these two error probabilities respectively.

\subsubsection{Target-Miss Error}\label{SubSub:first_item}
The following proposition characterizes the target-miss probability.

\begin{proposition}[Target-miss error]\label{prop:search_problem}
    The target-miss probability is given by
    \begin{align}
        P_{\text{I}} = \Pr(\theta(\underline{Y})\oplus \pi(\underline{e}(t))\neq \underline{W}(1),\forall\ t \leq T) = 1 - \sum_{t=1}^{T} p_t, \label{eq:first_err}
    \end{align}
    where $p_t = \mathbf{E}_{\underline{Y}}\left[P_{\underline{W}|\underline{Y}}(\theta(\underline{Y})\oplus \pi(\underline{e}(t)) \mid \underline{Y})\right]$.
\end{proposition}

\begin{proof} We have
    \begin{align}
    & \quad \Pr(\theta(\underline{Y})\oplus \pi(\underline{e}(t))\neq \underline{W}(1),\forall\ t \leq  T) \notag \\
    & = \sum_{\mathcal{C}} P_{\mathbfcal{C}}(\mathcal{C})\left[1-\sum_{t=1}^{T}\Pr(\theta(\underline{Y})\oplus \pi(\underline{e}(t))= \underline{w}(1))\right] \label{eq:Long1.1}\\
    & = 1 - \sum_{\underline{w}\in \mathbb{F}_2^N}\sum_{\mathcal{C}:\underline{w}(1)=\underline{w}} P_{\mathbfcal{C}}(\mathcal{C})\left[\sum_{t=1}^{T}\Pr(\theta(\underline{Y})\oplus \pi(\underline{e}(t))= \underline{w})\right] \label{eq:Long1.3}\\
    & = 1 - \sum_{\underline{w}\in \mathbb{F}_2^N} P_{\underline{W}}(\underline{w}) \left[\sum_{t=1}^{T} \int_{\mathbb{R}^N} p_{\underline{Y}|\underline{W}}(\underline{y} \mid \underline{w}) \mathbf{1}(\theta(\underline{y})\oplus \pi_{\underline{y}}(\underline{e}(t))= \underline{w})\mathrm{d}\underline{y}\right] \label{eq:Long1.4}\\
    & = 1 - \int_{\mathbb{R}^N}p_{\underline{Y}}(\underline{y})\left[ \sum_{t=1}^{T} \sum_{\underline{w}\in \mathbb{F}_2^N} P_{\underline{W}|\underline{Y}}(\underline{w} \mid \underline{y}) \mathbf{1}(\theta(\underline{y})\oplus \pi_{\underline{y}}(\underline{e}(t))= \underline{w})\right]\mathrm{d}\underline{y} \label{eq:Long1.5}\\
    & = 1 - \int_{\mathbb{R}^N}p_{\underline{Y}}(\underline{y})\left[ \sum_{t=1}^{T}P_{\underline{W}|\underline{Y}}(\theta(\underline{y})\oplus \pi_{\underline{y}}(\underline{e}(t)) \mid \underline{y})\right]\mathrm{d}\underline{y} \label{eq:Long1.6}\\
    & = 1 - \sum_{t=1}^{T}\left[ \int_{\mathbb{R}^N}p_{\underline{Y}}(\underline{y})P_{\underline{W}|\underline{Y}}(\theta(\underline{y})\oplus \pi_{\underline{y}}(\underline{e}(t)) \mid \underline{y})\mathrm{d}\underline{y}\right] \\
    & = 1 - \sum_{t=1}^{T}\mathbf{E}_{\underline{Y}} \left[P_{\underline{W}|\underline{Y}}(\theta(\underline{Y})\oplus \pi(\underline{e}(t)) \mid \underline{Y})\right] = 1 - \sum_{t=1}^{T} p_t .
    \end{align}
    Here,
    \begin{itemize}
        \item \eqref{eq:Long1.1} holds because under a given codebook $\mathcal{C}$, the transmitted codeword $\underline{w}(1)$ is fixed, and the target EP $\underline{e}^*=\theta(\underline{Y})\oplus\underline{w}(1)$ is unique. Hence the events $\{\theta(\underline{Y})\oplus\pi(\underline{e}(t))=\underline{w}(1)\}_{t=1}^{T}$ are mutually exclusive.
        \item \eqref{eq:Long1.3} is obtained by classifying all possible codebooks into $2^N$ categories according to the value of $\underline{w}(1)$.
        \item \eqref{eq:Long1.4} utilizes $\sum_{\mathcal{C}:\underline{w}(1)=\underline{w}} P_{\mathbfcal{C}}(\mathcal{C}) = P_{\underline{W}}(\underline{w})$ and expands $\Pr(\theta(\underline{Y})\oplus \pi(\underline{e}(t))= \underline{w})$ into integral form.
        \item \eqref{eq:Long1.5} swaps integral and summation, and rewrites the joint probability distribution of $(\underline{W}, \underline{Y})$.
        \item \eqref{eq:Long1.6} utilizes the key observation that, for given $\underline{y}$ and $\pi_{\underline{y}}(\underline{e}(t))$, there is exactly one $\underline{w}$, namely $\theta(\underline{y})\oplus \pi_{\underline{y}}(\underline{e}(t))$, for which the indicator function is one.
    \end{itemize}
\end{proof}

The following corollary further reveals the role of $p_t$.

\begin{corollary}[Distribution of the target EP position]
\label{coro:distribution_pt}
For a given ORB-type GRAND algorithm, the quantities $\{p_t\}_{t=1}^T$ describe the probability distribution of the position of the target EP in the ordered sequence $\{\pi(\underline{e}(t))\}$, i.e.,
\begin{equation}
p_t
=
\Pr(\theta(\underline{Y})\oplus\pi(\underline{e}(t))=\underline{W}(1)).
\end{equation}
\end{corollary}
\begin{proof}
From Proposition~\ref{prop:search_problem}, we have
\begin{equation}
    \Pr(\theta(\underline{Y})\oplus\pi(\underline{e}(i))\neq
    \underline{W}(1),\,\forall\, i \le t) = 1-\sum_{i=1}^{t}p_i .
\end{equation}
Therefore,
\begin{align}
    & \Pr(\theta(\underline{Y})\oplus \pi(\underline{e}(t)) = \underline{W}(1)) \\
    = & \Pr(\theta(\underline{Y})\oplus \pi(\underline{e}(i))\neq \underline{W}(1),\forall\ i \leq t-1) - \Pr(\theta(\underline{Y})\oplus \pi(\underline{e}(i))\neq \underline{W}(1),\forall\ i \leq t)\\
    = & \left(1 - \sum_{i=1}^{t-1}p_i\right) - \left(1 - \sum_{i=1}^{t}p_i\right) = p_t,
\end{align}
\end{proof}

From \eqref{eq:first_err}, together with
$p_t = \mathbf{E}_{\underline{Y}}\!\left[P_{\underline{W}|\underline{Y}}(\theta(\underline{Y})\oplus \pi(\underline{e}(t)) \mid \underline{Y})\right]$, we observe that the probability of target-miss error depends solely on the channel statistics through the quantities $\{p_t\}$ and is independent of the particular realization of the codebook. This reveals a useful separation between channel effects and codebook effects in the analysis. Intuitively, target-miss error occurs when the target EP appears late in the ordered EP sequence. When the channel causes the hard-decision vector $\theta(\underline{Y})$ to differ from the transmitted codeword in many positions, the corresponding target EP typically has low priority in the EP ordering, and thus a small testing budget $T$ leads to a high target-miss probability.

In our prior work \cite{wan2024approaching}, only this aspect was considered, where the decoding problem was effectively reduced to a search problem. However, when $T$ becomes sufficiently large, target-preemption error---arising from competing codewords encountered before the target EP---can become the dominant cause of decoding failure. This effect will be analyzed next.

\subsubsection{Target-Preemption Error}\label{SubSub:second_item}
Using the position-wise decomposition in Proposition~\ref{prop:error_division}, the target-preemption probability in \eqref{eq:error-decomposition} can be written as
\begin{align}
    P_{\text{II}}
    = & \sum_{t=1}^{T}\Pr(\theta(\underline{Y})\oplus \pi(\underline{e}(t)) = \underline{W}(1))\cdot\Pr(\theta(\underline{Y})\oplus \pi(\underline{e}(t'))\in \mathbfcal{C},\exists\ t'< t \mid \theta(\underline{Y})\oplus \pi(\underline{e}(t))=\underline{W}(1))\\
    = & \sum_{t=1}^{T}p_{t}\cdot\Pr(\theta(\underline{Y})\oplus \pi(\underline{e}(t'))\in \mathbfcal{C},\exists\ t'< t \mid \theta(\underline{Y})\oplus \pi(\underline{e}(t))=\underline{W}(1)),\label{eq:Type-II-simplified}
\end{align}
where we have utilized Corollary~\ref{coro:distribution_pt}.

For notational convenience, we introduce the pre-target codeword-hit probability:
\begin{equation}
    f(\mathbfcal{C},\mathcal{E},t) := \Pr(\theta(\underline{Y})\oplus \pi(\underline{e}(t'))\in \mathbfcal{C},\exists\ t'< t \mid \theta(\underline{Y})\oplus \pi(\underline{e}(t))=\underline{W}(1)), \label{def:pre_target_hit}
\end{equation}
which is the conditional probability that although the $t$-th test would identify the sent codeword, at least one earlier EP during the first $t-1$ tests leads to a competing codeword. Clearly we have $f(\mathcal{C}, \mathcal{E}, 1) = 0$. With the pre-target codeword-hit probability, the target-preemption probability in \eqref{eq:Type-II-simplified} is further rewritten as
\begin{equation}
    \sum_{t=1}^{T}p_{t} f(\mathbfcal{C},\mathcal{E},t).\label{eq:Type-II-simplified-f}
\end{equation}

The following lemma quantifies the pre-target codeword-hit probability, i.e., the probability that one of the first $t-1$ candidate words coincides with another codeword in the randomly generated codebook.

\begin{lemma}[Pre-target codeword-hit probability]\label{lemma:1}
    For the random code ensemble, the pre-target codeword-hit probability is given by:
    \begin{equation}\label{eq:random_function}
        f(\mathbfcal{C},\mathcal{E},t) = 1 - \prod_{t'=1}^{t-1}\frac{2^N-2^K+1-t'}{2^N-t'},
    \end{equation}
    for $t = 2, \ldots, T$, and $f(\mathbfcal{C},\mathcal{E},1) = 0$.
\end{lemma}

\begin{proof}
    We have
    \begin{align}
        f(\mathbfcal{C},\mathcal{E},t) &= \Pr(\theta(\underline{Y})\oplus \pi(\underline{e}(t'))\in \mathbfcal{C},\exists\ t'< t \mid \theta(\underline{Y})\oplus \pi(\underline{e}(t))=\underline{W}(1))\\
        &= 1 - \Pr(\theta(\underline{Y})\oplus \pi(\underline{e}(t'))\notin \mathbfcal{C},\forall\ t'< t \mid \theta(\underline{Y})\oplus \pi(\underline{e}(t))=\underline{W}(1))\\
        &= 1 - \frac{\Pr(\theta(\underline{Y})\oplus \pi(\underline{e}(t'))\notin \mathbfcal{C},\forall\ t'< t, \theta(\underline{Y})\oplus \pi(\underline{e}(t))=\underline{W}(1))}{\Pr(\theta(\underline{Y})\oplus \pi(\underline{e}(t))=\underline{W}(1))}\\
        &= 1 - \frac{1}{p_{t}} \Pr(\theta(\underline{Y})\oplus \pi(\underline{e}(t'))\notin \mathbfcal{C},\forall\ t'< t, \theta(\underline{Y})\oplus \pi(\underline{e}(t))=\underline{W}(1)),\label{eqn:pre-target-hit-prob-joint}
    \end{align}
    where the last equality is from Corollary~\ref{coro:distribution_pt}.

    We then manipulate the joint probability in \eqref{eqn:pre-target-hit-prob-joint} accordingly:
    \begin{align}
        & \Pr(\theta(\underline{Y})\oplus \pi(\underline{e}(t'))\notin \mathbfcal{C},\forall\ t'< t, \theta(\underline{Y})\oplus \pi(\underline{e}(t))=\underline{W}(1)) \notag\\
        &= \Pr(\{\underline{W}(2),\underline{W}(3),\ldots,\underline{W}(2^K)\}\cap\{\theta(\underline{Y})\oplus \pi(\underline{e}(1)),\ldots,\theta(\underline{Y})\oplus \pi(\underline{e}(t - 1))\} = \emptyset, \theta(\underline{Y})\oplus \pi(\underline{e}(t))=\underline{W}(1))\notag\\
        &= \sum_{\underline{w} \in \mathbb{F}_2^N} \int_{\mathbb{R}^N} \mathrm{d}\underline{y} p_{\underline{Y}}(\underline{y}) P_{\underline{W}(1)|\underline{Y}}(\underline{w} \mid \underline{y}) \mathbf{1}(\theta(\underline{y})\oplus \pi_{\underline{y}}(\underline{e}(t))=\underline{w}) \times\notag\\
        &\quad\quad\ \Pr(\{\underline{W}(2),\underline{W}(3),\ldots,\underline{W}(2^K)\}\cap\{\theta(\underline{y})\oplus \pi_{\underline{y}}(\underline{e}(1)),\ldots,\theta(\underline{y})\oplus \pi_{\underline{y}}(\underline{e}(t - 1))\} = \emptyset \mid \underline{Y} = \underline{y}, \underline{W}(1) = \underline{w}).\label{eqn:emptyset-y-W1}
    \end{align}
    Using the Markov chain $\{\underline{W}(2),\underline{W}(3),\ldots,\underline{W}(2^K)\} \leftrightarrow \underline{W}(1) \leftrightarrow \underline{Y}$, we have
    \begin{align}
        &\Pr(\{\underline{W}(2),\underline{W}(3),\ldots,\underline{W}(2^K)\}\cap\{\theta(\underline{y})\oplus \pi_{\underline{y}}(\underline{e}(1)),\ldots,\theta(\underline{y})\oplus \pi_{\underline{y}}(\underline{e}(t - 1))\} = \emptyset \mid \underline{Y} = \underline{y}, \underline{W}(1) = \underline{w}) \notag\\
        &= \Pr(\{\underline{W}(2),\underline{W}(3),\ldots,\underline{W}(2^K)\}\cap\{\theta(\underline{y})\oplus \pi_{\underline{y}}(\underline{e}(1)),\ldots,\theta(\underline{y})\oplus \pi_{\underline{y}}(\underline{e}(t - 1))\} = \emptyset \mid \underline{W}(1) = \underline{w}),
    \end{align}
    for which, we can follow the generation of $\mathbfcal{C}$ via sampling without replacement from $\mathbb{F}_2^N$ to obtain a closed-form expression as
    \begin{align}
        &\Pr(\{\underline{W}(2),\underline{W}(3),\ldots,\underline{W}(2^K)\}\cap\{\theta(\underline{y})\oplus \pi_{\underline{y}}(\underline{e}(1)),\ldots,\theta(\underline{y})\oplus \pi_{\underline{y}}(\underline{e}(t - 1))\} = \emptyset \mid \underline{W}(1) = \underline{w}) \notag\\
        &= \prod_{m=2}^{2^K}\frac{2^N+2-t-m}{2^N+1-m}.\label{eqn:P-emptyset-1}
    \end{align}
    This can be explained as follows: to ensure $\underline{W}(2) \notin \{\theta(\underline{y})\oplus \pi_{\underline{y}}(\underline{e}(1)),\ldots,\theta(\underline{y})\oplus \pi_{\underline{y}}(\underline{e}(t - 1))\}$ there are $(2^N - 1) - (t - 1) = 2^N - t$ choices among $2^N - 1$ possibilities (the condition $\underline{W}(1) = \underline{w}$ dictates that we can only sample $\underline{W}(2)$ from $2^N - 1$ possible vectors in $\mathbb{F}_2^N$), and continuing this, to ensure $\underline{W}(m) \notin \{\theta(\underline{y})\oplus \pi_{\underline{y}}(\underline{e}(1)),\ldots,\theta(\underline{y})\oplus \pi_{\underline{y}}(\underline{e}(t - 1))\}$ there are $2^N - (m - 1) - (t - 1) = 2^N - t - m + 2$ choices among $2^N - (m - 1)$ possibilities, until $m = 2^K$. By canceling out common denominators and numerators in \eqref{eqn:P-emptyset-1}, we can further rewrite it as
    \begin{align}
        &\Pr(\{\underline{W}(2),\underline{W}(3),\ldots,\underline{W}(2^K)\}\cap\{\theta(\underline{y})\oplus \pi_{\underline{y}}(\underline{e}(1)),\ldots,\theta(\underline{y})\oplus \pi_{\underline{y}}(\underline{e}(t - 1))\} = \emptyset  \mid  \underline{W}(1) = \underline{w}) \notag\\
        &= \prod_{t'=1}^{t-1}\frac{2^N-2^K+1-t'}{2^N-t'}.\label{eqn:P-emptyset-2}
    \end{align}
    Applying \eqref{eqn:P-emptyset-2} back to \eqref{eqn:emptyset-y-W1}, we obtain
    \begin{align}
        & \Pr(\theta(\underline{Y})\oplus \pi(\underline{e}(t'))\notin \mathbfcal{C},\forall\ t'< t, \theta(\underline{Y})\oplus \pi(\underline{e}(t))=\underline{W}(1)) \notag\\
        &= \left(\prod_{t'=1}^{t-1}\frac{2^N-2^K+1-t'}{2^N-t'} \right) \sum_{\underline{w} \in \mathbb{F}_2^N} \int_{\mathbb{R}^N} \mathrm{d}\underline{y} p_{\underline{Y}}(\underline{y}) P_{\underline{W}(1)|\underline{Y}}(\underline{w} \mid \underline{y}) \mathbf{1}(\theta(\underline{y})\oplus \pi_{\underline{y}}(\underline{e}(t))=\underline{w}) \\
        &= \left(\prod_{t'=1}^{t-1}\frac{2^N-2^K+1-t'}{2^N-t'} \right) \int_{\mathbb{R}^N} \mathrm{d}\underline{y} p_{\underline{Y}}(\underline{y}) P_{\underline{W}(1)|\underline{Y}}(\theta(\underline{y})\oplus \pi_{\underline{y}}(\underline{e}(t)) \mid \underline{y}) \\
        &= \left(\prod_{t'=1}^{t-1}\frac{2^N-2^K+1-t'}{2^N-t'} \right) \mathbf{E}_{\underline{Y}}\left[P_{\underline{W}|\underline{Y}}(\theta(\underline{Y})\oplus \pi(\underline{e}(t)) \mid \underline{Y})\right] \\
        &= \left(\prod_{t'=1}^{t-1}\frac{2^N-2^K+1-t'}{2^N-t'} \right) p_{t}.
    \end{align}
    Therefore, \eqref{eqn:pre-target-hit-prob-joint} is
    \begin{align}
        f(\mathbfcal{C},\mathcal{E},t) &= 1 - \frac{1}{p_{t}} \Pr(\theta(\underline{Y})\oplus \pi(\underline{e}(t'))\notin \mathbfcal{C},\forall\ t'< t, \theta(\underline{Y})\oplus \pi(\underline{e}(t))=\underline{W}(1)) \\
        &= 1 - \frac{1}{p_{t}} \left(\prod_{t'=1}^{t-1}\frac{2^N-2^K+1-t'}{2^N-t'} \right) p_{t} \\
        &= 1 - \prod_{t'=1}^{t-1}\frac{2^N-2^K+1-t'}{2^N-t'},
    \end{align}
    completing the proof.    
\end{proof}

We are now ready to establish the probability of target-preemption error, as given by the following proposition.

\begin{proposition}[Target-preemption error]\label{prop:random_pre_target_hit}
    The target-preemption probability is given by:
    \begin{align}
        P_{\mathrm{II}} = \sum_{t=1}^{T} p_{t} \left(1 - \prod_{t' = 1}^{t-1} \frac{2^N-2^K+1-t'}{2^N-t'}\right). \label{eq:second_err}
    \end{align}
\end{proposition}
\begin{proof}
This is immediate from \eqref{eq:Type-II-simplified-f} and Lemma~\ref{lemma:1}.
\end{proof}

Combining Propositions~\ref{prop:error_division}, \ref{prop:search_problem}, and~\ref{prop:random_pre_target_hit}, we establish Theorem~\ref{th:random_error_rate}.

\subsection{Analysis of Number of Tests}\label{SubSec:random_requests_number}

In this subsection, we study the distribution of the number of tests conducted by the codeword tester. As before, $t$ indexes the ordered EP sequence, and $f(t)=f(\mathbfcal{C}, \mathcal{E}, t)$ is the pre-target codeword-hit probability from Lemma~\ref{lemma:1}. We let $P_{\text{stop}}(t)$ denote the probability that the codeword tester stops after exactly $t$ tests.

Proposition~\ref{prop:stop_at_t} below characterizes the probability distribution of the number of tests before declaring the decoded codeword.


\begin{proposition}\label{prop:stop_at_t}
     For a given ORB-type GRAND algorithm, over the random code ensemble, the probability that the codeword tester declares the decoded codeword after $t$ tests is given by:
     \begin{eqnarray}
         P_{\text{stop}}(t) &=& p_t (1 - f(t)) + \left(1-\sum_{i=1}^{t}p_i\right)(f(t+1)-f(t)), \quad\quad t = 1, \ldots, T - 1;\label{eq:end_time}\\
         P_{\text{stop}}(T) &=& \left(1-\sum_{i=1}^{T-1}p_i\right)(1-f(T)).\label{eq:end_time_T}
     \end{eqnarray}
\end{proposition}
\begin{proof}
     For $t = 1, \ldots, T-1$, there are two situations under which the codeword tester stops after exactly $t$ tests: either the target EP is correctly identified at the $t$-th index in the ordered list and no earlier candidate produces a competing codeword, or the target EP is not among the first $t$ EPs in that list and the $t$-th test is the first pre-target codeword hit. So we have
     \begin{align}
        P_{\text{stop}}(t)& = p_t\Pr(\theta(\underline{Y})\oplus \pi(\underline{e}(t'))\notin \mathbfcal{C},\forall\ t'< t \mid \theta(\underline{Y})\oplus \pi(\underline{e}(t))=\underline{W}(1)) \notag \\
        &\quad\quad   +\sum_{i=t+1}^{2^N}\left( p_i  \left[\Pr(\theta(\underline{Y})\oplus \pi(\underline{e}(t'))\notin \mathbfcal{C},\forall\ t'< t \mid \theta(\underline{Y})\oplus \pi(\underline{e}(i)) =\underline{W}(1))\right.\right. \notag \\
        & \quad \quad \quad\quad \left.\left.- \Pr(\theta(\underline{Y})\oplus \pi(\underline{e}(t'))\notin \mathbfcal{C},\forall\ t'< t+1 \mid \theta(\underline{Y})\oplus \pi(\underline{e}(i))=\underline{W}(1))\right]\right) \label{eq:6.1}\\ 
        & = p_t(1 - f(t)) + \sum_{i=t+1}^{2^N}p_i[(1 - f(t)) - (1- f(t+1))], \label{eq:6.2}\\
        & = p_t(1 - f(t)) + \left(1 - \sum_{i=1}^{t}p_i\right)(f(t+1)-f(t)).\label{eq:6.3}
     \end{align}
    In \eqref{eq:6.1}, Corollary~\ref{coro:distribution_pt} is applied. The transition to \eqref{eq:6.2} follows from the exchangeability of the random code ensemble: conditioned on the target appearing at any index $i\ge t$, the probability that no competing codeword appears in the first $t-1$ tests depends only on the prefix length. Hence, for any $i\ge t$,
    \begin{equation}
    \Pr(\theta(\underline{Y})\oplus \pi(\underline{e}(t'))\notin \mathbfcal{C},\forall\, t'< t \mid \theta(\underline{Y})\oplus \pi(\underline{e}(i))=\underline{W}(1)) = 1-f(t),
    \end{equation}
    and similarly the corresponding probability with prefix length $t$ equals $1-f(t+1)$.

     The decoder may also stop only after the budget is reached: if the target has not been tested by the end of the first $T-1$ positions and there is no pre-target codeword hit in those tests, the tester stops after $T$ tests. So
     \begin{align}
        P_{\text{stop}}(T) & = \left(1 - \sum_{i=1}^{T-1} p_i \right) \Pr(\theta(\underline{Y})\oplus \pi(\underline{e}(t'))\notin \mathbfcal{C},\forall\, t'< T \mid \theta(\underline{Y})\oplus \pi(\underline{e}(j))=\underline{W}(1), j \geq T) \\
        &= \left(1 - \sum_{i=1}^{T-1} p_i \right)(1 - f(T)).
     \end{align}

\end{proof}

Based on Proposition~\ref{prop:stop_at_t}, we have the following theorem characterizing the average number of tests.

\begin{theorem}\label{th:ave_guess_num}
For a given ORB-type GRAND algorithm, the average number of tests conducted by the codeword tester over the random code ensemble, denoted by $Q$, is given by
\begin{align}
Q = T-\sum_{t=1}^{T}f(t) - \sum_{t=1}^{T-1}p_t \left( T-t-\sum_{i=t+1}^{T}f(i) \right),
\end{align}
where $p_t$ is defined in Section~\ref{SubSec:random_error_rate} and
$f(t)$ denotes the pre-target codeword-hit probability defined in
Lemma~\ref{lemma:1}.
\end{theorem}
\begin{proof}
    Using the expressions of $P_{\text{stop}}(t)$ in Proposition~\ref{prop:stop_at_t}, we can obtain, after some algebraic manipulations:
    \begin{align}
        Q & = \sum_{t=1}^{T}t\, P_{\text{stop}}(t) \\
        & = T- \sum_{t=1}^{T} f(t) - \sum_{t=1}^{T-1}p_t\left(T -t - \sum_{i=t+1} ^{T}f(i)\right). \label{eq:length.5}
    \end{align}
\end{proof}
    Substituting $f(t) = 1 - \prod_{i=1}^{t-1}\frac{2^N-2^K+1-i}{2^N-i}$ given by Lemma~\ref{lemma:1} into \eqref{eq:length.5}, we have
    \begin{equation}
        Q = \sum_{t=1}^{T}\left(\prod_{i=1}^{t-1}\frac{2^N-2^K+1-i}{2^N-i}\right) - \sum_{t=1}^{T-1}p_t\left(\sum_{i=t+1}^{T}\prod_{j=1}^{i-1}\frac{2^N-2^K+1-j}{2^N-j}\right).
    \end{equation}

From Proposition~\ref{prop:stop_at_t} we can also quantify the probability of decoding correctly, conditioned on the event that the codeword tester stops after $t$ tests.

\begin{corollary}[Conditional decoding success probability]\label{coro:stop_succ_prob}
    For a given ORB-type GRAND algorithm, conditioned on the event that the codeword tester stops after the $t$-th test, the probability that the declared codeword equals the transmitted codeword is
    \begin{align}
        P_{\text{succ}\mid\text{stop}}(t) &= \frac{p_t}{p_t + (1 - \sum_{i=1}^{t}p_i)\left( \frac{2^K-1}{2^N-t}\right)}, \quad\quad t = 1, \ldots, T - 1;\\
        P_{\text{succ}\mid\text{stop}}(T) &= \frac{p_T}{1 - \sum_{i=1}^{T-1}p_i}.
    \end{align}
\end{corollary}
\begin{proof}
    From the proof of Proposition~\ref{prop:stop_at_t}, conditional on the event that the codeword tester stops after the $t$-th test, there are two mutually exclusive possibilities: the target EP is encountered at index~$t$ in the ordered list, and no earlier candidate produces a competing codeword, corresponding to correct decoding; or the target EP is not among the first $t$ EPs, and the $t$-th test is a pre-target codeword hit, corresponding to incorrect decoding. The probabilities of these two situations are respectively
    \begin{equation}
        \Pr(\text{correct at }t) = p_t(1-f(t)) \text{ and } \Pr(\text{incorrect at }t) =  \left(1 - \sum_{i=1}^{t}p_i\right)(f(t+1)-f(t)).
    \end{equation}
    Therefore, under the condition that the codeword tester stops after the $t$-th test, the probability of decoding correctly is:
    \begin{align}
        \frac{\Pr(\text{correct at }t)}{P_{\text{stop}}(t)} & = \frac{p_t(1-f(t))}{p_t(1-f(t)) + (1 - \sum_{i=1}^{t}p_i)(f(t+1)-f(t))} \\
        & = \frac{p_t}{p_t + (1 - \sum_{i=1}^{t}p_i)\left(1 - \frac{1 - f(t+1)}{1 - f(t)}\right)} = \frac{p_t}{p_t + (1 - \sum_{i=1}^{t}p_i)\left( \frac{2^K-1}{2^N-t}\right)},
    \end{align}
    for $t = 1, \ldots, T - 1$, and
    \begin{equation}
        \frac{\Pr(\text{correct at }T)}{P_{\text{stop}}(T)} = \frac{p_T(1-f(T))}{(1-\sum_{i=1}^{T-1}p_i)(1 - f(T))} = \frac{p_T}{1 - \sum_{i=1}^{T-1}p_i},
    \end{equation}
    for stopping after $T$ tests.
\end{proof}

\subsection{Optimal AGP Ordering for the Random Code Ensemble}\label{SubSec:random_conclusion}

By the analysis of the average BLER and the average number of tests in the preceding two subsections, we now establish the optimal ordering principle for ORB-type GRAND over the random code ensemble.

\begin{theorem}[Optimal AGP ordering]\label{th:random_opti}
For ORB-type decoding applied to the random code ensemble, impose a maximum number of tests $T$ and consider the EP set available to the decoder. The decoding success probability $P_{\mathrm{succ}}$ is maximized by selecting the tested EPs so that their AGPs dominate those of the untested EPs and ordering the tested EPs in non-increasing AGP order, i.e.,
 \begin{equation}
     p_1\geq p_2\geq \cdots\geq p_T, \text{ and }\ p_i\geq p_j,\ \forall \ i\leq T<j.
 \end{equation}
This ordering also minimizes the average number of tests.
\end{theorem}
\begin{proof}
    For the probability of decoding correctly, Theorem~\ref{th:random_error_rate} gives
     \begin{equation}\label{eq:prob_succ}
         P_{\text{succ}} = \sum_{t=1}^{T}p_t(1 - f(t)) = \sum_{t=1}^{T}p_t \left(\prod_{i=1}^{t-1}\frac{2^N-2^K+1-i}{2^N-i}\right),
     \end{equation}
     notice that $\left(\prod_{i=1}^{t-1}\frac{2^N-2^K+1-i}{2^N-i}\right)$ decreases monotonically with $t$. Therefore, in order to maximize $P_{\text{succ}}$, the tested EPs should be selected and ordered according to:
     \begin{equation}\label{eq:opti_condition_2}
        p_1\geq p_2\geq \cdots\geq p_T, \text{ and }\ p_i\geq p_j,\ \forall\ i\leq T<j.
     \end{equation}
     Meanwhile, Theorem~\ref{th:ave_guess_num} gives the average number of tests as
     \begin{equation}
         Q = T- \sum_{t=1}^{T} f(t) - \sum_{t=1}^{T-1}p_t\left(T -t - \sum_{i=t+1} ^{T}f(i)\right),
     \end{equation}
     notice that $\left(T -t - \sum_{i=t+1}^{T}f(i)\right)$ decreases monotonically with $t$ since any $f(i)$ is less than 1. Therefore, minimizing $Q$ is equivalent to assigning larger AGPs to larger coefficients, and the EPs should be selected and ordered according to:
     \begin{equation}\label{eq:opti_condition_3}
         p_1\geq p_2\geq \cdots\geq p_{T-1}, \text{ and }\ p_i\geq p_j,\ \forall\ i\leq T-1<j,
     \end{equation}
     which is slightly weaker than \eqref{eq:opti_condition_2}.

     Since \eqref{eq:opti_condition_2} implies the weaker condition \eqref{eq:opti_condition_3}, arranging the EPs according to non-increasing AGP simultaneously maximizes the decoding success probability and minimizes the average number of tests.

\end{proof}

Theorem~\ref{th:random_opti} is an ordering result rather than an algorithmic claim. Algorithm~\ref{alg:E_RS} implements this principle within the ORB-type structure by estimating the AGPs of a finite candidate EP list and then reshuffling the list in non-increasing AGP order. Thus, for the random code ensemble, RS-ORBGRAND can be viewed as a practical realization of the optimal AGP ordering over the candidate list used for construction.

\subsection{Illustrations}
\label{SubSec:random_simulation}

We present some numerical results to illustrate the analysis in this section. In the numerical illustrations we consider the random code ensemble described in Section~\ref{SubSec:channel_model} with $N = 127$ and $K = 113$. We plot in Fig.~\ref{fig:err_1and2} the average BLER and its target-miss and target-preemption components, according to Theorem~\ref{th:random_error_rate} and Propositions~\ref{prop:search_problem} and~\ref{prop:random_pre_target_hit}. We use the original ORBGRAND algorithm~\cite{duffy2022ORBGRAND}, and let the maximum number of tests $T$ change from $1$ to $10^4$. We see that the target-miss probability dominates for small $T$ and decreases with $T$, while the target-preemption probability increases with $T$ and becomes dominant for large $T$. This behavior is consistent with our intuition.
\begin{figure}[htbp]
    \centering
    \includegraphics[width = 0.72\textwidth]{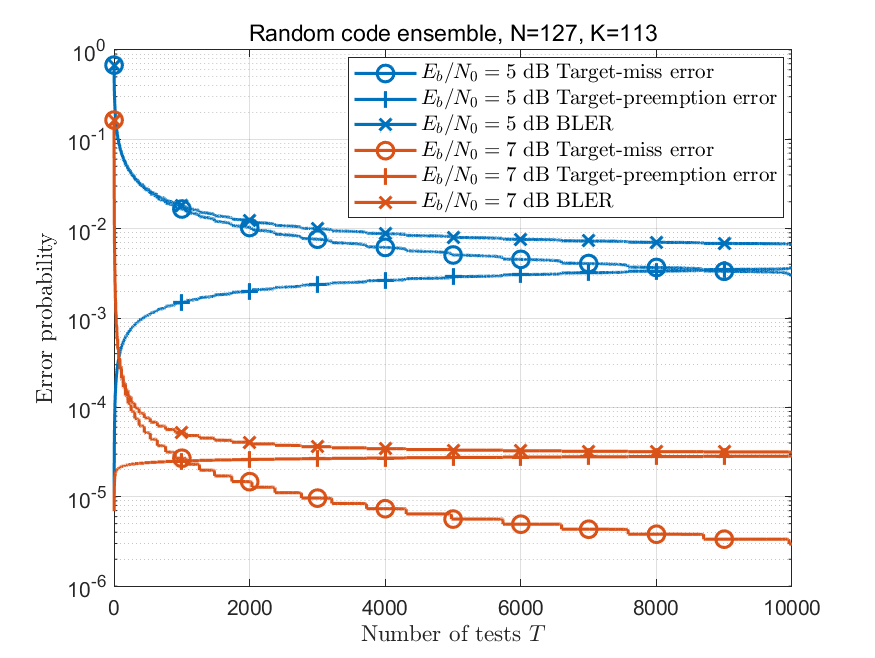}
    \caption{Random code ensemble: error-probability decomposition.}
    \label{fig:err_1and2}
\end{figure}

To gain further intuition on ORB-type GRAND algorithms, Fig.~\ref{fig:pt_stop} plots the AGP sequence $\{p_{t}\}_{t=1,\ldots,T}$ from Proposition~\ref{prop:agp_pi}, with $T = 10^4$, and the stopping probability sequence $\{P_{\text{stop}}(t)\}_{t=1,\ldots,T}$ from Proposition~\ref{prop:stop_at_t}. For ORBGRAND, both sequences show an overall downward trend with $t$ but have highly patterned variations. For RS-ORBGRAND, in contrast, both sequences decrease monotonically with $t$.
\begin{figure}[htbp]
    \centering
    \includegraphics[width = 0.72\textwidth]{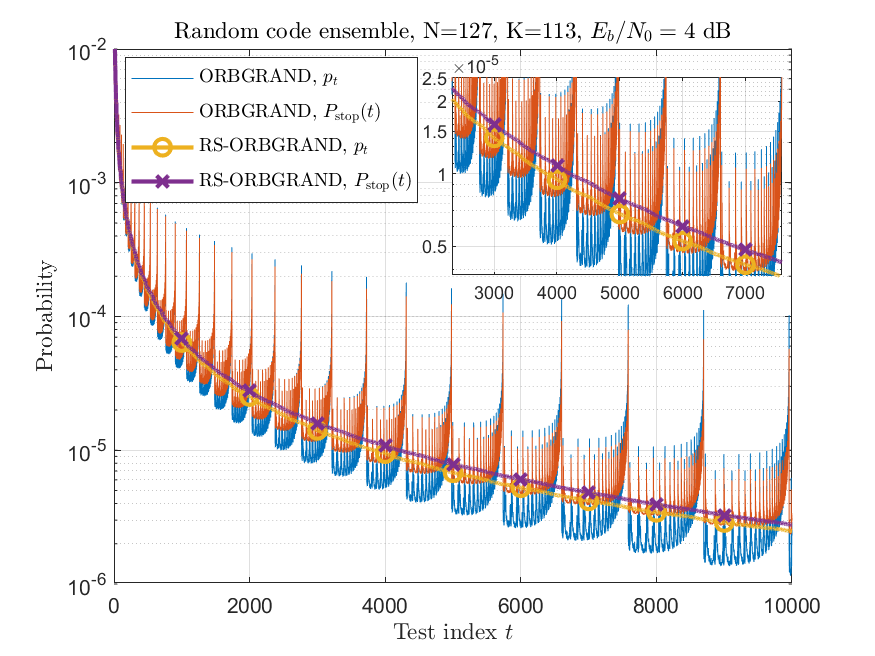}
    \caption{Random code ensemble: $p_{t}$ and $P_{\text{stop}}(t)$ for ORBGRAND and RS-ORBGRAND.}
    \label{fig:pt_stop}
\end{figure}

To illustrate the improvement of RS-ORBGRAND in decoding performance, we compare the BLER of ORBGRAND and RS-ORBGRAND in Fig.~\ref{fig:errprob} for the random code ensemble. We plot BLER curves for different choices of $T$. Starting from $T = 100$, RS-ORBGRAND visibly outperforms ORBGRAND, showing its advantage of prioritizing EPs with larger AGPs. We also observe that, for a fixed $T$, as $E_b/N_0$ grows sufficiently large, the downward trend of BLER curves gradually stalls. This is because ORB-type GRAND performance is ultimately limited by EPs outside the tested set of size $T$, unless $T = 2^N$.
\begin{figure}[htbp]
    \centering
    \includegraphics[width = 0.72\textwidth]{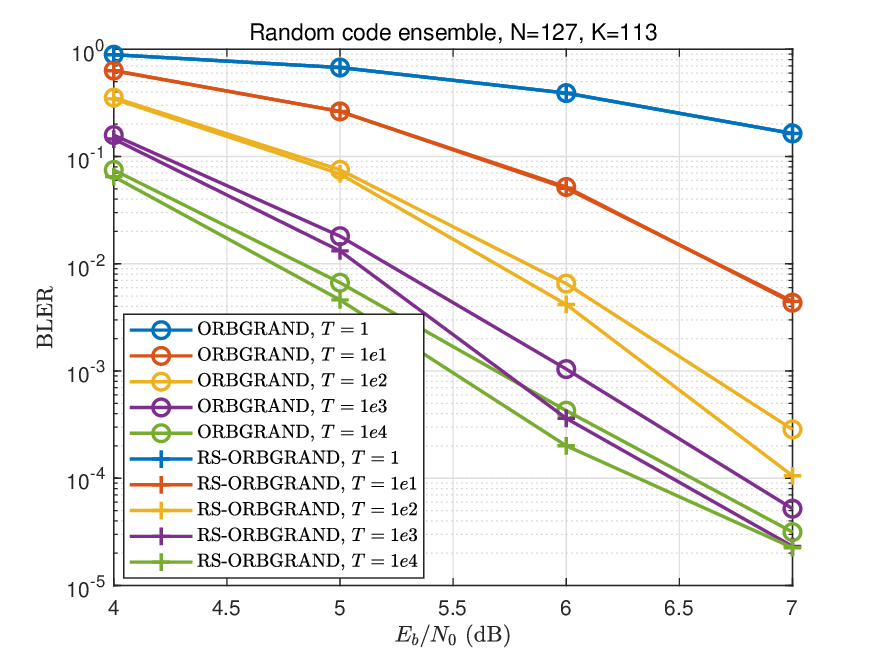}
    \caption{Random code ensemble: BLER comparison for different test budgets.}
    \label{fig:errprob}
\end{figure}

We also examine the performance improvement of RS-ORBGRAND from another perspective, as shown in Fig.~\ref{fig:complete}. The bars show the stopping probability $P_{\text{stop}}(t)$ over different test intervals, obtained from Proposition~\ref{prop:stop_at_t}, while the curves show the corresponding conditional error probability $1-P_{\text{succ}\mid\text{stop}}(t)$, obtained from Corollary~\ref{coro:stop_succ_prob}. Compared with ORBGRAND, RS-ORBGRAND has a higher chance to stop decoding with $10$--$10^2$ tests, and lower chance with $10^2$--$10^3$ and $10^3$--$10^4$ tests. These results further corroborate the ordering principle based on non-increasing AGP in Theorem~\ref{th:random_opti}.
\begin{figure}[htbp]
    \centering
    \includegraphics[width = 0.72\textwidth]{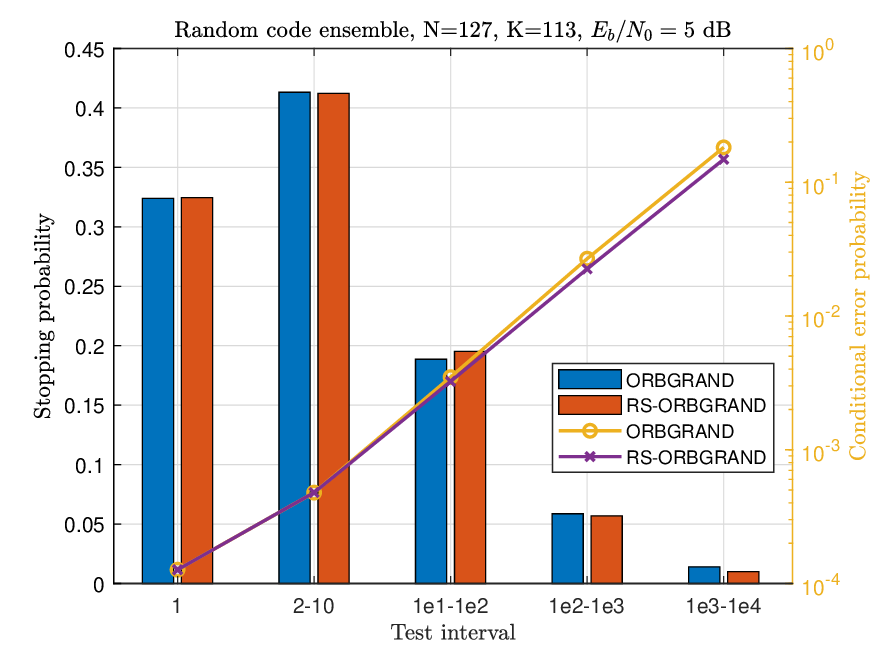}
    \caption{Stopping and conditional error probabilities by test interval.}
    \label{fig:complete}
\end{figure}

\section{Linear Block Codes}\label{Sec:linear_code}

In this section, we extend the ORB-type GRAND analysis from the random code ensemble to fixed linear block codes. The same target-miss/target-preemption decomposition remains useful, but the preemption term is no longer governed only by the number of earlier tests. Instead, it depends on algebraic relations among the tested EP differences and the codewords of the underlying code.

We first formulate the fixed-code model in Section~\ref{SubSec:linear_channel_model}. Section~\ref{SubSec:linear_error_rate} derives the BLER expression by reusing the random-ensemble decomposition where it remains applicable and isolating the code-dependent preemption term. Section~\ref{SubSec:linear_code_conclusion} studies AGP-ordered EP sequences in this setting, and Section~\ref{SubSec:pre_target_hit_linear} characterizes the remaining code-dependent term through code-weight relationships.

\subsection{Linear Block Code Model}\label{SubSec:linear_channel_model}

Consider a binary linear block code with rate $R=K/N$ and generator matrix 
$G\in\mathbb{F}_2^{K\times N}$. The corresponding codebook is
\begin{equation}
\mathcal{C}=\{G^T\underline{u}:\underline{u}\in\mathbb{F}_2^K\},
\end{equation}
which contains $2^K$ distinct codewords.

To facilitate a notation parallel to that used for the random code ensemble in Section~\ref{Sec:random_code}, we introduce a randomized indexing of the codewords in $\mathcal{C}$. Specifically, we regard the codewords in $\mathcal{C}$ as arranged in a random order, chosen uniformly from all $2^K!$ permutations, and denote the resulting ordered random codebook by $\mathbfcal{C}=\{\underline{W}(1),\underline{W}(2),\ldots,\underline{W}(2^K)\}$. Here, $\mathbfcal{C}$ is random only through the ordering of the codewords, whereas the underlying codebook $\mathcal{C}$ itself remains fixed.

Equivalently, this construction can be viewed as sampling all codewords in $\mathcal{C}$ without replacement and recording them in the order in which they are drawn. This randomized ordering is introduced only for analytical convenience, so that the transmitted codeword can be represented as $\underline{W}(1)$ without loss of generality. This randomized indexing does not affect the decoding rule or the performance, but enables a unified probabilistic representation consistent with the random code ensemble.
In the analysis below, we keep $\mathbfcal{C}$ when invoking this randomized indexing argument. Once membership in the underlying code is considered, we write the fixed codebook as $\mathcal{C}$.

\subsection{Analysis of Block Error Rate}\label{SubSec:linear_error_rate}
Unlike the random code ensemble, where the decoding performance depends solely on the EP ordering through $\{p_t\}$, the performance of ORB-type GRAND over linear block codes is also influenced by the structure of the code.

In particular, interactions between different EPs may lead to decoding errors even when the target EP is included in the tested sequence, giving rise to a non-negligible target-preemption component. To capture this effect, we establish a unified analytical expression that explicitly relates the BLER to both the EP ordering and the code structure.

\begin{theorem}[BLER of a linear block code]\label{th:linear_error_rate}
    For a given linear block code $\mathcal{C}$ and a given ORB-type GRAND algorithm, the BLER is given by
    \begin{equation}\label{eq:err_linear}
        P_{\text{err}} = 1 - P_{\text{succ}} =  1 - \sum_{t=1}^{T} p_t \Pr\!\left( \pi(\underline{e}(t') \oplus \underline{e}(t)) \notin \mathcal{C}, \forall\, t' < t \,\middle|\, \theta(\underline{Y})\oplus \pi(\underline{e}(t))=\underline{W}(1) \right)
    \end{equation}
    where $p_t = \mathbf{E}_{\underline{Y}} \left[P_{\underline{W}|\underline{Y}}(\theta(\underline{Y})\oplus \pi(\underline{e}(t))\mid\underline{Y})\right]$.
    
    If the channel is output-symmetric and the permutation $\pi$ is the rank-based permutation, the expression simplifies to
    \begin{equation}\label{eq:err_linear_pi}
        P_{\text{err}} = 1 - \sum_{t=1}^{T}p_t\Pr(\pi(\underline{e}(t') \oplus \underline{e}(t)) \notin \mathcal{C},\forall\, t'< t).
    \end{equation}
\end{theorem}

The rest of this subsection is devoted to proving Theorem~\ref{th:linear_error_rate}. Fig.~\ref{fig:tikz2} summarizes the main conclusions of this section and outlines the logical flow of the derivations for fixed linear block codes.

In the analysis of the BLER, we continue to decompose the overall error event into target-miss and target-preemption components. It is shown that the characterization of the target-miss component remains identical to that under the random code ensemble, whereas the treatment of the target-preemption component requires additional consideration that accounts for the structure of linear block codes.

\begin{figure*}[htbp]
    \centering
    \begin{tikzpicture}[
        scale=0.95, >=latex,
        box/.style={
            draw,
            rounded corners=2pt,
            minimum height=5.2em,
            text width=7.8em,
            align=center,
            font=\small
        },
        arrow/.style={
            ->,
            line width=1.0pt
        }
    ]

        \node[box] (n2) at (4,0) {Theorem~\ref{th:linear_error_rate}\\BLER};
        \node[box] (n3) at (8,0) {Theorem~\ref{th:linear_opti}\\AGP-Ordered EPs};

        \node[box] (n5) at (0,-3.2) {Proposition~\ref{prop:error_division}\\Error Event Decomposition};
        \node[box] (n6) at (4,-3.2) {Proposition~\ref{prop:search_problem}\\Target Miss Error};
        \node[box] (n7) at (8,-3.2) {Theorem~\ref{th:linear_pre_target_hit}\\Target Preemption Error};

        \node[box] (n9) at (0,-6.4) {Proposition~\ref{prop:agp_pi}\\AGP Calculation};
        \node[box] (n10) at (4,-6.4) {Corollary~\ref{coro:distribution_pt}\\Target EP Distribution};
        \node[box] (n11) at (8,-6.4) {Lemma~\ref{lem:linear_pre_target_hit}\\Pre-Target Hit: Event Expansion};
        \node[box] (n12) at (12,-6.4) {Propositions~\ref{prop:rearrange_number} and~\ref{prop:Bj_recursion}\\Pre-Target Hit: Structural Counting};

        \draw[arrow] (n5) -- (n2);
        \draw[arrow] (n6) -- (n2);
        \draw[arrow] (n7) -- (n2);

        \draw[arrow] (n2) -- (n3);

        \draw[arrow] (n9) -- (n6);
        \draw[arrow] (n10) -- (n6);

        \draw[arrow] (n10) -- (n7);
        \draw[arrow] (n11) -- (n7);
        \draw[arrow] (n12) -- (n7);

    \end{tikzpicture}
    \caption{Main dependencies among the fixed linear block code results.}
    \label{fig:tikz2}
\end{figure*}
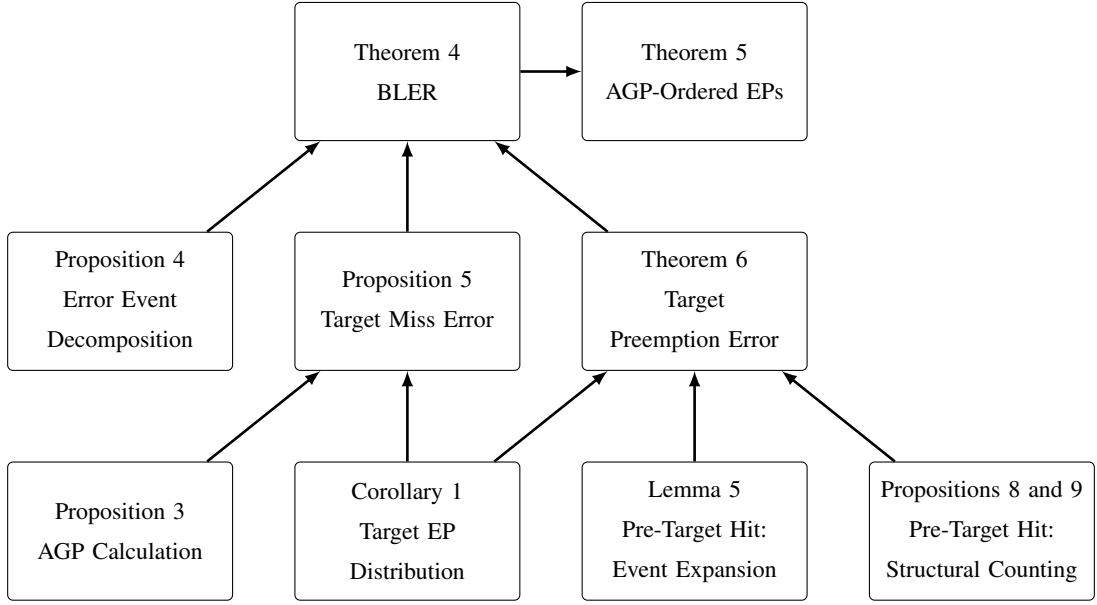
\begin{proof}
    Based on the randomized indexing introduced in Section~\ref{SubSec:linear_channel_model} and the reasoning of Proposition~\ref{prop:error_division}, the total decoding error probability can be expressed as
    \begin{align}
        P_{\text{err}} & = \Pr(\theta(\underline{Y})\oplus \pi(\underline{e}(t))\neq \underline{W}(1),\forall\ t\leq  T) \notag\\
        & \quad \quad + \sum_{t=1}^{T}\Pr(\theta(\underline{Y})\oplus \pi(\underline{e}(t)) = \underline{W}(1)) \notag\\
        & \quad \quad \cdot\Pr(\theta(\underline{Y})\oplus \pi(\underline{e}(t'))\in \mathbfcal{C},\exists\ t'< t \mid \theta(\underline{Y})\oplus \pi(\underline{e}(t))=\underline{W}(1)). \label{eq:err_div_linear}
    \end{align}
    The first term represents the event that the target EP $\theta(\underline{Y}) \oplus \underline{W}(1)$ is not contained in the tested EP set $\pi(\mathcal{E})$ (target-miss error), whereas the second term corresponds to the event that $\theta(\underline{Y}) \oplus \underline{W}(1)\in\pi(\mathcal{E})$ but at least another codeword in $\mathcal{C}$ also satisfies the verification condition before the target is reached (target-preemption error).

    For the target-miss error, the analysis follows the same reasoning as in Section~\ref{SubSec:random_error_rate}. Since its characterization depends only on the ordering of tested EPs and the channel statistics, Proposition~\ref{prop:search_problem}, Corollary~\ref{coro:distribution_pt}, and Proposition~\ref{prop:agp_pi} remain applicable to fixed linear block codes.

    We next consider the target-preemption error. Since $\mathbfcal{C}$ is only a randomized indexing of the fixed codebook $\mathcal{C}$, the codeword-membership event in the second term of \eqref{eq:err_div_linear} can be written with $\mathcal{C}$. We then have
    \begin{align}
        & \sum_{t=1}^{T}\Pr(\theta(\underline{Y})\oplus \pi(\underline{e}(t)) = \underline{W}(1))\Pr(\theta(\underline{Y})\oplus \pi(\underline{e}(t'))\in \mathcal{C},\exists\ t'< t \mid \theta(\underline{Y})\oplus \pi(\underline{e}(t))=\underline{W}(1)) \notag\\
        = & \sum_{t=1}^{T}p_{t}\Pr(\theta(\underline{Y})\oplus \pi(\underline{e}(t'))\in \mathcal{C},\exists\ t'< t \mid \theta(\underline{Y})\oplus \pi(\underline{e}(t))=\underline{W}(1)) \label{eq:Long5.1}\\
        = & \sum_{t=1}^{T}p_{t}\Pr(\theta(\underline{Y})\oplus \pi(\underline{e}(t')) \oplus\theta(\underline{Y})\oplus \pi(\underline{e}(t)) \in \mathcal{C},\exists\ t'< t \mid \theta(\underline{Y})\oplus \pi(\underline{e}(t))=\underline{W}(1)) \label{eq:Long5.2}\\
        = & \sum_{t=1}^{T}p_{t}\Pr(\pi(\underline{e}(t') \oplus \underline{e}(t)) \in \mathcal{C},\exists\ t'< t \mid \theta(\underline{Y})\oplus \pi(\underline{e}(t))=\underline{W}(1)), \label{eq:Long5.3}
    \end{align}
    where the transition from~\eqref{eq:Long5.1} to~\eqref{eq:Long5.2} follows from the fact that $\mathcal{C}$ forms a vector space over $\mathbb{F}_2$ and is closed under modulo-two addition. In particular, under the condition $\theta(\underline{Y})\oplus \pi(\underline{e}(t))=\underline{W}(1)\in \mathcal{C}$, we have
    \begin{equation}
    \theta(\underline{Y})\oplus \pi(\underline{e}(t')) \in \mathcal{C}
    \;\Longleftrightarrow\;
    \pi(\underline{e}(t')) \oplus \pi(\underline{e}(t)) \in \mathcal{C},
    \end{equation}
    which leads to~\eqref{eq:Long5.2}.
    
    In summary,~\eqref{eq:err_linear} is obtained by decomposing the overall decoding error into two parts according to Proposition~\ref{prop:error_division}. Specifically, the target-miss and target-preemption errors are analyzed using Proposition~\ref{prop:search_problem} and the above derivations, respectively, thereby yielding the following unified expression:
    \begin{align}
        P_{\text{err}} & = \underbrace{1 - \sum_{t=1}^{T}p_t}_{\text{target-miss error}} +  \underbrace{\sum_{t=1}^{T}p_{t}\Pr(\pi(\underline{e}(t') \oplus \underline{e}(t)) \in \mathcal{C},\exists\ t'< t \mid \theta(\underline{Y})\oplus \pi(\underline{e}(t))=\underline{W}(1))}_{\text{target-preemption error}} 
    \end{align}

    For an output-symmetric channel and the rank-based random permutation $\pi_{\underline{Y}}$, we show that
    \begin{equation}
        \Pr(\pi_{\underline{Y}}(\underline{e}(t') \oplus \underline{e}(t)) \in \mathcal{C},\exists\, t'< t \mid \theta(\underline{Y})\oplus \pi_{\underline{Y}}(\underline{e}(t))=\underline{W}(1)) = \Pr(\pi_{\underline{Y}}(\underline{e}(t') \oplus \underline{e}(t)) \in \mathcal{C},\exists\, t'< t).
    \end{equation}
    The key observation is that, under an output-symmetric channel, the distribution of $\pi_{\underline{Y}}$ depends only on the order statistics of $|\ell_i|$ and is invariant under permutations of the channel outputs. As a result, $\pi_{\underline{Y}}$ is independent of the event $\theta(\underline{Y})\oplus \pi_{\underline{Y}}(\underline{e}(t))=\underline{W}(1)$.

    To formally establish this, consider any fixed coordinate permutation $\sigma$ with $\Pr(\pi_{\underline{Y}}=\sigma)=1/N!$. By Bayes' rule,
    \begin{equation}\label{eq:pt_with_pi_0}
         \Pr(\pi_{\underline{Y}} = \sigma \mid \theta(\underline{Y})\oplus \pi_{\underline{Y}}(\underline{e}(t))=\underline{W}(1)) = \frac{ \Pr(\pi_{\underline{Y}} = \sigma) \Pr(\theta(\underline{Y})\oplus \pi_{\underline{Y}}(\underline{e}(t))=\underline{W}(1) \mid \pi_{\underline{Y}} = \sigma)}{\Pr(\theta(\underline{Y})\oplus \pi_{\underline{Y}}(\underline{e}(t))=\underline{W}(1))}.
    \end{equation}
    By Corollary~\ref{coro:distribution_pt}, the denominator in \eqref{eq:pt_with_pi_0} equals $p_t$, and for the numerator, following Propositions~\ref{prop:search_problem} and~\ref{prop:agp_pi}, we have:
    \begin{align}
        & \Pr(\pi_{\underline{Y}} = \sigma) \Pr(\theta(\underline{Y})\oplus \pi_{\underline{Y}}(\underline{e}(t))=\underline{W}(1) \mid \pi_{\underline{Y}} = \sigma) \notag \\
        =  & \int_{\mathbb{R}^N}p_{\underline{Y}}(\underline{y})\left(\prod_{i:e_{r_i} = 1 } \frac{1}{1+\exp(|\ell_i|)} \prod_{i:e_{r_i} = 0} \frac{\exp(|\ell_i|)} {1+\exp(|\ell_i|)}\right)\mathbf{1}(\pi_{\underline{y}}= \sigma)\mathrm{d}\underline{y}\label{eq:pt_with_pi_1}
    \end{align}
    
    In~\eqref{eq:pt_with_pi_1}, consider any $\underline{y}$ and all $N!$ permutations $\tilde{\underline{y}}$ obtained by shuffling its elements. The corresponding $\tilde{\underline{r}}$ and $\tilde{\underline{\ell}}$ satisfy $p_{\underline{Y}}(\underline{y}) = p_{\underline{Y}}(\tilde{\underline{y}})$, and
    \begin{equation}
        \prod_{i:e_{r_i} = 1 } \frac{1}{1+\exp(|\ell_i|)} \prod_{i:e_{r_i} = 0} \frac{\exp(|\ell_i|)} {1+\exp(|\ell_i|)} = \prod_{i:e_{\Tilde{r}_i} = 1 } \frac{1}{1+\exp(|\Tilde{\ell}_i|)} \prod_{i:e_{\Tilde{r}_i} = 0} \frac{\exp(|\Tilde{\ell}_i|)} {1+\exp(|\Tilde{\ell}_i|)}.
    \end{equation}
    That is, these $N!$ terms have identical values, but only one of the permutations $\tilde{\underline{y}}$ yields $\pi_{\tilde{\underline{y}}} = \sigma$. Hence,
    \begin{equation}\label{eq:pt_with_pi_2}
        \eqref{eq:pt_with_pi_1} = \frac{1}{N!}\int_{\mathbb{R}^N}p_{\underline{Y}}(\underline{y})\left(\prod_{i:e_{r_i} = 1 } \frac{1}{1+\exp(|\ell_i|)} \prod_{i:e_{r_i} = 0} \frac{\exp(|\ell_i|)} {1+\exp(|\ell_i|)}\right)\mathrm{d}\underline{y} = \frac{p_t}{N!}
    \end{equation}
    Substituting~\eqref{eq:pt_with_pi_2} into~\eqref{eq:pt_with_pi_0} gives $\Pr(\pi_{\underline{Y}} = \sigma \mid \theta(\underline{Y})\oplus \pi_{\underline{Y}}(\underline{e}(t))=\underline{W}(1)) = \frac{1}{N!}= \Pr(\pi_{\underline{Y}} = \sigma)$; that is, $\pi_{\underline{Y}}$ and the event $\theta(\underline{Y})\oplus \pi_{\underline{Y}}(\underline{e}(t))=\underline{W}(1)$ are independent. So we have:
    \begin{equation}
        \eqref{eq:Long5.3} = \sum_{t=1}^{T}p_{t}\Pr(\pi_{\underline{Y}}(\underline{e}(t') \oplus \underline{e}(t)) \in \mathcal{C},\exists\ t'< t). \label{eq:Long5.4}
    \end{equation}
    
\end{proof}


\subsection{Optimal ORB-type GRAND under Linear Block Codes}\label{SubSec:linear_code_conclusion}

This subsection studies the ordering of a fixed candidate EP set for a fixed linear block code. The candidate set itself is assumed to be given; hence the result below should be read as a structural optimality statement for arranging the selected EPs, not as a guarantee that the selected set is globally optimal among all $2^N$ EPs.

Recall from~\eqref{def:pre_target_hit} that $f(\mathcal{C},\mathcal{E},t)$ denotes the pre-target codeword-hit probability when the target EP is $\pi(\underline{e}(t))$. Theorem~\ref{th:linear_error_rate} shows that, for an output-symmetric channel and the rank-based permutation $\pi$, this quantity reduces to
\begin{equation}
    f(\mathcal{C},\mathcal{E},t)
    = \Pr\!\big(\pi(\underline{e}(t') \oplus \underline{e}(t)) \in \mathcal{C},\ \exists\, t' < t\big),
    \label{eq:mif_linear}
\end{equation}
and accordingly,
\begin{equation}
    P_{\text{err}}
    = 1 - \sum_{t=1}^{T} p_t \!\left(1 - \Pr\!\big(\pi(\underline{e}(t') \oplus \underline{e}(t)) \in \mathcal{C},\ \exists\, t' < t\big)\right)
    = 1 - \sum_{t=1}^{T} p_t \big(1 - f(\mathcal{C},\mathcal{E},t)\big),
    \quad
    P_{\text{succ}}
    = \sum_{t=1}^{T} p_t \big(1 - f(\mathcal{C},\mathcal{E},t)\big).
\end{equation}

The output-symmetric channel assumption and the rank-based permutation are the standing conditions under which the conditional expression in Theorem~\ref{th:linear_error_rate} reduces to the unconditional form above. In what follows, we present the corresponding results for a fixed linear block code. 
Compared with the case of the random code ensemble in Theorem~\ref{th:random_opti}, the optimality structure becomes more involved due to the presence of code-induced dependencies among EPs. We therefore state the result in the following form.

\begin{theorem}\label{th:linear_opti}
For any linear block code and any output-symmetric channel, consider ORB-type GRAND decoding with $\pi$ the rank-based permutation. Fix a $T$-element candidate EP set and consider all its orderings. Among the orderings that maximize the decoding success probability, at least one has an associated sequence $\{p_t\}_{t=1}^{T}$ that is non-increasing, i.e.,
\begin{equation}\label{eq:opti_linear_condition}
    p_1 \ge p_2 \ge \cdots \ge p_{T},
\end{equation}
where $p_t=\mathbf{E}_{\underline{Y}}\!\left[P_{\underline{W}|\underline{Y}}(\theta(\underline{Y})\oplus\pi(\underline{e}(t))\mid\underline{Y})\right].$

\end{theorem}

\begin{proof}
We prove the theorem by a pairwise exchange argument. In the same ORB-type GRAND setting as above, only finitely many orderings of the fixed candidate EP set are possible, so $P_{\mathrm{succ}}$ attains its maximum over this finite set; let $\mathcal{E}$ be any ordering that achieves this maximum. If $\{p_t\}$ is not non-increasing, then there exists at least one adjacent pair $(t^{*}, t^{*}+1)$ such that
\begin{equation}\label{eq:th_7_prove_0}
    p_{t^{*}} < p_{t^{*}+1}.
\end{equation}
Exchange these two EPs and denote the resulting ordered set by $\mathcal{E}'$, where $\underline{e}'(i)=\underline{e}(i)$ for $i\neq t^{*},t^{*}+1$, $\underline{e}'(t^{*})=\underline{e}(t^{*}+1)$, and $\underline{e}'(t^{*}+1)=\underline{e}(t^{*})$. Correspondingly, let $p'_i=p_i$ for $i\neq t^{*},t^{*}+1$, $p'_{t^{*}}=p_{t^{*}+1}$, and $p'_{t^{*}+1}=p_{t^{*}}$. The decoding success probabilities before and after the exchange are, respectively,
    \begin{gather}
    P_{\mathrm{succ}}=\sum_{t=1}^{T}p_t-\sum_{t=1}^{T}p_tf(\mathcal{C},\mathcal{E},t),\\
    P_{\mathrm{succ}}'=\sum_{t=1}^{T}p'_t-\sum_{t=1}^{T}p'_tf(\mathcal{C},\mathcal{E}',t).
\end{gather}
We next prove that $P_{\mathrm{succ}}'\ge P_{\mathrm{succ}}$, equivalently that $P_{\mathrm{succ}}'-P_{\mathrm{succ}}\ge 0$. We expand this difference, express the relevant $f(\cdot)$ terms using auxiliary events, and show that the resulting expression is nonnegative under \eqref{eq:th_7_prove_0}. Writing out the difference, we obtain
\begin{align}
    P_{\mathrm{succ}}' - P_{\mathrm{succ}} ={}& \sum_{t=1}^{T}p'_t - \sum_{t=1}^{T}p'_t f(\mathcal{C},\mathcal{E}',t) - \sum_{t=1}^{T}p_t + \sum_{t=1}^{T}p_tf(\mathcal{C},\mathcal{E},t) \notag \\
    ={}& p_{t^{*}}f(\mathcal{C},\mathcal{E},t^{*}) + p_{t^{*}+1}f(\mathcal{C},\mathcal{E},t^{*}+1) - p'_{t^{*}}f(\mathcal{C},\mathcal{E}',t^{*}) - p'_{t^{*}+1}f(\mathcal{C},\mathcal{E}',t^{*}+1) \label{eq:th_7_prove_1}\\
    ={}& p_{t^{*}}f(\mathcal{C},\mathcal{E},t^{*}) + p_{t^{*}+1}f(\mathcal{C},\mathcal{E},t^{*}+1) - p_{t^{*}+1}f(\mathcal{C},\mathcal{E}',t^{*}) - p_{t^{*}}f(\mathcal{C},\mathcal{E}',t^{*}+1) \notag \\
    ={}& p_{t^{*}+1}(f(\mathcal{C},\mathcal{E},t^{*}+1) - f(\mathcal{C},\mathcal{E}',t^{*})) - p_{t^{*}}(f(\mathcal{C},\mathcal{E}',t^{*}+1) - f(\mathcal{C},\mathcal{E},t^{*})). \label{eq:th_7_prove_2}
\end{align}

Define the events
\begin{equation}
B \triangleq \{\pi(\underline{e}(t^{*})\oplus \underline{e}(t^{*}+1))\in\mathcal{C}\},
\end{equation}
\begin{equation}
A_1 \triangleq \{\pi(\underline{e}(t')\oplus \underline{e}(t^{*}))\in\mathcal{C},\ \exists\, t'<t^{*}\},\quad
A_2 \triangleq \{\pi(\underline{e}(t')\oplus \underline{e}(t^{*}+1))\in\mathcal{C},\ \exists\, t'<t^{*}\}.
\end{equation}
As illustrated in Fig.~\ref{fig:linear_pairwise_exchange}, $A_1$ and $A_2$ describe the prefix-related codeword-hit events associated with $\underline{e}(t^{*})$ and $\underline{e}(t^{*}+1)$, respectively, whereas $B$ describes the codeword-hit event induced by the adjacent pair. Hence the adjacent-pair contribution is attached to the later EP in the current ordering: before the exchange it appears in the term for $\underline{e}(t^{*}+1)$, while after the exchange it appears in the term for $\underline{e}(t^{*})$.
Then, by construction,
\begin{align}
    f(\mathcal{C},\mathcal{E},t^{*}) &= \Pr(A_1),\ f(\mathcal{C},\mathcal{E},t^{*}+1) = \Pr(A_2 \cup B), \\
    f(\mathcal{C},\mathcal{E}',t^{*}) &= \Pr(A_2),\ f(\mathcal{C},\mathcal{E}',t^{*}+1) = \Pr(A_1 \cup B).
\end{align}

Therefore, we have
\begin{align}
    f(\mathcal{C},\mathcal{E},t^{*}+1)-f(\mathcal{C},\mathcal{E}',t^{*})
    &= \Pr(A_2\cup B)-\Pr(A_2) = \Pr(B\cap A_2^c), \label{eq:th_7_clean_2} \\
    f(\mathcal{C},\mathcal{E}',t^{*}+1)-f(\mathcal{C},\mathcal{E},t^{*})
    &= \Pr(A_1\cup B)-\Pr(A_1) = \Pr(B\cap A_1^c). \label{eq:th_7_clean_3}
\end{align}

\begin{figure}[t]
\centering
\resizebox{0.98\columnwidth}{!}{%
\begin{tikzpicture}[x=0.88cm,y=0.66cm,>=stealth,every node/.style={font=\scriptsize}]
    \node[anchor=east] at (-0.15,4.2) {Before exchange};
    \draw (0,3.85) rectangle (8.95,4.55);
    \draw (4.4,3.85) -- (4.4,4.55);
    \draw (5.75,3.85) -- (5.75,4.55);
    \draw (7.15,3.85) -- (7.15,4.55);
    \node at (2.2,4.2) {Previous EPs};
    \node at (5.08,4.2) {$\underline{e}(t^{*})$};
    \node at (6.48,4.2) {$\underline{e}(t^{*}+1)$};
    \node at (8.05,4.2) {$\underline{e}(t^{*}+2)\cdots$};

    \draw (0.55,3.76) to[out=-28,in=-152] node[pos=0.57,above=-1pt] {$A_1$} (5.08,3.76);
    \draw (0.35,3.58) to[out=-30,in=-150] node[pos=0.6,above=-1pt] {$A_2$} (6.48,3.58);
    \draw[dashed] (4.33,3.78) rectangle (7.22,4.62);
    \node at (5.78,4.85) {$B$};

    \node[anchor=east] at (-0.15,1.45) {After exchange};
    \draw (0,1.1) rectangle (8.95,1.8);
    \draw (4.4,1.1) -- (4.4,1.8);
    \draw (5.75,1.1) -- (5.75,1.8);
    \draw (7.15,1.1) -- (7.15,1.8);
    \node at (2.2,1.45) {Previous EPs};
    \node at (5.08,1.45) {$\underline{e}(t^{*}+1)$};
    \node at (6.48,1.45) {$\underline{e}(t^{*})$};
    \node at (8.05,1.45) {$\underline{e}(t^{*}+2)\cdots$};

    \draw (0.55,1.01) to[out=-28,in=-152] node[pos=0.57,above=-1pt] {$A_2$} (5.08,1.01);
    \draw (0.35,0.83) to[out=-30,in=-150] node[pos=0.6,above=-1pt] {$A_1$} (6.48,0.83);
    \draw[dashed] (4.33,1.03) rectangle (7.22,1.87);
    \node at (5.78,2.1) {$B$};

    \draw[->,thick] (5.78,3.22) -- (5.78,2.58) node[midway,right] {Exchange};
\end{tikzpicture}
}
\caption{Adjacent exchange of two EPs.}
\label{fig:linear_pairwise_exchange}
\end{figure}
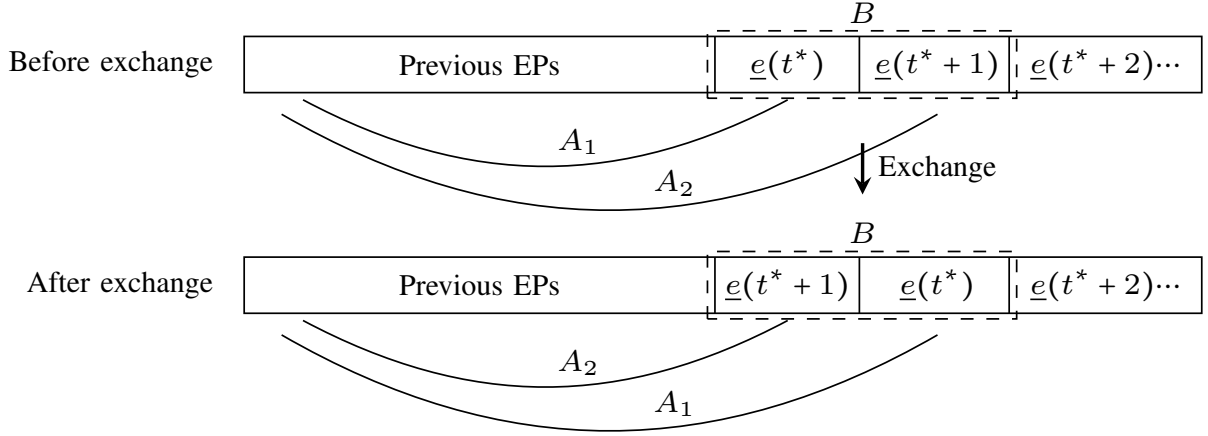

We now show that
\begin{equation}\label{eq:th_7_clean_4}
    B\cap A_1^c = B\cap A_2^c,
\end{equation}
so as to prove $\eqref{eq:th_7_clean_2} = \eqref{eq:th_7_clean_3}$. Under the event $B$, we have
\begin{equation}
\pi(\underline{e}(t^{*})\oplus \underline{e}(t^{*}+1))\in\mathcal{C}.
\end{equation}
For any $t'<t^{*}$,
\begin{equation}
\pi(\underline{e}(t')\oplus \underline{e}(t^{*}))
= \pi(\underline{e}(t')\oplus \underline{e}(t^{*}+1)) \oplus \pi(\underline{e}(t^{*}+1)\oplus \underline{e}(t^{*})).
\end{equation}
Since $\mathcal{C}$ is linear, the modulo-two sum of a codeword and a non-codeword is a non-codeword. Hence, under $B$, the condition
\begin{equation}
\pi(\underline{e}(t')\oplus \underline{e}(t^{*}+1))\notin\mathcal{C},\ \forall\, t'<t^{*}
\end{equation}
holds if and only if
\begin{equation}
\pi(\underline{e}(t')\oplus \underline{e}(t^{*}))\notin\mathcal{C},\ \forall\, t'<t^{*}.
\end{equation}
This proves \eqref{eq:th_7_clean_4}, and therefore
\begin{equation}
    \Pr(B\cap A_1^c)=\Pr(B\cap A_2^c).
\end{equation}
Substituting \eqref{eq:th_7_clean_2} and \eqref{eq:th_7_clean_3} into \eqref{eq:th_7_prove_2}, we obtain
\begin{align}
    P_{\mathrm{succ}}'-P_{\mathrm{succ}}
    &= p_{t^{*}+1}\Pr(B\cap A_2^c)-p_{t^{*}}\Pr(B\cap A_1^c) = \Pr(B\cap A_1^c)\left(p_{t^{*}+1}-p_{t^{*}}\right)\overset{(a)}{\geq}0,\label{eq:th_7_succ_diff_nonneg}
\end{align}
where $(a)$ follows from \eqref{eq:th_7_prove_0}. Hence, whenever an adjacent pair satisfies $p_{t^{*}}<p_{t^{*}+1}$, exchanging these two EPs cannot decrease $P_{\mathrm{succ}}$. The same reasoning applies to any index pair $(t,t+1)$ with $p_t<p_{t+1}$: the corresponding adjacent exchange does not decrease $P_{\mathrm{succ}}$.

Therefore, starting from any ordering of the fixed candidate EP set, we may repeatedly apply such adjacent exchanges (whenever some adjacent pair violates $p_t\ge p_{t+1}$) until the associated $\{p_t\}_{t=1}^{T}$ is non-increasing, and the decoding success probability never decreases along this process. In particular, let $\mathcal{E}$ maximize $P_{\mathrm{succ}}$ over all orderings of this candidate set, which exists because only finitely many such orderings exist. Performing the above exchanges from $\mathcal{E}$ preserves the value of $P_{\mathrm{succ}}$ at its global maximum, and after finitely many steps we obtain an ordering whose associated $\{p_t\}_{t=1}^{T}$ satisfies
\begin{equation}
p_1 \ge p_2 \ge \cdots \ge p_T.
\end{equation}
This exhibits an ordering of the fixed candidate EP set that maximizes the decoding success probability and whose associated sequence $\{p_t\}_{t=1}^{T}$ is non-increasing, as stated in the theorem.
\end{proof}

\begin{remark}
The formulation of Theorem~\ref{th:linear_opti} differs from that of Theorem~\ref{th:random_opti}. 
For linear block codes, code-induced dependencies may lead to degenerated cases in which exchanging two adjacent EPs does not affect the decoding success probability. 
Such cases occur when $\Pr(B\cap A_1^c)=0$, implying that several distinct orderings may achieve the same optimal performance.
\end{remark}

\begin{remark}
Theorem~\ref{th:linear_opti} characterizes a structural property of optimal orderings within a given set of $T$ tested EPs, but it does not provide a complete ordering of the entire EP space. 
In particular, it does not guarantee that every EP selected among the first $T$ tests is preferable to every untested EP. 
At the boundary between the $T$-th and $(T+1)$-th EPs, the impact of exchanging the two depends intricately on the code structure and generally cannot be determined by a simple ordering rule.
\end{remark}

Despite the non-uniqueness of optimal orderings and the lack of a precise characterization at the boundary between the $T$-th and $(T+1)$-th EPs, arranging EPs in a non-increasing order of $\{p_t\}$ remains a natural and well-motivated design principle. 
Theorem~\ref{th:linear_opti} shows that such an ordering is always attainable without loss of optimality within any fixed candidate set, thereby providing a systematic construction guideline with a clear fixed-code scope.

Building upon the structural characterization of optimal EP ordering in Theorem~\ref{th:linear_opti}, we next turn to a more explicit characterization of the decoding error probability for linear block codes. In particular, we aim to characterize the remaining code-dependent term in the BLER and reveal how it is governed by the EP ordering and higher-order weight relationships of the code.

\subsection{Pre-Target Codeword-Hit Probability for Linear Block Codes}\label{SubSec:pre_target_hit_linear}

In this subsection, we further characterize the pre-target codeword-hit probability $f(\mathcal{C},\mathcal{E},t)$ for linear block codes. This is the only code-dependent term in the BLER expression of Theorem~\ref{th:linear_error_rate}. Thus the purpose of this subsection is narrow: we explain how the structure of a fixed code enters this term.

Unlike the random code case, where the pre-target codeword-hit probability depends only on the target position, the linear block code case involves dependencies among multiple EP differences. These dependencies arise because several permuted EP differences may simultaneously fall into the fixed code. The main tool below is therefore a combinatorial representation based on the weight relationships among ordered codeword tuples.

The analysis in this subsection is technically self-contained. It mainly supports the first two parts of the simulation section, where the BLER expression and the impact of code structure are validated. Readers primarily interested in the RS-ORBGRAND performance comparison may skip the derivation on a first reading and return to it when examining those validation results.

The main idea is to separate each joint pre-target hit probability into a code-structure factor and a coordinate-permutation factor. The former counts codeword tuples in $\mathcal{C}$ with a prescribed weight relationship, while the latter measures how likely the rank-based permutation is to realize a matching tuple.

Under a memoryless output-symmetric channel, the coordinate LLRs, and hence the reliabilities $|\ell_1|,\ldots,|\ell_N|$, are i.i.d. Therefore, for the rank-based permutation used for ORB-type GRAND above, each reliability ordering occurs with probability $1/N!$ (assuming ties occur with probability zero, or are broken uniformly). Accordingly, $\pi$ may be regarded as a uniform random permutation in the expressions below. For $i<j$, denote the EP difference by $\underline{e}(i,j)\triangleq\underline{e}(i)\oplus\underline{e}(j)$, and define the event $A_{ij}\triangleq \{\pi(\underline{e}(i,j))\in\mathcal{C}\}$. Then the pre-target codeword-hit probability can be written as $f(\mathcal{C},\mathcal{E},t)
= \Pr\!\left(\bigcup_{i=1}^{t-1}A_{it}\right)$. Applying the inclusion–exclusion principle, we obtain
\begin{equation}\label{eq:pre_target_hit_IEP}
f(\mathcal{C},\mathcal{E},t)
= \sum_{j=1}^{t-1}(-1)^{j-1}
\sum_{1\le i_1<\cdots<i_j<t}
\Pr\!\left(\bigcap_{k=1}^{j}A_{i_k t}\right).
\end{equation}
In particular, retaining only the first-order term in \eqref{eq:pre_target_hit_IEP} gives the union bound, equivalently the first Bonferroni inequality~\cite{feller1968probability},
\begin{equation}\label{eq:first_order_f_bound}
f(\mathcal{C},\mathcal{E},t)
\le \sum_{i=1}^{t-1}\Pr(A_{it}).
\end{equation}
Since the BLER expression in Theorem~\ref{th:linear_error_rate} is increasing in $f(\mathcal{C},\mathcal{E},t)$ for each $t$, this directly yields the computable upper bound
\begin{equation}\label{eq:first_order_bler_bound}
P_{\mathrm{err}}
\le
1-\sum_{t=1}^{T}p_t + 
\sum_{t=1}^{T}p_t\sum_{i=1}^{t-1}\Pr(A_{it}).
\end{equation}
This bound only requires the single-event probabilities $\Pr(A_{it})$, given explicitly in Corollary~\ref{coro:single_event_prob}, and will be used as a simple analytical benchmark in the numerical experiments.

The first-order bound already reflects the code structure through the ordinary weight distribution. The remaining joint probabilities in \eqref{eq:pre_target_hit_IEP} capture finer dependencies: whether several events $A_{i_k t}$ occur simultaneously depends on the algebraic relationships among the corresponding EP differences. These relationships can be fully described by the weights of all XOR combinations of the involved vectors. This motivates the weight-relationship notation below.

Specifically, let $\{0,1\}^t$ denote the set of all binary vectors of length $t$, ordered in lexicographic order. The next two definitions provide the bookkeeping needed to count the joint events in \eqref{eq:pre_target_hit_IEP}: $\mathcal{F}(\cdot)$ records the relevant weight relationships, while $Z(\cdot)$ counts how often a given relationship appears in the codebook.

\begin{definition}[Weight relationship of an ordered codeword tuple]
For an ordered codeword tuple $[\underline{u}(1),\ldots,\underline{u}(t)]$, define its weight relationship as the collection
\begin{equation}\label{eq:def_weight_relationship}
\mathcal{F}(\underline{u}(1),\ldots,\underline{u}(t))
\triangleq
\bigl(g(\underline{b})\bigr)_{\underline{b}\in\{0,1\}^t},
\end{equation}
where $g(\underline{0}) = N$, and for $\underline{b}=(b_1,\ldots,b_t)\neq \underline{0}$,
\begin{equation}
g(\underline{b})
= w_{\mathrm{H}}\!\left(
\bigoplus_{i=1}^{t} b_i\,\underline{u}(i)
\right),
\label{eq:def_g}
\end{equation}
where $w_{\mathrm{H}}(\underline{v})$ denotes the Hamming weight of $\underline{v}\in\{0,1\}^N$.
\end{definition}

\begin{definition}[Multiplicity of a weight relationship]
For any target weight-relationship vector $\mathbf{g}=\bigl(g(\underline{b})\bigr)_{\underline{b}\in\{0,1\}^t}$, define
\begin{equation}\label{eq:def_Z}
Z(\mathbf{g})
=
\sum_{[\underline{u}(1),\ldots,\underline{u}(t)]\in \mathcal{C}^t}
\mathbf{1}\!\left(
\mathcal{F}(\underline{u}(1),\ldots,\underline{u}(t))=\mathbf{g}
\right),
\end{equation}
which counts the number of ordered codeword tuples in $\mathcal{C}^t$ having the weight relationship $\mathbf{g}$.
\end{definition}

The conventional weight distribution is recovered as the special case $t=1$. More generally, $\mathcal{F}(\cdot)$ and $Z(\cdot)$ characterize higher-order structural relationships among multiple codewords. For example, when $t=3$, the collection $\bigl(g(\underline{b})\bigr)_{\underline{b}\in\{0,1\}^3,\,\underline{b}\neq \underline{0}}$ consists of the Hamming weights of all nonempty XOR combinations of $\underline{u}(1),\underline{u}(2),\underline{u}(3)$, and $Z\bigl((g(\underline{b}))_{\underline{b}\in\{0,1\}^3}\bigr)$ counts the number of ordered triples having that prescribed relationship.

The following theorem formalizes this decomposition exactly for all joint terms in the inclusion--exclusion expansion.

\begin{theorem}\label{th:linear_pre_target_hit}
For any linear block code $\mathcal{C}$, the pre-target codeword-hit probability admits the representation
\begin{equation}\label{eq:th_line_mif_main}
f(\mathcal{C},\mathcal{E},t)
= \sum_{j=1}^{t-1} (-1)^{j-1}
\sum_{1\le i_1<\cdots<i_j<t}
\frac{1}{N!}\,
Z(\mathbf{g})\,
\mathcal{H}(\mathbf{g}),
\end{equation}
where $\mathbf{g}=\mathcal{F}(\underline{e}(i_1,t),\ldots,\underline{e}(i_j,t))$, and $\mathcal{H}(\mathbf{g})$ is the coordinate-permutation factor characterized in Proposition~\ref{prop:rearrange_number}.
\end{theorem}

Theorem~\ref{th:linear_pre_target_hit} is an exact fixed-code expression when the full inclusion--exclusion sum is retained and the codeword-tuple counts $Z(\mathbf{g})$ are evaluated exactly. In practical low-BLER evaluation, two distinct approximations may be introduced: truncating the inclusion--exclusion sum to a finite order, and replacing unavailable higher-order values of $Z(\mathbf{g})$ by a structural model. These approximations are not part of the theorem itself and will be stated explicitly when used in the numerical section.

Before proving Theorem~\ref{th:linear_pre_target_hit}, we separate the argument into two simple counting steps. First, a joint pre-target hit can occur only through codeword tuples whose weight relationship matches that of the EP-difference tuple. Second, once such a tuple is fixed, the number of coordinate permutations that realize the match depends only on this common weight relationship.

\begin{lemma}\label{lem:linear_pre_target_hit}
For any $1\le i_1<i_2<\cdots<i_j<t$, the joint probability 
$\Pr\!\left(\bigcap_{k=1}^{j}A_{i_k t}\right)$ can be expressed as
\begin{equation}
\sum_{\substack{
[\underline{u}(1),\ldots,\underline{u}(j)] \in \mathcal{C}^{j}:\\
\mathcal{F}(\underline{u}(1),\ldots,\underline{u}(j)) = \mathbf{g}
}}
\Pr\!\left(
\bigcap_{k=1}^{j}\{\pi(\underline{e}(i_k,t))=\underline{u}(k)\}
\right),
\label{eq:linear_pre_target_hit_prop}
\end{equation}
where $\mathbf{g} = \mathcal{F}\big(\underline{e}(i_1,t),\ldots,\underline{e}(i_j,t)\big)$.
\end{lemma}

\begin{proof}
By definition, $\Pr\!\left(\bigcap_{k=1}^{j}A_{i_k t}\right)=\Pr\!\left(\bigcap_{k=1}^{j}\{\pi(\underline{e}(i_k,t))\in\mathcal{C}\}\right).$ Partition the event on the right-hand side according to the ordered codeword tuple $[\underline{u}(1),\ldots,\underline{u}(j)]\in\mathcal{C}^j$ satisfying $\pi(\underline{e}(i_k,t))=\underline{u}(k)$ for all $k$. Then
\begin{equation}
\Pr\!\left(\bigcap_{k=1}^{j}A_{i_k t}\right)
=
\sum_{[\underline{u}(1),\ldots,\underline{u}(j)]\in\mathcal{C}^j}
\Pr\!\left(\bigcap_{k=1}^{j}\{\pi(\underline{e}(i_k,t))=\underline{u}(k)\}\right).
\label{eq:linear_pre_target_hit_p2_raw}
\end{equation}
Now consider any term in \eqref{eq:linear_pre_target_hit_p2_raw} with nonzero probability. Then there exists a coordinate permutation $\sigma$ such that $\sigma(\underline{e}(i_k,t))=\underline{u}(k)$ for all $k$. Since coordinate permutations preserve the Hamming weights of all XOR combinations, we must have
\begin{equation}
\mathcal{F}(\underline{u}(1),\ldots,\underline{u}(j))
=
\mathcal{F}(\underline{e}(i_1,t),\ldots,\underline{e}(i_j,t))
=
\mathbf{g}.
\end{equation}
Therefore, every term with $\mathcal{F}(\underline{u}(1),\ldots,\underline{u}(j))\neq \mathbf{g}$ is zero, and \eqref{eq:linear_pre_target_hit_prop} follows.
\end{proof}

Lemma~\ref{lem:linear_pre_target_hit} reduces the joint probability to a permutation-counting problem. The following proposition solves this counting problem by grouping coordinates according to the binary column that they form across an ordered tuple.

\begin{proposition}\label{prop:rearrange_number}
For any two ordered binary vector tuples 
$[\underline{x}(1),\ldots,\underline{x}(j)]$ and 
$[\underline{u}(1),\ldots,\underline{u}(j)]$, suppose that
\begin{equation}
\mathcal{F}(\underline{x}(1),\ldots,\underline{x}(j))
=
\mathcal{F}(\underline{u}(1),\ldots,\underline{u}(j))
=
\mathbf{g}.
\end{equation}
For each coordinate $k$, define the state of $[\underline{x}(1),\ldots,\underline{x}(j)]$ by
\begin{equation}
\underline{S}_x(k)=[x_k(1),\ldots,x_k(j)]\in\{0,1\}^j,
\end{equation}
and define $\underline{S}_u(k)$ similarly. Let
\begin{equation}
J(\underline{s})=\bigl|\{k:\underline{S}_x(k)=\underline{s}\}\bigr|,
\qquad \underline{s}\in\{0,1\}^j.
\end{equation}
Then the number of coordinate permutations $\sigma$ satisfying
$\sigma(\underline{x}(i))=\underline{u}(i)$ for all $i=1,\ldots,j$ is
\begin{equation}\label{eq:H_by_J}
\mathcal{H}(\mathbf{g})
=
\prod_{\underline{s}\in\{0,1\}^j} J(\underline{s})!.
\end{equation}
\end{proposition}

\begin{proof}
For a state $\underline{s}\in\{0,1\}^j$, define
\begin{equation}
\mathcal{I}_x(\underline{s})=\{k:\underline{S}_x(k)=\underline{s}\},
\qquad
\mathcal{I}_u(\underline{s})=\{k:\underline{S}_u(k)=\underline{s}\}.
\end{equation}
A permutation $\sigma$ satisfies $\sigma(\underline{x}(i))=\underline{u}(i)$ for all $i$ if and only if, for every state $\underline{s}$, it maps the coordinates in $\mathcal{I}_x(\underline{s})$ bijectively onto the coordinates in $\mathcal{I}_u(\underline{s})$.

This correspondence can be visualized by arranging the vectors into a matrix form, where each column represents a coordinate and its associated state:
\begin{align}
    &\left[\begin{array}[pos]{cccc}
        x_1(1) & x_2(1) & \cdots & x_N(1)\\
        x_1(2) & x_2(2) & \cdots & x_N(2)\\
        \vdots & \vdots & \ddots & \vdots\\
        x_1(j) & x_2(j) & \cdots & x_N(j)
    \end{array}\right]
    \overset{\sigma}{\longrightarrow}
    \left[\begin{array}[pos]{cccc}
        \sigma(\underline{x}(1))_1 & \sigma(\underline{x}(1))_2 & \cdots & \sigma(\underline{x}(1))_N\\
        \sigma(\underline{x}(2))_1 & \sigma(\underline{x}(2))_2 & \cdots & \sigma(\underline{x}(2))_N\\
        \vdots & \vdots & \ddots & \vdots\\
        \sigma(\underline{x}(j))_1 & \sigma(\underline{x}(j))_2 & \cdots & \sigma(\underline{x}(j))_N
    \end{array}\right],
    \label{eq:matrix}
\end{align}
where each column is a state vector and must be mapped to a column with the same state in the target tuple.

Because the two tuples have the same weight relationship $\mathbf{g}$, their state multiplicities agree for all $\underline{s}$, as made explicit in Proposition~\ref{prop:Bj_recursion}. Hence, for each state $\underline{s}$ there are $J(\underline{s})!$ possible bijections between $\mathcal{I}_x(\underline{s})$ and $\mathcal{I}_u(\underline{s})$. Multiplying over all states gives \eqref{eq:H_by_J}.
\end{proof}

It remains only to justify that the state multiplicities in Proposition~\ref{prop:rearrange_number} are indeed determined by the weight relationship $\mathbf{g}$. The following proposition gives this relation explicitly; its recursion and invertibility details are deferred to the appendix.

\begin{proposition}\label{prop:Bj_recursion}
Let $\mathbf{g}^{(j)}=\bigl(g(\underline{b})\bigr)_{\underline{b}\in\{0,1\}^j}$ and $\mathbf{J}^{(j)}=\bigl(J(\underline{s})\bigr)_{\underline{s}\in\{0,1\}^j}$, where both vectors are ordered according to the lexicographic order on $\{0,1\}^j$. Then
\begin{equation}\label{eq:relation_g_and_j_matrix}
\mathbf{g}^{(j)}=\mathbf{B}_j\,\mathbf{J}^{(j)},
\end{equation}
where the matrix-vector product is taken over the integers, and $\mathbf{B}_j$ is the $2^j\times 2^j$ matrix indexed by $\{0,1\}^j\times\{0,1\}^j$ with entries
\begin{equation}\label{eq:Bj_entry_new}
\bigl[\mathbf{B}_j\bigr]_{\underline{b},\underline{s}}
=
\begin{cases}
1, & \underline{b}=\underline{0},\\
\underline{b}\cdot \underline{s}, & \underline{b}\neq \underline{0},
\end{cases}
\qquad
\underline{b},\underline{s}\in\{0,1\}^j,
\end{equation}
and $\underline{b}\cdot \underline{s}=\bigoplus_{i=1}^j b_i s_i$.
Moreover, $\mathbf{B}_j$ admits the recursion
\begin{equation}\label{eq:Bj_recursion_new}
\mathbf{B}_{j+1}
=
\begin{bmatrix}
\mathbf{B}_j & \mathbf{B}_j\\
\mathbf{E}^{2^j\times 2^j}_{1,1:2^j}\oplus \mathbf{B}_j &
\mathbf{E}^{2^j\times 2^j}_{1,1:2^j}\oplus \mathbf{1}^{2^j\times 2^j}\oplus \mathbf{B}_j
\end{bmatrix},
\qquad
\mathbf{B}_1=
\begin{bmatrix}
1 & 1\\
0 & 1
\end{bmatrix}.
\end{equation}
Furthermore, $\mathbf{B}_j$ is invertible over $\mathbb{R}$, and hence $\mathbf{J}^{(j)}$ can be uniquely recovered from $\mathbf{g}^{(j)}$ via $\mathbf{J}^{(j)}=\mathbf{B}_j^{-1}\mathbf{g}^{(j)}$.
\end{proposition}
\begin{proof}
For each coordinate $k$, define its state as $\underline{S}(k)=[u_k(1),u_k(2),\ldots,u_k(j)] \in \{0,1\}^j$. The value $J(\underline{s})$ counts how many coordinates have state $\underline{s}$.

The key observation is that every XOR combination is determined state by state. For $\underline{b}\neq\underline{0}$, by \eqref{eq:def_g},
\begin{align}
g(\underline{b})
&=
w_{\mathrm{H}}\!\left(\bigoplus_{i=1}^{j} b_i\,\underline{u}(i)\right)
=
\sum_{k=1}^{N}\left(\bigoplus_{i=1}^{j}b_i u_k(i)\right) \notag\\
&=
\sum_{k=1}^{N}\bigl(\underline{b}\cdot\underline{S}(k)\bigr)
=
\sum_{\underline{s}\in\{0,1\}^j}J(\underline{s})(\underline{b}\cdot\underline{s}).
\label{eq:Bj_pf_state_sum}
\end{align}
For $\underline{b}=\underline{0}$, we have $g(\underline{0})=N=\sum_{\underline{s}}J(\underline{s})$. These two cases are exactly captured by the entries of $\mathbf{B}_j$ in \eqref{eq:Bj_entry_new}. Arranging the resulting equations for all $\underline{b}\in\{0,1\}^j$ in lexicographic order gives $\mathbf{g}^{(j)}=\mathbf{B}_j\mathbf{J}^{(j)}$. The initial matrix $\mathbf{B}_1=
\begin{bmatrix}
1 & 1\\
0 & 1
\end{bmatrix}$ follows directly from the definition.

The recursion in \eqref{eq:Bj_recursion_new} and the invertibility of $\mathbf{B}_j$ are established in Appendix~\ref{app:Bj_recursion}.
\end{proof}

Now we are ready to complete the proof of Theorem~\ref{th:linear_pre_target_hit}.

It remains to combine the above ingredients with the inclusion--exclusion expansion \eqref{eq:pre_target_hit_IEP}. For any fixed indices $1\le i_1<\cdots<i_j<t$, Lemma~\ref{lem:linear_pre_target_hit} shows that $\Pr\!\left(\bigcap_{k=1}^{j}A_{i_k t}\right)$ can be evaluated by summing over all ordered codeword tuples in $\mathcal{C}^j$ having the same weight relationship as $[\underline{e}(i_1,t),\ldots,\underline{e}(i_j,t)]$.

By Proposition~\ref{prop:rearrange_number}, for every ordered codeword tuple appearing in this summation, the number of admissible coordinate permutations depends only on the common weight relationship $\mathbf{g}=\mathcal{F}(\underline{e}(i_1,t),\ldots,\underline{e}(i_j,t))$ and is given by $\mathcal{H}(\mathbf{g})$ in \eqref{eq:H_by_J}. Hence,
\begin{equation}
\Pr\!\left(
\bigcap_{k=1}^{j}\{\pi(\underline{e}(i_k,t))=\underline{u}(k)\}
\right)
=
\frac{1}{N!}\mathcal{H}(\mathbf{g}).
\end{equation}

Hence, every ordered codeword tuple with weight relationship $\mathbf{g}$ contributes the same probability, and summing over all such tuples yields
\begin{equation}
\Pr\!\left(\bigcap_{k=1}^{j}A_{i_k t}\right)
=
\frac{1}{N!}Z(\mathbf{g})\mathcal{H}(\mathbf{g}).
\end{equation}

Substituting this expression into \eqref{eq:pre_target_hit_IEP} gives \eqref{eq:th_line_mif_main}, which completes the proof of Theorem~\ref{th:linear_pre_target_hit}.

Theorem~\ref{th:linear_pre_target_hit} separates the computation of $f(\mathcal{C},\mathcal{E},t)$ into two parts. The factor $\mathcal{H}(\mathbf{g})$ is purely combinatorial and captures the symmetry induced by coordinate permutations. The term $Z(\mathbf{g})$ is code-dependent and describes how often a given higher-order weight relationship appears in the codebook.

This separation is useful at two levels. The first-order term leads to a simple upper bound that depends only on the ordinary weight distribution. Higher-order terms refine this bound by incorporating joint codeword relationships. We spell out both points below.

\begin{corollary}[Single-event probability]\label{coro:single_event_prob}
For any linear block code $\mathcal{C}$, let
\begin{equation}
\mathbf{g}_{it}^{(1)}
\triangleq
\bigl(N,w_{\mathrm{H}}(\underline{e}(i,t))\bigr).
\end{equation}
Then the probability of the event $A_{it}=\{\pi(\underline{e}(i,t))\in\mathcal{C}\}$ is given by
\begin{equation}\label{eq:single_event_prob}
\Pr(A_{it}) = \frac{Z(\mathbf{g}_{it}^{(1)})}{C_{N}^{w_{\mathrm{H}}(\underline{e}(i,t))}}.
\end{equation}
Consequently, the first-order BLER upper bound in \eqref{eq:first_order_bler_bound} can be written as
\begin{equation}\label{eq:first_order_bler_bound_weight}
P_{\mathrm{err}}
\le
1-\sum_{t=1}^{T}p_t + 
\sum_{t=1}^{T}p_t\sum_{i=1}^{t-1}
\frac{Z(\mathbf{g}_{it}^{(1)})}{C_{N}^{w_{\mathrm{H}}(\underline{e}(i,t))}}.
\end{equation}
\end{corollary}

\begin{proof}
Applying Theorem~\ref{th:linear_pre_target_hit} with $j=1$, we have
\begin{align}
\Pr(A_{it}) = \frac{1}{N!}Z(\mathbf{g}_{it}^{(1)})\mathcal{H}(\mathbf{g}_{it}^{(1)}).
\end{align}
To compute $\mathcal{H}(\mathbf{g}_{it}^{(1)})$, note that for $j=1$,
\begin{equation}
\left[\begin{array}{c}
J(0) \\ J(1)
\end{array}\right]
=
\left[\begin{array}{cc}
1 & 1 \\
0 & 1
\end{array}\right]^{-1}
\left[\begin{array}{c}
N \\ w_{\mathrm{H}}(\underline{e}(i,t))
\end{array}\right]
=
\left[\begin{array}{c}
N-w_{\mathrm{H}}(\underline{e}(i,t)) \\ w_{\mathrm{H}}(\underline{e}(i,t))
\end{array}\right].
\end{equation}
Hence $\mathcal{H}(\mathbf{g}_{it}^{(1)})=\bigl(N-w_{\mathrm{H}}(\underline{e}(i,t))\bigr)!w_{\mathrm{H}}(\underline{e}(i,t))!$, which gives
\begin{equation}
\Pr(A_{it}) = \frac{Z(\mathbf{g}_{it}^{(1)})}{C_{N}^{w_{\mathrm{H}}(\underline{e}(i,t))}}.
\end{equation}
Substituting this expression into \eqref{eq:first_order_bler_bound} gives \eqref{eq:first_order_bler_bound_weight}.
\end{proof}

\begin{example}\label{ex:hamming_first_order}
Consider the Hamming$(7,4)$ code. For the first-order relationship associated with an EP difference,
$\mathbf{g}^{(1)}=(7,w_{\mathrm{H}}(\underline{e}(i,t)))$. The nonzero multiplicities are
\begin{equation}
Z((7,0))=1,\qquad Z((7,3))=7,\qquad Z((7,4))=7,\qquad Z((7,7))=1,
\end{equation}
with zero multiplicity for all other Hamming weights. For this first-order relationship,
\begin{equation}
\mathcal{H}(\mathbf{g}^{(1)})=
\bigl(7-w_{\mathrm{H}}(\underline{e}(i,t))\bigr)!w_{\mathrm{H}}(\underline{e}(i,t))!,
\qquad
\begin{array}{c|cccc}
w_{\mathrm{H}}(\underline{e}(i,t)) & 0 & 3 & 4 & 7\\ \hline
Z(\mathbf{g}^{(1)}) & 1 & 7 & 7 & 1\\
\mathcal{H}(\mathbf{g}^{(1)}) & 7! & 4!3! & 3!4! & 7!\\
\Pr(A_{it}) & 1 & 1/5 & 1/5 & 1
\end{array}
\end{equation}
For all other Hamming weights, $\Pr(A_{it})=0$. For example, the EPs $1100000$ and $0010000$ differ by $1110000$, whose weight is $3$. If one of them is the target EP, testing the other one first produces a competing codeword with probability $Z((7,3))/C_7^3=7/C_7^3=1/5$. The zero-weight case is listed only for completeness; it does not arise for distinct EPs.

\end{example}

\begin{example}\label{ex:1}
To illustrate the evaluation of a higher-order term, consider the triple-intersection probability $\Pr(A_{1t}\cap A_{2t}\cap A_{3t})$. This example shows how the EP-difference tuple determines $\mathbf{g}$ and hence the combinatorial factor $\mathcal{H}(\mathbf{g})$; the remaining factor $Z(\mathbf{g})$ is then supplied by the codebook. The three EP differences $\underline{e}(1,t)$, $\underline{e}(2,t)$, and $\underline{e}(3,t)$ are shown in Fig.~\ref{fig:visio3}. In the figure, gray denotes $1$ and white denotes $0$; for example, $\underline{e}(1,t)=[110111000110000]$.
\begin{figure}[ht]
    \centering
    \includegraphics[width = 0.72\textwidth]{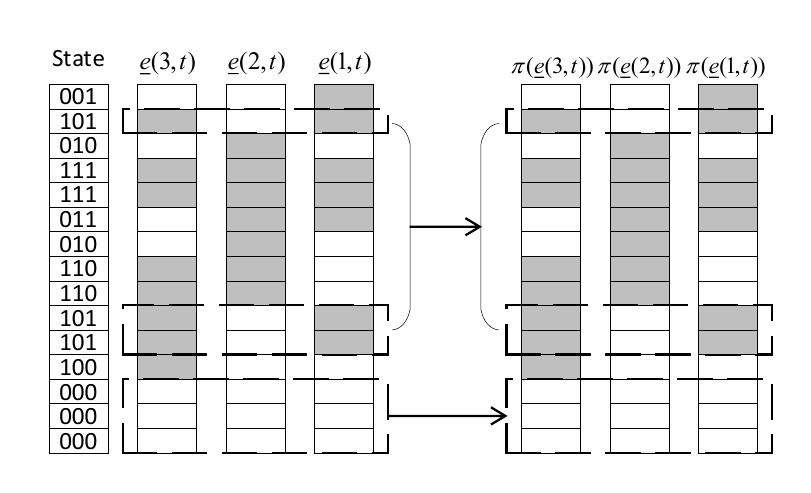}
    \caption{Schematic diagram of the three EP differences in Example~\ref{ex:1}.}
    \label{fig:visio3}
\end{figure}
\end{example}

\noindent\textit{Evaluation.}
From the figure, we obtain $\mathbf{g}=(15,7,7,8,8,5,7,6)$. Applying Theorem~\ref{th:linear_pre_target_hit}, the corresponding joint-event probability is
\begin{equation}
\frac{1}{15!}\,Z(\mathbf{g})\,\mathcal{H}(\mathbf{g}).
\end{equation}

To compute $\mathcal{H}(\mathbf{g})$, we solve $\mathbf{J}^{(3)}=\mathbf{B}_3^{-1}\mathbf{g}^{(3)}$. Using
\begin{equation}
\mathbf{B}_3=
\left[\begin{array}{llll:llll}
1 & 1 & 1 & 1 & 1 & 1 & 1 & 1 \\
0 & 1 & 0 & 1 & 0 & 1 & 0 & 1 \\
0 & 0 & 1 & 1 & 0 & 0 & 1 & 1 \\
0 & 1 & 1 & 0 & 0 & 1 & 1 & 0 \\
\hdashline
0 & 0 & 0 & 0 & 1 & 1 & 1 & 1 \\
0 & 1 & 0 & 1 & 1 & 0 & 1 & 0 \\
0 & 0 & 1 & 1 & 1 & 1 & 0 & 0 \\
0 & 1 & 1 & 0 & 1 & 0 & 0 & 1
\end{array}\right],\qquad
\mathbf{B}_3^{-1}
=
\frac{1}{4}
\left[\begin{array}{rrrrrrrr}
4 & -1 & -1 & -1 & -1 & -1 & -1 & -1 \\
0 & 1 & -1 & 1 & -1 & 1 & -1 & 1 \\
0 & -1 & 1 & 1 & -1 & -1 & 1 & 1 \\
0 & 1 & 1 & -1 & -1 & 1 & 1 & -1 \\
0 & -1 & -1 & -1 & 1 & 1 & 1 & 1 \\
0 & 1 & -1 & 1 & 1 & -1 & 1 & -1 \\
0 & -1 & 1 & 1 & 1 & 1 & -1 & -1 \\
0 & 1 & 1 & -1 & 1 & -1 & -1 & 1
\end{array}\right].
\end{equation}
Hence,
\begin{equation}
[J(000),J(100),\ldots,J(111)]^T=\mathbf{B}_3^{-1}\mathbf{g}^T=[3,1,2,1,1,3,2,2]^T,
\end{equation}
and therefore
\begin{equation}
\mathcal{H}(\mathbf{g})=\prod_{l=0}^{7}J(l)!=(3!)(1!)(2!)(1!)(1!)(3!)(2!)(2!)=288.
\end{equation}

Thus,
\begin{equation}
\Pr(A_{1t}\cap A_{2t}\cap A_{3t})
=
\frac{288}{15!}\,Z(\mathbf{g}),
\end{equation}
where $Z(\mathbf{g})$ is determined by the codebook. This example illustrates that the joint-event probability decomposes into a combinatorial factor $\mathcal{H}(\cdot)$ and a structural term $Z(\cdot)$.

\section{Numerical Experiments}\label{Sec:simulation}

In this section, we provide numerical results to validate the proposed analytical framework and evaluate the performance of RS-ORBGRAND. We first examine how many Monte Carlo trials are used and how the resulting estimates behave statistically, then verify the derived BLER expression on representative codes, and finally compare RS-ORBGRAND with existing decoding methods in terms of accuracy and complexity.

\subsection{Accuracy of Monte Carlo Estimation}
\label{SubSec:Err_analysis}

In this subsection, we quantify the statistical accuracy of the simulation results used later in this section. Both the conventional BLER from decoding simulation and the BLER computed from the analytical expression in Section~\ref{Sec:linear_code} rely on Monte Carlo estimation, so we derive CLT-based sufficient conditions on the numbers of independent Monte Carlo trials. We let $\mathcal{L}_{\text{dec}}$ denote the number of independent decoding simulations used to estimate the BLER, and let $\mathcal{L}_{\text{agp}}$ denote, in each run of the analytical method, the number of channel uses used to obtain Monte Carlo estimates of the AGPs~$\{p_t\}$.

We first consider the BLER under conventional decoding. Based on the Markov chain $\underline{W}\rightarrow \underline{Y}\rightarrow \hat{\underline{W}}$, we write $P_{\text{err}} = \Pr(\underline{W} \neq \hat{\underline{W}})$ and estimate it from $\mathcal{L}_{\text{dec}}$ independent Monte Carlo trials. The natural estimator is
\begin{equation}\label{eq:decoding_bler_estimator}
\hat{P}_{\text{err}} = \frac{1}{\mathcal{L}_{\text{dec}}} \sum_{i=1}^{\mathcal{L}_{\text{dec}}} \mathbf{1}(\underline{w} \neq \hat{\underline{w}}),
\end{equation}
which is unbiased, with $\operatorname{Var}(\hat{P}_{\text{err}}) = P_{\text{err}}(1 - P_{\text{err}})/\mathcal{L}_{\text{dec}}$. By the central limit theorem (CLT), for sufficiently large $\mathcal{L}_{\text{dec}}$ we use the approximation $\hat{P}_{\text{err}} \sim \mathcal{N}\!\left(P_{\text{err}}, P_{\text{err}}(1 - P_{\text{err}})/\mathcal{L}_{\text{dec}}\right)$, so that for a target relative deviation~$\Delta$ and confidence level~$\alpha$ a sufficient condition on $\mathcal{L}_{\text{dec}}$ is
\begin{equation}\label{eq:decoding_clt_condition}
2\Phi\!\left(\sqrt{\frac{\mathcal{L}_{\text{dec}}}{P_{\text{err}}(1 - P_{\text{err}})}} \, \Delta P_{\text{err}} \right) - 1 \geq \alpha.
\end{equation}

We next consider the BLER when it is obtained from the analysis of Section~\ref{Sec:linear_code}, i.e., from
\begin{equation}\label{eq:linear_bler_decomp}
P_{\text{err}} = 1 - \sum_{t=1}^{T} p_t\bigl(1 - f(\mathcal{C}, \mathcal{E}, t)\bigr).
\end{equation}
The same $1 - \sum_{t=1}^{T} p_t\bigl(1 - f(\mathcal{C}, \mathcal{E}, t)\bigr)$ structure applies in both the random-ensemble and fixed-code settings. For a random code ensemble, $f(\mathcal{C}, \mathcal{E}, t)$ is available in closed form by \eqref{eq:random_function}. For a fixed linear block code, the exact $f(\mathcal{C}, \mathcal{E}, t)$ is given by the full inclusion--exclusion expression in Theorem~\ref{th:linear_pre_target_hit}; practical computation may instead use a finite-order truncation and, when necessary, an approximation to the code-dependent counts $Z(\mathbf{g})$. For brevity, write $f(t)=f(\mathcal{C}, \mathcal{E}, t)$ in the remainder of this subsection.

We first isolate the sampling error caused by estimating the AGPs. To this end, treat the evaluated values $\{f(t)\}_{t=1}^{T}$ as fixed, so that the only randomness in the analytical BLER estimate comes from replacing $\{p_t\}$ by Monte Carlo estimates from an $\mathcal{L}_{\text{agp}}$-trial run. The reported BLER is then
\begin{equation}\label{eq:agp_bler_estimator}
\hat{P}_{\text{err}} = 1 - \sum_{t=1}^{T} \hat{p}_t\bigl(1 - f(t)\bigr).
\end{equation}
By the multivariate CLT, $(\hat{p}_1,\ldots,\hat{p}_T)$ is asymptotically jointly normal as $\mathcal{L}_{\text{agp}}\to\infty$. The AGP estimators are unbiased for $(p_1,\ldots,p_T)$, and \eqref{eq:linear_bler_decomp} is affine in $(p_1,\ldots,p_T)$, so the corresponding $\hat{P}_{\text{err}}$ in \eqref{eq:linear_bler_decomp} is unbiased for $P_{\text{err}}$ for the same $\{f(t)\}$ and is asymptotically normal. In the expressions below, $\operatorname{Var}(\hat{P}_{\text{err}})$ denotes the asymptotic variance of $\sqrt{\mathcal{L}_{\text{agp}}}\bigl(\hat{P}_{\text{err}}-P_{\text{err}}\bigr)$; in practice it is estimated from the sample variance of $\{\hat{P}_{\text{err}}^{(r)}\}$ over independent repetition indices~$r$, each obtained from a separate $\mathcal{L}_{\text{agp}}$-trial run. For large~$\mathcal{L}_{\text{agp}}$, we use $\hat{P}_{\text{err}} \sim \mathcal{N}\!\left(P_{\text{err}}, \operatorname{Var}(\hat{P}_{\text{err}})/\mathcal{L}_{\text{agp}}\right)$, which yields the following sample-size condition
\begin{equation}\label{eq:agp_clt_condition}
2\Phi\!\left(\sqrt{\frac{\mathcal{L}_{\text{agp}} P_{\text{err}}^2}{\operatorname{Var}(\hat{P}_{\text{err}})}} \, \Delta \right) - 1 \geq \alpha.
\end{equation}

In the same spirit, to target $\Delta = 1\%$ and $\alpha = 99\%$, a convenient sufficient condition for the number of decoding simulations~$\mathcal{L}_{\text{dec}}$ is
\begin{equation}\label{eq:decoding_sample_rule}
\mathcal{L}_{\text{dec}} \gtrsim 9 \times 10^4 \cdot \frac{1 - P_{\text{err}}}{P_{\text{err}}}.
\end{equation}
For the analytical evaluation that uses Monte Carlo estimates of~$\{p_t\}$, the quantities $P_{\text{err}}$ and $\operatorname{Var}(\hat{P}_{\text{err}})$ in the display below are taken for the same $\{f(t)\}$ as in \eqref{eq:linear_bler_decomp}, and a corresponding condition for~$\mathcal{L}_{\text{agp}}$ is
\begin{equation}\label{eq:agp_sample_rule}
\mathcal{L}_{\text{agp}} \gtrsim 9 \times 10^4 \cdot \frac{\operatorname{Var}(\hat{P}_{\text{err}})}{P_{\text{err}}^2}.
\end{equation}

\begin{remark}
For a fixed code, a finite-order evaluation of the inclusion--exclusion expansion can make the computed $f(\mathcal{C}, \mathcal{E}, t)$ differ from the exact pre-target codeword-hit probability. If some $Z(\mathbf{g})$ values are further replaced by a structural model, this introduces an additional model error. The CLT calculation above isolates only the Monte Carlo error in~$\{p_t\}$ and does not account for these deterministic errors. In the experiments below, the expansion order and, when used, the approximation of $Z(\mathbf{g})$ are chosen so that their effect is small at the target precision, while~$\mathcal{L}_{\text{agp}}$ is set large enough to keep the AGP sampling error small.
\end{remark}

\subsection{Analytical and Simulated BLER}\label{SubSec:Simulation1}

We compare decoding simulation with the BLER from the analytical expression for representative linear block codes, in line with Subsection~\ref{SubSec:Err_analysis} and the preceding remark: the deviation from decoding includes Monte Carlo error in the estimates of the AGPs, while a finite expansion order in~$f$ or an approximation to $Z(\mathbf{g})$ may introduce an additional bias. We begin with the Hamming$(7,4)$ code, with ORBGRAND and at most $T = 10$ tests.

The evaluation follows directly from the theoretical results in Section~\ref{Sec:linear_code}. Specifically, we compute the AGP $\{p_t\}$ using Proposition~\ref{prop:agp_pi}, evaluate the pre-target codeword-hit function $f(\mathcal{C}, \mathcal{E}, t)$ via Theorem~\ref{th:linear_pre_target_hit}, and obtain the BLER from Theorem~\ref{th:linear_error_rate}. For the short Hamming code, the required codeword-tuple counts are evaluated exactly; the first-order single-event probabilities are illustrated in Example~\ref{ex:hamming_first_order}. For the longer BCH code considered next, the first-order term uses the exact weight distribution, while the second-order term uses the pairwise approximation described below.

\begin{figure}[ht]
    \centering
    \includegraphics[width = 0.72\textwidth]{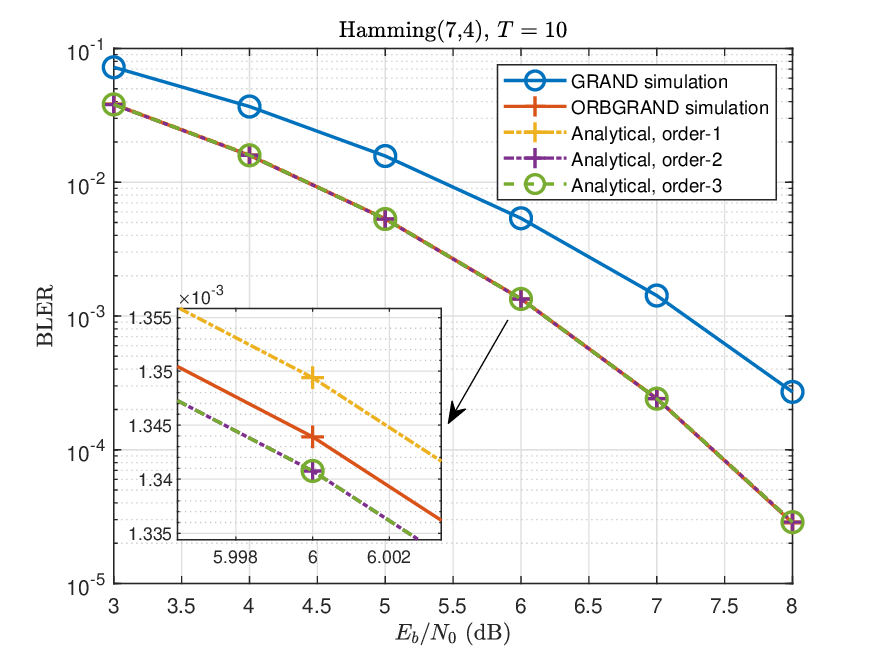}
    \caption{Hamming$(7,4)$: analytical BLER evaluation.}
    \label{fig:Hamming7_4}
\end{figure}

Fig.~\ref{fig:Hamming7_4} compares the BLER obtained from decoding simulation and from the analytical expression. To compute $f(\mathcal{C}, \mathcal{E}, t)$, we apply the inclusion--exclusion principle and truncate the expansion after a finite number of terms. The curves labeled by different orders correspond to retaining the first one, two, and three terms in \eqref{eq:th_line_mif_main}, respectively; in particular, the first-order curve is the union-bound-based BLER upper bound in \eqref{eq:first_order_bler_bound}.

It is observed that retaining two terms already yields the same result as retaining three terms. This indicates that, under the given EP set $\mathcal{E}$ and testing budget $T$, the probability of having three or more pre-target codeword hits is zero. Consequently, the truncated expression of $f(\mathcal{C}, \mathcal{E}, t)$ is exact in this case.

The small gap is therefore due mainly to sampling error in~$\{\hat{p}_t\}$ (cf.\ Subsection~\ref{SubSec:Err_analysis}), while $f$ is effectively free of truncation error. The case confirms that the analysis tracks decoding when the $f$ expansion is exact at the order retained.



We next consider a longer Bose--Chaudhuri--Hocquenghem (BCH) code, namely BCH$(127,113)$, and evaluate both the accuracy and computational complexity of the finite-order analytical approximation with $T = 10^4$.

Fig.~\ref{fig:BCH_calcu} compares the BLER obtained from decoding simulation and from different analytical approximations. The order-$k$ curves correspond to retaining the first $k$ terms in the inclusion--exclusion expansion \eqref{eq:th_line_mif_main}. In particular, order-0 ignores the pre-target codeword-hit effect, order-1 gives the first-order upper bound in \eqref{eq:first_order_bler_bound}, and order-2 further incorporates pairwise interactions using the approximation to $Z(\mathbf{g}^{(2)})$ described below. This matches the discussion in Subsection~\ref{SubSec:Err_analysis} and the remark: a small~$k$ inflates the approximation error in~$f$, while order-2 together with a sufficiently large~$\mathcal{L}_{\text{agp}}$ brings the analytical points close to the decoding curve (see Fig.~\ref{fig:BCH_calcu}).

\begin{figure}[ht]
    \centering
    \includegraphics[width = 0.72\textwidth]{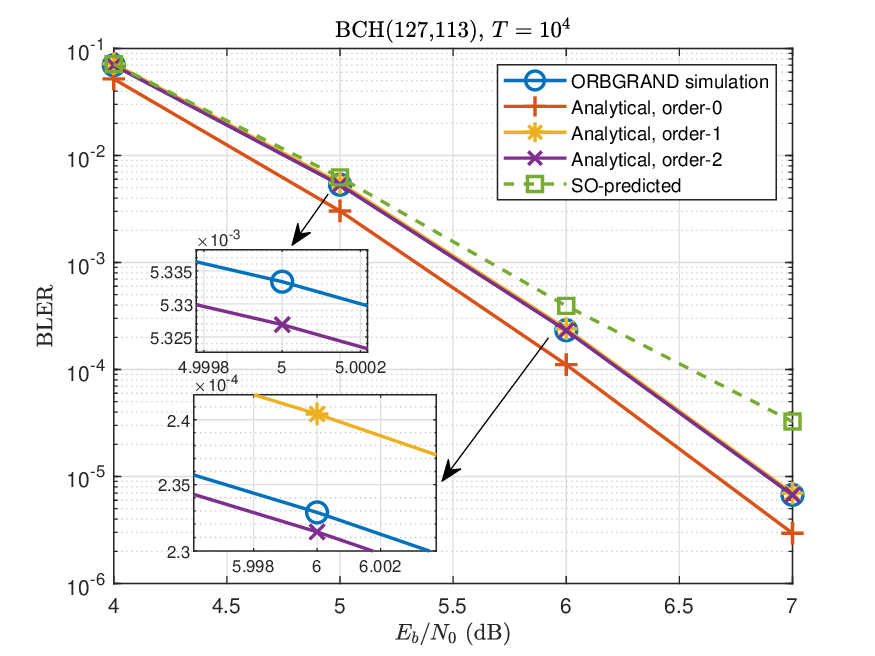}
    \caption{BCH$(127,113)$: finite-order analytical BLER approximation.}
    \label{fig:BCH_calcu}
\end{figure}

It is observed that the first-order bound is already reasonably close, although still conservative (around $3\%$ at $6$~dB), while the second-order approximation achieves high accuracy, with a deviation below $0.5\%$. Thus, the first-order term provides a simple and useful BLER upper bound, and adding the pairwise term substantially tightens the analytical prediction. For comparison, the method in~\cite{duffy2024soft} (referred to as ``SO-predicted'') assumes a geometric distribution for competing-codeword hits, which leads to noticeable inaccuracies at high SNR.

To further illustrate this behavior, Fig.~\ref{fig:dis_list} shows the empirical distribution of the number of pre-target codeword hits, obtained from decoding simulation. It is observed that the probability decays rapidly as the number of pre-target codeword hits increases, and events involving more than two such hits are extremely rare. This empirical observation supports the effectiveness of the second-order approximation, as higher-order terms contribute negligibly in practice.

\begin{figure}[ht]
    \centering
    \includegraphics[width = 0.72\textwidth]{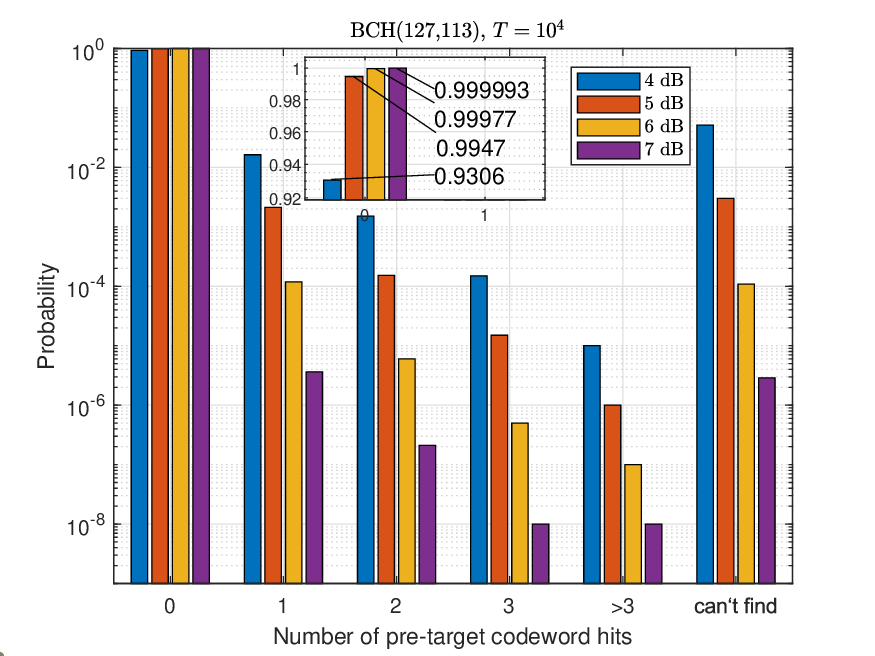}
    \caption{BCH$(127,113)$: pre-target codeword-hit count ($T=10^4$).}
    \label{fig:dis_list}
\end{figure}

We briefly describe how the analytical results are computed in practice. By Corollary~\ref{coro:single_event_prob}, the first-order term of $f(\mathcal{C}, \mathcal{E}, t)$ depends only on the weight distribution of the code, which can be efficiently obtained from the dual code via the MacWilliams identity~\cite{macwilliams1963theorem}:
\begin{equation}\label{eq:macwilliams_identity}
A(z)=2^{-(N-K)}(1+z)^N B\!\left(\frac{1-z}{1+z}\right).
\end{equation}

For the second-order term, the exact fixed-code expression requires $Z(\mathbf{g}^{(2)})$, which involves the joint weight distribution of codeword pairs, where
\begin{equation}
\mathbf{g}^{(2)}=(N,g(01),g(10),g(11)).
\end{equation}
Since such pairwise weight enumerators are typically unavailable for structured codes, the BCH evaluation adopts a tractable uniform support model: conditioned on their individual weights, the supports of two codewords are treated as uniformly distributed subsets of $\{1,\ldots,N\}$. Under this model, the overlap profile---and hence $Z(\cdot)$---is determined by simple combinatorial counting, leading to the following approximation:
\begin{equation}\label{eq:second_order_Z_approx}
Z(\mathbf{g}^{(2)}) \approx Z(\mathbf{g}_{01}^{(1)}) Z(\mathbf{g}_{10}^{(1)}) \cdot \frac{C_{g(01)}^{(g(01)+g(10)-g(11))/2} C_{N-g(01)}^{(g(10)+g(11)-g(01))/2}}{C_{N}^{g(10)}}.
\end{equation}
Here, $\mathbf{g}_{01}^{(1)}=(N,g(01))$ and $\mathbf{g}_{10}^{(1)}=(N,g(10))$.

Table~\ref{Tab:simu_num_BCH} lists the numbers of decoding runs and AGP-trial runs (values of~$\mathcal{L}_{\text{dec}}$ and~$\mathcal{L}_{\text{agp}}$) obtained from the bounds in Subsection~\ref{SubSec:Err_analysis} for $\Delta = 1\%$ and $\alpha = 99\%$, for second-order~$f$ in the same analytical setting as in Fig.~\ref{fig:BCH_calcu}.

The row for $\operatorname{Var}(\hat{P}_{\text{err}})$ uses the sample variance in Subsection~\ref{SubSec:Err_analysis} for the path that estimates~$\{p_t\}$ by Monte Carlo, not the Bernoulli variance $P_{\text{err}}(1-P_{\text{err}})$ from decoding simulation; the difference between the two explains why the sufficient condition on~$\mathcal{L}_{\text{dec}}$ in the last row is vastly larger than that on~$\mathcal{L}_{\text{agp}}$ at low~$P_{\text{err}}$.

As the SNR increases, the lower bound on~$\mathcal{L}_{\text{dec}}$ from the same sufficient conditions grows very quickly, while the corresponding lower bound on~$\mathcal{L}_{\text{agp}}$ changes only mildly across SNR, consistent with the variance scalings in Subsection~\ref{SubSec:Err_analysis}.


\begin{table}[ht]
    \renewcommand{\arraystretch}{1.2}
    \caption{BCH$(127,113)$: required sample sizes for low-BLER evaluation.}
    \label{Tab:simu_num_BCH}
    \centering
    \setlength{\tabcolsep}{5pt}
    \begin{tabular}{lccccc}
    \toprule
    & \textbf{$4$~dB} & \textbf{$5$~dB} & \textbf{$6$~dB} & \textbf{$7$~dB} & \textbf{$8$~dB} \\
    \midrule
    \multicolumn{6}{c}{\textit{Error Statistics}} \\
    \midrule
    $\hat{P}_{\text{err}}$  & $6.9\!\times\!10^{-2}$ & $5.3\!\times\!10^{-3}$ & $2.3\!\times\!10^{-4}$ & $6.7\!\times\!10^{-6}$ & $1.0\!\times\!10^{-7}$ \\
    $\operatorname{Var}(\hat{P}_{\text{err}})$   & $3.5\!\times\!10^{-3}$ & $4.2\!\times\!10^{-6}$ & $5.0\!\times\!10^{-8}$ & $3.2\!\times\!10^{-11}$ & $8.0\!\times\!10^{-15}$ \\
    \midrule
    \multicolumn{6}{c}{\textit{Required Samples}} \\
    \midrule
    Proposed $\mathcal{L}_{\text{agp}}$  & $6.7\!\times\!10^{4}$ & $1.4\!\times\!10^{5}$ & $8.5\!\times\!10^{4}$ & $6.5\!\times\!10^{4}$ & $7.2\!\times\!10^{4}$ \\
    Decoding simulation $\mathcal{L}_{\text{dec}}$        & $1.3\!\times\!10^{5}$ & $1.7\!\times\!10^{7}$ & $3.9\!\times\!10^{8}$ & $1.4\!\times\!10^{10}$ & $8.9\!\times\!10^{11}$ \\
    \bottomrule
    \end{tabular}
\end{table}

\subsection{Performance and Near-Optimality of RS-ORBGRAND}\label{SubSec:RSORB}

In this subsection, we evaluate the performance of the proposed RS-ORBGRAND through numerical simulations. We consider two representative codes: BCH$(127,113)$ and cyclic-redundancy-check (CRC)-aided polar$(128,114)$.

The EP list $\mathcal{E}_{RS}$ is constructed following Algorithm~\ref{alg:E_RS}: we first generate an initial candidate list using CDF-ORBGRAND~\cite{duffy2022ORBGRAND, liu2022orbgrand} and then reshuffle it by sorting the estimated AGPs $\{p_t\}$ in descending order. Unless otherwise specified, the maximum number of tests is set to $T=10^4$.

\begin{figure}[ht]
    \centering
    \includegraphics[width = 0.72\textwidth]{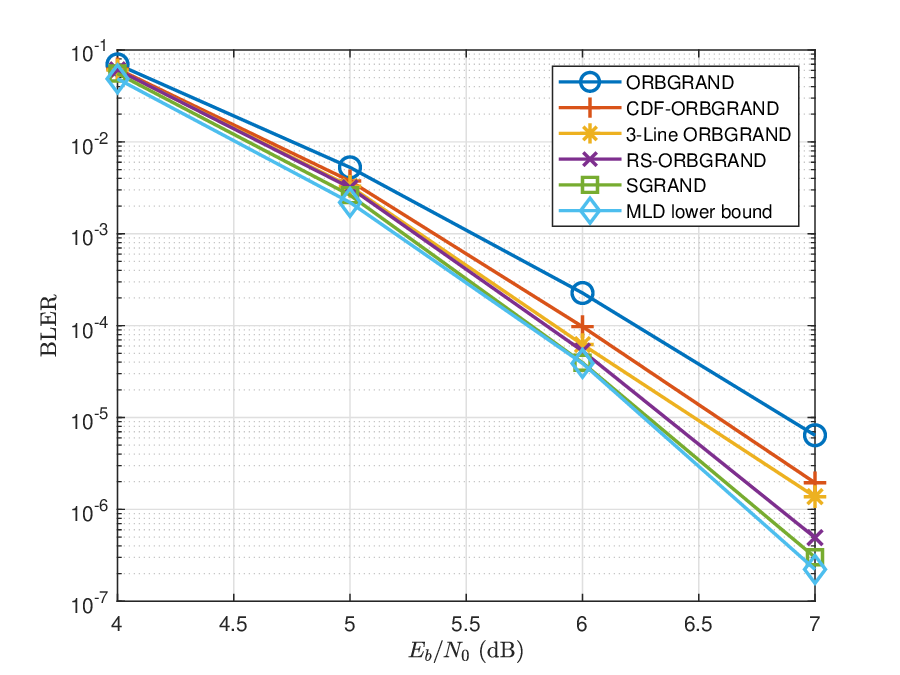}
    \caption{BCH$(127,113)$: RS-ORBGRAND performance.}
    \label{fig:BCH_perfomance}
\end{figure}

Fig.~\ref{fig:BCH_perfomance} compares the BLER of RS-ORBGRAND with several ORB-type decoders, as well as SGRAND and a lower bound on the MLD BLER. Since this lower-bound curve is no higher than the true MLD BLER curve, its SNR gap provides a conservative reference for near-MLD performance. It is observed that RS-ORBGRAND approaches SGRAND and remains within $0.1$~dB of this lower-bound curve at high SNR. Compared with existing ORB-type decoders, RS-ORBGRAND provides a consistent performance gain. For example, at a BLER of $10^{-5}$, the SNR gain relative to ORBGRAND is about $0.5$~dB.

Table~\ref{Tab:2} further reports the average number of tests. Among ORB-type GRAND decoders, RS-ORBGRAND requires the smallest number of queries. Although SGRAND achieves the lowest average number of tests, it generates its EP queries adaptively from the received LLR vector for each decoding instance, which results in significantly higher computational complexity.

\begin{table}[ht]
    \renewcommand{\arraystretch}{1.2}
    \caption{Average number of tests for BCH$(127,113)$ under different decoders.}
    \label{Tab:2}
    \centering
    \setlength{\tabcolsep}{6pt}
    \begin{tabular}{lcccc}
    \toprule
    \textbf{Decoder} & \textbf{$4$~dB} & \textbf{$5$~dB} & \textbf{$6$~dB} & \textbf{$7$~dB} \\
    \midrule
    ORBGRAND & 790.8 & 83.89 & 7.072 & 1.479 \\
    CDF-ORBGRAND & 727.9 & 67.44 & 5.476 & 1.478 \\
    3-Line ORBGRAND & 730.0 & 62.34 & 4.732 & 1.445 \\
    \textbf{RS-ORBGRAND (proposed)} & \textbf{715.6} & \textbf{60.63} & \textbf{4.445} & \textbf{1.350} \\
    \midrule
    SGRAND & 666.5 & 52.99 & 3.932 & 1.328 \\
    \bottomrule
    \end{tabular}
\end{table}

Fig.~\ref{fig:Polar_perfomance} shows the BLER performance for CRC-aided polar$(128,114)$. We further investigate the impact of the candidate set size $T_1$ used in Algorithm~\ref{alg:E_RS} to form the initial EP candidate list before reshuffling. It is observed that a sufficiently large candidate set is essential for achieving good performance. Increasing $T_1$ enables more EPs with larger AGPs to be identified and prioritized during reshuffling, thereby improving decoding performance. Notably, $T_1$ only affects the one-time construction of $\mathcal{E}_{RS}$, while the per-decoding complexity remains bounded by $T = 10^4$.

\begin{figure}[ht]
    \centering
    \includegraphics[width = 0.72\textwidth]{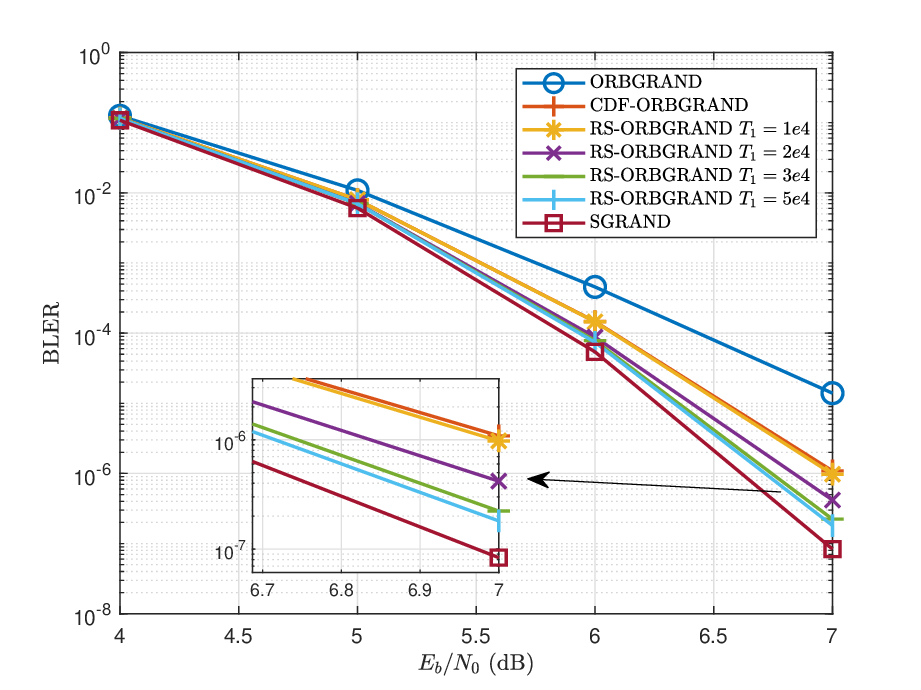}
    \caption{CRC-aided polar$(128,114)$: effect of candidate-set size.}
    \label{fig:Polar_perfomance}
\end{figure}

\section{Conclusion}
\label{Sec:conclusion}

This paper develops an AGP-based framework for analyzing and designing ORB-type GRAND algorithms. For random code ensembles, we derive exact expressions for the BLER, stopping-time distribution, and average number of tests under a fixed test budget. The analysis separates target-miss and target-preemption errors and establishes that ordering EPs by non-increasing AGP simultaneously minimizes the average BLER and the average number of tests over the EP set under consideration.

For fixed linear block codes, we extend the BLER analysis by isolating the code-dependent target-preemption term. This term is characterized through higher-order weight relationships of codeword tuples, with the ordinary weight distribution appearing as the first-order special case. Guided by these results, we formulate RS-ORBGRAND as an AGP-driven offline reshuffling scheme that preserves the ranking-based ORB-type structure. Numerical results show consistent gains over existing ORB-type GRAND decoders and near-MLD performance, while finite-order analytical evaluations based on the derived expressions enable efficient low-error-rate prediction beyond brute-force decoding simulation.

{\appendix

\subsection{Proof of Proposition~\ref{prop:Bj_recursion}}
\label{app:Bj_recursion}

\paragraph{Proof of the recursive structure.}

We derive \eqref{eq:Bj_recursion_new} from the entry-wise definition of $\mathbf{B}_{j+1}$.

\begin{proof}
    
Index the rows and columns of $\mathbf{B}_{j+1}$ by $\underline{b}=(b_1,\tilde{\underline{b}})$ and $\underline{s}=(s_1,\tilde{\underline{s}})$ in $\{0,1\}^{j+1}$, where $\tilde{\underline{b}},\tilde{\underline{s}}\in\{0,1\}^j$. Under the lexicographic order, $\mathbf{B}_{j+1}$ is partitioned into four $2^j\times 2^j$ blocks according to the leading bits $b_1$ and $s_1$.

If $b_1=0$, then $\underline{b}=\underline{0}$ holds if and only if $\tilde{\underline{b}}=\underline{0}$, and for $\underline{b}\neq \underline{0}$ we have
\begin{equation}
\underline{b}\cdot\underline{s}
=
\tilde{\underline{b}}\cdot\tilde{\underline{s}}.
\end{equation}
Therefore, both the upper-left block $(b_1,s_1)=(0,0)$ and the upper-right block $(b_1,s_1)=(0,1)$ are equal to $\mathbf{B}_j$.

Now consider $b_1=1$. Then $\underline{b}\neq \underline{0}$ always holds. If $s_1=0$, then
\begin{equation}
\underline{b}\cdot\underline{s}
=
\tilde{\underline{b}}\cdot\tilde{\underline{s}}.
\end{equation}
Hence the lower-left block agrees with $\mathbf{B}_j$ for all rows except the first row, corresponding to $\tilde{\underline{b}}=\underline{0}$. In that row, the entries of $\mathbf{B}_j$ are all $1$, whereas the entries of the lower-left block are all $0$. Therefore, the lower-left block is
\begin{equation}
\mathbf{E}^{2^j\times 2^j}_{1,1:2^j}\oplus \mathbf{B}_j.
\end{equation}

If $s_1=1$, then
\begin{equation}
\underline{b}\cdot\underline{s}
=
1\oplus \tilde{\underline{b}}\cdot\tilde{\underline{s}}.
\end{equation}
Thus, for all rows except the first row, the lower-right block is obtained from $\mathbf{B}_j$ by flipping all bits, namely by $\mathbf{1}^{2^j\times 2^j}\oplus \mathbf{B}_j$. For the first row corresponding to $\tilde{\underline{b}}=\underline{0}$, the entries should all be $1$, so an additional flip of the first row is needed. Hence the lower-right block is
\begin{equation}
\mathbf{E}^{2^j\times 2^j}_{1,1:2^j}\oplus \mathbf{1}^{2^j\times 2^j}\oplus \mathbf{B}_j.
\end{equation}

Combining the four blocks gives
\begin{equation}
\mathbf{B}_{j+1}
=
\begin{bmatrix}
\mathbf{B}_j & \mathbf{B}_j\\
\mathbf{E}_{1,1:2^j}^{2^j\times 2^j}\oplus \mathbf{B}_j &
\mathbf{E}_{1,1:2^j}^{2^j\times 2^j}\oplus \mathbf{1}^{2^j\times 2^j}\oplus \mathbf{B}_j
\end{bmatrix},
\end{equation}
which proves \eqref{eq:Bj_recursion_new}.
\end{proof}

\paragraph{Proof of invertibility.}
We prove that $\mathbf{B}_j$ is invertible over $\mathbb{R}$ for all $j$.
\begin{proof}
Consider the sign matrix $\mathbf{H}_j$ indexed by $\{0,1\}^j\times\{0,1\}^j$, with entries
\begin{equation}
\bigl[\mathbf{H}_j\bigr]_{\underline{b},\underline{s}}
=
(-1)^{\underline{b}\cdot\underline{s}},
\qquad
\underline{b},\underline{s}\in\{0,1\}^j .
\end{equation}
The rows of $\mathbf{H}_j$ are orthogonal, since
\begin{equation}
\sum_{\underline{s}\in\{0,1\}^j}
(-1)^{(\underline{b}\oplus\underline{b}')\cdot\underline{s}}
=
\begin{cases}
2^j, & \underline{b}=\underline{b}',\\
0, & \underline{b}\neq \underline{b}'.
\end{cases}
\end{equation}
Indeed, when $\underline{b}\neq \underline{b}'$, the nonzero linear form $(\underline{b}\oplus\underline{b}')\cdot\underline{s}$ takes the values $0$ and $1$ equally often as $\underline{s}$ ranges over $\{0,1\}^j$, so the corresponding signs cancel.
Thus $\mathbf{H}_j\mathbf{H}_j^{\mathsf T}=2^j\mathbf{I}$, and $\mathbf{H}_j$ is nonsingular.

We now relate the rows of $\mathbf{B}_j$ to the rows of $\mathbf{H}_j$. For the all-zero row index,
\begin{equation}
\mathbf{B}_j(\underline{0},\cdot)
=
\mathbf{H}_j(\underline{0},\cdot).
\end{equation}
For any $\underline{b}\neq\underline{0}$,
\begin{equation}
\mathbf{B}_j(\underline{b},\cdot)
=
\frac{1}{2}\left(
\mathbf{H}_j(\underline{0},\cdot)
-
\mathbf{H}_j(\underline{b},\cdot)
\right),
\end{equation}
because the right-hand side equals $0$ when $\underline{b}\cdot\underline{s}=0$ and equals $1$ when $\underline{b}\cdot\underline{s}=1$.

Suppose that a linear combination of the rows of $\mathbf{B}_j$ is zero:
\begin{equation}
a_{\underline{0}}\mathbf{B}_j(\underline{0},\cdot)
+
\sum_{\underline{b}\neq\underline{0}}
a_{\underline{b}}\mathbf{B}_j(\underline{b},\cdot)
=
\mathbf{0}.
\end{equation}
Substituting the row relations above gives
\begin{equation}
\left(
a_{\underline{0}}
+
\frac{1}{2}\sum_{\underline{b}\neq\underline{0}}a_{\underline{b}}
\right)
\mathbf{H}_j(\underline{0},\cdot)
-
\frac{1}{2}
\sum_{\underline{b}\neq\underline{0}}
a_{\underline{b}}\mathbf{H}_j(\underline{b},\cdot)
=
\mathbf{0}.
\end{equation}
Since the rows of $\mathbf{H}_j$ are linearly independent, all coefficients in this combination are zero. Hence $a_{\underline{b}}=0$ for every $\underline{b}\neq\underline{0}$, and then $a_{\underline{0}}=0$. Therefore the rows of $\mathbf{B}_j$ are linearly independent, so $\mathbf{B}_j$ is nonsingular over $\mathbb{R}$.
\end{proof}
}


 
%

\bibliographystyle{IEEEtran}


\bibliography{reference}

@article{shannon1948,
  title={A mathematical theory of communication},
  author={Shannon, Claude Elwood},
  journal={Bell Syst. Tech. J.},
  volume={27},
  number={3},
  pages={379--423},
  year={1948},
  publisher={Nokia Bell Labs}
}

@article{berlekamp1978,
  title={On the inherent intractability of certain coding problems},
  author={Berlekamp, Elwyn and McEliece, Robert and Van Tilborg, Henk},
  journal={IEEE Trans. Inf. Theory},
  volume={24},
  number={3},
  pages={384--386},
  year={1978},
  publisher={IEEE}
}

@article{you2021towards,
  title={Towards {6G} wireless communication networks: Vision, enabling technologies, and new paradigm shifts},
  author={You, Xiaohu and Wang, Cheng-Xiang and Huang, Jie and Gao, Xiqi and Zhang, Zaichen and Wang, Mao and Huang, Yongming and Zhang, Chuan and Jiang, Yanxiang and Wang, Jiaheng and others},
  journal={Sci. China Inf. Sci.},
  volume={64},
  pages={1--74},
  year={2021},
  publisher={Springer}
}

@article{winzer2018fiber,
  title={Fiber-optic transmission and networking: The previous 20 and the next 20 years},
  author={Winzer, Peter J},
  journal={Opt. Express},
  volume={26},
  number={18},
  pages={24190--24239},
  year={2018},
  publisher={Optica Publishing Group}
}

@article{mickevicius2021serdes,
  title={A survey of high-speed serializer/deserializer architectures},
  author={Mickevicius, T and others},
  journal={IEEE Access},
  volume={9},
  pages={116770--116793},
  year={2021},
  publisher={IEEE}
}

@article{shirvanimoghaddam2018short,
  title={Short block-length codes for ultra-reliable low latency communications},
  author={Shirvanimoghaddam, Mahyar and Mohammadi, Mohammad Sadegh and Abbas, Rana and Minja, Aleksandar and Yue, Chentao and Matuz, Balazs and Han, Guojun and Lin, Zihuai and Liu, Wanchun and Li, Yonghui and others},
  journal={IEEE Commun. Mag.},
  volume={57},
  number={2},
  pages={130--137},
  year={2018},
  publisher={IEEE}
}

@article{yue2023efficient,
  title={Efficient decoders for short block length codes in {6G} {URLLC}},
  author={Yue, Chentao and Miloslavskaya, Vera and Shirvanimoghaddam, Mahyar and Vucetic, Branka and Li, Yonghui},
  journal={IEEE Commun. Mag.},
  volume={61},
  number={4},
  pages={84--90},
  year={2023},
  publisher={IEEE}
}

@article{duffy2019GRAND,
  title={Capacity-achieving guessing random additive noise decoding},
  author={Duffy, Ken R and Li, Jiange and M{\'e}dard, Muriel},
  journal={IEEE Trans. Inf. Theory},
  volume={65},
  number={7},
  pages={4023--4040},
  year={2019},
  publisher={IEEE}
}

@inproceedings{duffy2018guessing,
  title={Guessing noise, not code-words},
  author={Duffy, Ken R and Li, Jiange and M{\'e}dard, Muriel},
  booktitle={Proc. IEEE Int. Symp. Inf. Theory (ISIT)},
  pages={671--675},
  year={2018}
}

@inproceedings{solomon2020SGRAND,
  title={Soft maximum likelihood decoding using {GRAND}},
  author={Solomon, Amit and Duffy, Ken R and M{\'e}dard, Muriel},
  booktitle={Proc. IEEE Int. Conf. Commun. (ICC)},
  pages={1--6},
  year={2020}
}

@article{duffy2021SRGRAND,
  title={Guessing random additive noise decoding with symbol reliability information ({SRGRAND})},
  author={Duffy, Ken R and M{\'e}dard, Muriel and An, Wei},
  journal={IEEE Trans. Commun.},
  volume={70},
  number={1},
  pages={3--18},
  year={2021},
  publisher={IEEE}
}

@inproceedings{abbas2022grand,
  title={{GRAND} for {R}ayleigh fading channels},
  author={Abbas, Syed Mohsin and Jalaleddine, Marwan and Gross, Warren J},
  booktitle={Proc. IEEE Global Commun. Conf. Workshops (GC Wkshps)},
  pages={504--509},
  year={2022}
}

@inproceedings{duffy2023using,
  title={Using channel correlation to improve decoding-{ORBGRAND-AI}},
  author={Duffy, Ken R and Grundei, Moritz and M{\'e}dard, Muriel},
  booktitle={Proc. IEEE Global Commun. Conf. (GLOBECOM)},
  pages={3585--3590},
  year={2023}
}

@inproceedings{sarieddeen2022grand,
  title={{GRAND} for fading channels using pseudo-soft information},
  author={Sarieddeen, Hadi and M{\'e}dard, Muriel and Duffy, Ken R},
  booktitle={Proc. IEEE Global Commun. Conf. (GLOBECOM)},
  pages={3502--3507},
  year={2022}
}

@inproceedings{feng2024laplacian,
  title={Laplacian-{ORBGRAND}: Decoding for impulsive noise},
  author={Feng, Jiewei and Duffy, Ken R and M{\'e}dard, Muriel},
  booktitle={Proc. IEEE Mil. Commun. Conf. (MILCOM)},
  pages={1--6},
  year={2024}
}

@inproceedings{condo2021high,
  title={High-performance low-complexity error pattern generation for {ORBGRAND} decoding},
  author={Condo, Carlo and Bioglio, Valerio and Land, Ingmar},
  booktitle={Proc. IEEE Global Commun. Conf. Workshops (GC Wkshps)},
  pages={1--6},
  year={2021}
}

@inproceedings{yuan2023guessing,
  title={Guessing random additive noise decoding with quantized soft information},
  author={Yuan, Peihong and Duffy, Ken R and Gabhart, Evan P and M{\'e}dard, Muriel},
  booktitle={Proc. IEEE Global Commun. Conf. Workshops (GC Wkshps)},
  pages={1698--1703},
  year={2023}
}

@book{Lin2004ErrorCC,
    author ={Shu Lin and Daniel J. Costello},
    title = {Error Control Coding: Fundamentals and Applications},
    publisher = {Englewood Cliffs, NJ, USA: Prentice Hall},
    year = {2004}
}

@article{wan2025parallelism,
  title={A parallelization strategy for {GRAND} with optimality guarantee by exploiting error pattern tree representation},
  author={Wan, Li and Yin, Huarui and Zhang, Wenyi},
  journal={IEEE Trans. Commun.},
  volume={74},
  pages={8517--8532},
  year={2026},
  publisher={IEEE}
}

@article{duffy2022ORBGRAND,
  title={Ordered reliability bits guessing random additive noise decoding},
  author={Duffy, Ken R and An, Wei and M{\'e}dard, Muriel},
  journal={IEEE Trans. Signal Process.},
  volume={70},
  pages={4528--4542},
  year={2022},
  publisher={IEEE}
}

@inproceedings{an2023soft,
  title={Soft decoding without soft demapping with {ORBGRAND}},
  author={An, Wei and M{\'e}dard, Muriel and Duffy, Ken R},
  booktitle={Proc. IEEE Int. Symp. Inf. Theory (ISIT)},
  pages={1080--1084},
  year={2023}
}

@inproceedings{galligan2023block,
  title={Block turbo decoding with {ORBGRAND}},
  author={Galligan, Kevin and M{\'e}dard, Muriel and Duffy, Ken R},
  booktitle={Proc. 57th Annu. Conf. Inf. Sci. Syst. (CISS)},
  pages={1--6},
  year={2023}
}

@article{abbas2022high,
  title={High-throughput and energy-efficient {VLSI} architecture for ordered reliability bits {GRAND}},
  author={Abbas, Syed Mohsin and Tonnellier, Thibaud and Ercan, Furkan and Jalaleddine, Marwan and Gross, Warren J},
  journal={IEEE Trans. Very Large Scale Integr. (VLSI) Syst.},
  volume={30},
  number={6},
  pages={681--693},
  year={2022},
  publisher={IEEE}
}

@book{abbas2023guessing,
  title={Guessing Random Additive Noise Decoding: A Hardware Perspective},
  author={Abbas, Syed Mohsin and Jalaleddine, Marwan and Gross, Warren J},
  year={2023},
  publisher={Springer Nature}
}

@article{ji2024efficient,
  title={Efficient {ORBGRAND} implementation with parallel noise sequence generation},
  author={Ji, Chao and You, Xiaohu and Zhang, Chuan and Studer, Christoph},
  journal={IEEE Trans. Very Large Scale Integr. (VLSI) Syst.},
  volume={33},
  number={2},
  pages={435--448},
  year={2024},
  publisher={IEEE}
}

@inproceedings{wang2023improved,
  title={Improved {ORB-GRAND} for {PAC} Codes},
  author={Wang, Yueh and Shi, Zhiping and Han, Ziyu and Li, Kunyang and others},
  booktitle={Proc. 8th IEEE Int. Conf. Commun., Image Signal Process. (CCISP)},
  pages={477--481},
  year={2023}
}

@article{condo2022fixed,
  title={A fixed latency {ORBGRAND} decoder architecture with {LUT}-aided error-pattern scheduling},
  author={Condo, Carlo},
  journal={IEEE Trans. Circuits Syst. I, Reg. Papers},
  volume={69},
  number={5},
  pages={2203--2211},
  year={2022},
  publisher={IEEE}
}

@inproceedings{wan2024approaching,
  title={Approaching Maximum Likelihood Decoding Performance via Reshuffling {ORBGRAND}},
  author={Wan, Li and Zhang, Wenyi},
  booktitle={Proc. IEEE Int. Symp. Inf. Theory (ISIT)},
  pages={31--36},
  year={2024}
}

@article{liu2022orbgrand,
  title={{ORBGRAND} is almost capacity-achieving},
  author={Liu, Mengxiao and Wei, Yuejun and Chen, Zhenyuan and Zhang, Wenyi},
  journal={IEEE Trans. Inf. Theory},
  volume={69},
  number={5},
  pages={2830--2840},
  year={2022},
  publisher={IEEE}
}

@inproceedings{li2024orbgrand,
  title={{ORBGRAND}: Achievable Rate for General Bit Channels and Application in {BICM}},
  author={Li, Zhuang and Zhang, Wenyi},
  booktitle={Proc. IEEE 35th Int. Symp. Pers., Indoor Mobile Radio Commun. (PIMRC)},
  pages={1--7},
  year={2024}
}

@article{li2025orbgrand,
  title={{ORBGRAND} is exactly capacity-achieving via rank companding},
  author={Li, Zhuang and Zhang, Wenyi},
  journal={arXiv preprint arXiv:2512.00347},
  year={2025}
}

@article{li2026finite,
  title={A Finite-Blocklength Analysis for {ORBGRAND}},
  author={Li, Zhuang and Zhang, Wenyi},
  journal={arXiv preprint arXiv:2603.07526},
  year={2026}
}

@article{abbas2022list,
  title={List-{GRAND}: A practical way to achieve maximum likelihood decoding},
  author={Abbas, Syed Mohsin and Jalaleddine, Marwan and Gross, Warren J},
  journal={IEEE Trans. Very Large Scale Integr. (VLSI) Syst.},
  volume={31},
  number={1},
  pages={43--54},
  year={2022},
  publisher={IEEE}
}

@inproceedings{galligan2023upgrade,
  title={Upgrade error detection to prediction with {GRAND}},
  author={Galligan, Kevin and Yuan, Peihong and M{\'e}dard, Muriel and Duffy, Ken R},
  booktitle={Proc. IEEE Global Commun. Conf. (GLOBECOM)},
  pages={1818--1823},
  year={2023}
}

@article{yuan2024softoutput,
  title={Soft-output ({SO}) {GRAND} and iterative decoding to outperform {LDPC} codes},
  author={Yuan, Peihong and M{\'e}dard, Muriel and Galligan, Kevin and Duffy, Ken R},
  journal={IEEE Trans. Wireless Commun.},
  volume={24},
  number={4},
  pages={3386--3399},
  year={2025},
  publisher={IEEE}
}

@article{duffy2024soft,
  title={Soft-output guessing codeword decoding},
  author={Duffy, Ken R and Yuan, Peihong and Griffin, Joseph and M{\'e}dard, Muriel},
  journal={IEEE Commun. Lett.},
  volume={29},
  number={2},
  pages={328--332}, 
  year={2024},
  publisher={IEEE}
}

@book{feller1968probability,
  title={An Introduction to Probability Theory and Its Applications, Volume I},
  author={Feller, William},
  edition={3},
  publisher={Wiley},
  address={New York, NY, USA},
  year={1968}
}

@article{macwilliams1963theorem,
  title={A theorem on the distribution of weights in a systematic code},
  author={MacWilliams, Jessie},
  journal={Bell Syst. Tech. J.},
  volume={42},
  number={1},
  pages={79--94},
  year={1963},
  publisher={Wiley Online Library}
}


\end{document}